\documentclass[11pt]{article}

\usepackage[singlespacing]{setspace}
\usepackage[letterpaper, margin=1in]{geometry}

\usepackage{graphics}
\usepackage{xspace,xcolor}
\usepackage{amsmath}
\usepackage{amssymb}
\usepackage{amsthm}
\usepackage{mathtools}
\usepackage{complexity}
\usepackage[hidelinks]{hyperref}
\usepackage[noadjust]{cite}
\usepackage[sort&compress,capitalise,noabbrev]{cleveref}
\usepackage[skins]{tcolorbox}
\usepackage{tabularx}
\usepackage{enumitem}

\usepackage{thm-restate}
\usepackage{orcidlink}
\usepackage{fontawesome5}

\newtheorem{alphatheorem}{Theorem}

\definecolor{authormailgray}{rgb}{0.3,0.3,0.3}
\newcommand{\mailicon}{\raisebox{0.045em}{\textcolor{authormailgray}{\fontsize{9}{12}\selectfont\faIcon[regular]{envelope}}}}

\newenvironment{bracketenumerate}
  {\begin{enumerate}[label=(\arabic*),ref=\arabic*]}
  {\end{enumerate}}

\newcommand{\upsf}[1]{\textnormal{\textsf{#1}}}
\newcommand{\problemfont}[1]{\textnormal{\textsc{#1}}}

\newcommand{\td}{\textsf{\textnormal{td}}}
\newcommand{\tw}{\textsf{\textnormal{tw}}}

\newcommand{\vtxptype}{\ensuremath{p}\textnormal{-}\upsf{type}}
\newcommand{\rt}{\upsf{rt}}
\newcommand{\orig}{\upsf{orig}}
\newcommand{\soltype}{\upsf{sol}\textnormal{-}\upsf{type}}
\newcommand{\ramsey}{\upsf{Ram}}
\DeclareMathOperator{\atp}{atp}

\newcommand{\normprop}[1]{\Cref{lem:normalization_properties}\,(\ref{#1})}
\newcommand{\multprop}[1]{\Cref{lem:multiplicity_properties}\,(\ref{#1})}
\newcommand{\patprop}[1]{\Cref{lem:pattern_properties}\,(\ref{#1})}
\newcommand{\Size}[1]{\ensuremath{\left\vert #1 \right\vert}}
\newcommand{\BigO}[1]{\ensuremath{\mathcal{O}(#1)}}
\DeclareMathOperator{\arity}{ar}
\newcommand{\clos}{\ensuremath{\text{clos}}}

\newcommand{\MSO}{\ensuremath{\ComplexityFont{MSO}}\xspace}
\newcommand{\MSOone}{\ensuremath{\ComplexityFont{MSO}_1}\xspace}
\newcommand{\MSOtwo}{\ensuremath{\ComplexityFont{MSO}_2}\xspace}
\renewcommand{\FO}{\ensuremath{\ComplexityFont{FO}}\xspace}
\newcommand{\ESO}{\ComplexityFont{ESO}}

\newcommand{\RR}{\ensuremath{\mathcal{R}}}
\newcommand{\JJ}{\ensuremath{\mathcal{J}}}

\newcommand{\snpmodelchecking}{\SNP \problemfont{Model Checking}}

\newcommand{\N}{\mathbb{N}}
\newcommand{\real}{\mathbb{R}}
\newcommand{\Part}{\mathcal{P}}
\newcommand{\tww}{\textsf{\textnormal{tww}}}
\newcommand{\rdeg}{\operatorname{rdeg}}
\newcommand{\Red}{\mathrm{Red}}
\renewcommand{\SNP}{\ensuremath{\mathsf{SNP}}\xspace}
\newcommand{\Lo}{\mathrm{L}}
\newcommand{\Su}{\mathrm{S}}
\newcommand{\Tr}{\mathrm{T}}
\newcommand{\tn}{\textsf{\textnormal{tn}}}
\newcommand{\bx}{\textsf{\textnormal{box}}}
\newcommand{\Pg}{\mathrm{P}}
\newcommand{\sn}{\textsf{\textnormal{sn}}}
\newcommand{\qn}{\textsf{\textnormal{qn}}}

\graphicspath{{./graphics/}}
\newtheorem{theorem}{Theorem}
\newtheorem{remark}{Remark}

\newtheorem{fact}{Fact}

\newtheorem{lemma}[theorem]{Lemma}
\newtheorem{definition}[theorem]{Definition}
\newtheorem{observation}[theorem]{Observation}
\newtheorem{corollary}[theorem]{Corollary}

\newtheoremstyle{claimstyle}
  {\topsep}{\topsep}
  {\normalfont}{0pt}
  {\normalfont\itshape}{.}{0.5em}
  {\textcolor{black!60}{\ensuremath{\vartriangleright}}\nobreakspace\thmname{#1}\thmnumber{ #2}\thmnote{ (#3)}}
\theoremstyle{claimstyle}
\newtheorem{claim}[theorem]{Claim}

\newenvironment{claimproof}[1][Proof of claim]
  {\begin{proof}[{\normalfont\itshape #1}]}
  {\end{proof}}

\theoremstyle{definition}
\newtheorem*{example*}{Example}
\newtheorem{algorithm}{Algorithm}

\theoremstyle{plain}

\crefname{mtheorem}{Theorem}{Main Theorems}
\crefname{claim}{Claim}{Claims}
\crefname{algorithm}{Algorithm}{Algorithms}

\newcommand{\defparproblem}[4]{\begin{tcolorbox}[
            enhanced,
            colback=white,
            colframe=black,
            boxrule = 0.5pt,
            coltitle=black,
            title=#1,
            rounded corners,
            attach boxed title to top left={yshift=-10pt, xshift=6pt},
            bottom=-0.5mm,
            boxed title style={
                    interior style={fill=white},
                    frame hidden
                }
        ]\begin{tabularx}{12.5cm}{r X p {0.5cm}}
            Input:     & #2 \\
            Parameter: & #3 \\
            Question:  & #4
        \end{tabularx}
    \end{tcolorbox}
}

\renewcommand{\phi}{\varphi}

\title{A Uniform Algorithm for Strict NP on Bounded-Treedepth Graphs}
\makeatletter
\AddToHook{cmd/@maketitle/before}{\let\titleoriginalfnsymbol\@fnsymbol
  \renewcommand{\@fnsymbol}[1]{\ifnum#1=2 \ensuremath{\ddag}\else\titleoriginalfnsymbol{#1}\fi}\renewcommand{\@makefnmark}{\hbox{\@textsuperscript{\normalfont\@thefnmark}}}}
\makeatother
\author{Thomas Depian\thanks{Algorithms and Complexity Group, TU Wien, Austria.}\,\href{mailto:tdepian@ac.tuwien.ac.at}{\mailicon}\,\orcidlink{0009-0003-7498-6271},
  Robert Ganian\footnotemark[1]\,\href{mailto:rganian@gmail.com}{\mailicon}\,\orcidlink{0000-0002-7762-8045},
  Jakob Greilhuber\thanks{CISPA Helmholtz Center for Information Security and Saarland University, Saarbrücken, Germany.}\,\href{mailto:jakob.greilhuber@cispa.de}{\mailicon}\,\orcidlink{0009-0001-8796-6400},\\
  Marlene Gr\"undel\footnotemark[1]\,\href{mailto:mgruendel@ac.tuwien.ac.at}{\mailicon}\,\orcidlink{0000-0003-3470-1326}, and Simon Wietheger\footnotemark[1]\,\href{mailto:swietheger@ac.tuwien.ac.at}{\mailicon}\,\orcidlink{0000-0002-0734-0708}}
\date{\vspace{-5ex}}

\begin{document}

\maketitle

\begin{abstract}
A classical well-quasi-ordering result guarantees the existence of non-uniform linear-time algorithms for all problems in Strict NP on relational structures of bounded treedepth; however, this provides neither a procedure for constructing these algorithms nor computable bounds on their parameter dependence. We turn this existential result into a uniform algorithmic metatheorem. Given a Strict NP sentence $\varphi$ and a relational structure $\mathcal{R}$, our algorithm decides whether $\mathcal{R}\models\varphi$ in time $f(\Size{\varphi},\td(\mathcal{R}))\cdot\Size{\mathcal{R}}$ for a computable function $f$, where the treedepth of~\RR\ is measured on the Gaifman graph. The algorithm also constructs witness relations, with the polynomial exponent depending on their arity, and provides a unified framework for settling hereditary graph problems parameterized by treedepth.
We also present several applications---among others, our result resolves open questions on the fixed-parameter tractability of computing the stack number, queue number, track number and twin-width parameterized by treedepth.
\end{abstract}

\maketitle

\thispagestyle{empty}

\section{Introduction}
\label{sec:introduction}

Robertson and Seymour's Graph Minor Theorem~\cite{RobertsonS04} states that in every infinite sequence of graphs, some earlier graph is a minor of a later one. In particular, every minor-closed graph class is characterized by a finite set of forbidden minors. Unfortunately, the same does not hold for the induced subgraph relation: there are simple examples of infinite sequences of graphs without any pair being an induced subgraph of another (e.g., the sequence of all cycles). However, a well-quasi-ordering under induced subgraphs can be guaranteed when we restrict our attention to graphs of bounded \emph{treedepth}~\cite{dingSubgraphsWellquasiordering1992,NesetrilM12}---a fundamental graph parameter with ties to the theory of graph sparsity. For a graph $G$, the treedepth measures the minimum depth among all forests $F$, called elimination forests, such that $G$ is a subgraph of the ancestor-descendant closure of $F$~\cite{NesetrilM12}.

The existence of the aforementioned well-quasi-ordering has a broad algorithmic consequence for problems in the complexity class \emph{Strict NP}, denoted \SNP~\cite{FederV98,CCbook}.
This class consists of all problems over relational structures expressible via a second-order sentence of the form
\[
    \varphi =
    \exists R_1\cdots\exists R_r\,
    \forall x_1\cdots\forall x_q\,
    \psi(R_1,\ldots,R_r,x_1,\ldots,x_q),
\]
where $\psi$ is quantifier-free and may use the relations $E_1,\ldots,E_b$ of the input structure.
The existentially quantified variables range over relations on the input domain, and the universally quantified variables range over its elements. The relations $R_1,\ldots,R_r$ may have arbitrary arity; a satisfying interpretation of these relations then constitutes a witness for the property expressed by $\varphi$.

The fact that the arity of the relations is unbounded is a crucial difference between \SNP{} and monadic second-order logic (\MSO).
For example, a witness for an \SNP{} property may contain a linear order on the entire domain or higher-arity relations describing interactions between several vertices.
Moreover, the witness relations may be dense even if the input structure is sparse.

A crucial feature of this syntax is that satisfaction is preserved under taking induced substructures.
Indeed, deleting elements and restricting the witness relations to the remaining domain preserves every constraint imposed by the universal first-order part.
The well-quasi-ordering result also extends to relational structures over any fixed finite signature, with treedepth measured on their Gaifman graph\footnote{The vertex set of the \emph{Gaifman} (also called \emph{primal}) graph is the set of elements, and two distinct elements are connected by an edge whenever they occur together in a tuple of an input relation.}~\cite{lopezFirstOrderPreservation2023,dingSubgraphsWellquasiordering1992}. Consequently, for each fixed formula and treedepth bound, the negative instances have a finite set of minimal obstructions under the induced-substructure relation.
These fixed obstructions can be tested for on structures of bounded treedepth in linear time~\cite{Courcelle90},
yielding the following well-established fact:

\begin{fact}
\label{fact:nonuniform_snp}
For every problem $\mathcal{P}$ in \SNP expressed by a fixed sentence $\varphi$ and for every fixed integer $k$, there exists a linear-time algorithm that solves $\mathcal{P}$ on all relational structures whose Gaifman graph has treedepth at most $k$.
\end{fact}

In the setting of graphs, the well-quasi-ordering argument applies to every hereditary graph property and Fact~\ref{fact:nonuniform_snp} is often understood in those terms. Note that avoiding a given set of finitely many induced subgraph obstructions arising from the well-quasi-ordering on bounded-treedepth graphs can be easily encoded in an \SNP{} sentence---in this sense, Fact~\ref{fact:nonuniform_snp} establishes the existence of non-uniform linear-time algorithms for testing hereditary graph properties on graphs of bounded treedepth.

However, the generality of Fact~\ref{fact:nonuniform_snp} comes at a substantial price (see also~\cite{FellowsL94}):
\begin{enumerate}[itemsep=0.5pt]
    \item It guarantees the existence of algorithms, but does not provide a procedure for constructing any such algorithm from $\varphi$ and $k$.
    \item It gives no computable bound on how the running time depends on $\varphi$ and $k$.
    \item The promised algorithm merely tests for obstructions---it does not construct witness relations.
\end{enumerate}

In the language of parameterized complexity, algorithms whose existence is guaranteed in this way are sometimes called \emph{non-uniform fixed-parameter algorithms}~\cite[Definition~6.15]{CyganFKLMPPS15}. These are \textbf{not} fixed-parameter algorithms in the usual sense: fixed-parameter tractability (\FPT) requires an algorithm whose running-time dependence on the formula length and treedepth is bounded by a computable function; see, for example,~\cite[Definition~2.1.1]{DowneyF13},~\cite[Definition~1.4]{FlumG06}, and~\cite[Definition~1.2]{CyganFKLMPPS15}.\footnote{We also refer to these books for an introduction to parameterized complexity.}
For a range of hereditary graph problems, fixed-parameter tractability with respect to treedepth remained open despite the non-uniform linear-time algorithms guaranteed by Fact~\ref{fact:nonuniform_snp}. Well-known examples include computing the queue number~\cite{HeathR92,DujmovicJMMUW20,BhoreGMN22},
stack number (also known as book thickness)~\cite{BK.btg.1979,Yan.EPG.1989,BhoreGMN20}, track number~\cite{DujmovicWood2004,DujmovicMW05,BannisterDDEW19},
or twin-width~\cite{BonnetKTW22,BonnetGKTW24,GanianR26} of a graph\footnote{Further examples are provided in Section~\ref{sec:applications}.}; all four cases have been explicitly identified as open problems even in recent works~\cite{BannisterDDEW19,GanianR26,DepianFGS25,BhoreGMN22}.

\subsection*{Our Contribution}
We present a \emph{single}, \emph{uniform}, and \emph{constructive} fixed-parameter algorithm for all problems in \SNP{} parameterized by treedepth.
In particular, we establish:

\begin{alphatheorem}[Simplified formulation of \cref{main_thm:treedepth_meta_theorem}]
\label{thm:treedepth_meta_theorem_simplified}
There is an \Cref{alg:main-algorithm} that takes as input an \SNP{} sentence $\varphi$ and a relational structure $\mathcal{R}$, decides whether $\mathcal{R}\models\varphi$, and runs in time
\[
    f(\Size{\varphi},\td(\mathcal{R}))
    \cdot \Size{\mathcal{R}}
\]
for some computable function $f$.
Moreover, if $\mathcal{R}\models\varphi$, the algorithm can output a witness for this in time
\[
    f(\Size{\varphi},\td(\mathcal{R}))
    \cdot \Size{\mathcal{R}}^{\arity(\varphi)}.
\]
\end{alphatheorem}

Here, $\Size{\mathcal{R}}$ denotes the encoding size of the input structure, $\td(\mathcal{R})$ is the treedepth of its Gaifman graph, and $\arity(\varphi)$ denotes the maximum arity of any relation in $\{R_1,\dots,R_r\}$. In particular, the theorem yields a fixed-parameter decision algorithm for every fixed sentence $\varphi$ in \SNP, parameterized by treedepth; it does not require a forbidden-substructure characterization as part of the input and the polynomial degree does not depend on $\varphi$.
Moreover, it also allows us to output an explicit interpretation of the existentially quantified relations in $\varphi$ as a witness for positive instances---however, such a witness may have size $\Size{\mathcal{R}}^{\arity(\varphi)}$, necessitating a more relaxed running time bound.

\begin{remark}
\label{remark:fine_grained_guarantees}
In the full statement of our main result, \cref{main_thm:treedepth_meta_theorem} in \Cref{sec:metatheorem}, we give a more fine-grained guarantee.
The algorithm can output any requested subset of the witness relations, and the polynomial exponent depends only on the maximum arity among those relations.
\end{remark}

A concrete description of $\mathbb{A}$ is deferred to the technical overview in the upcoming \Cref{sec:overview}; the full proof of our metatheorem is deferred to \Cref{sec:metatheorem}. One peculiarity of the result is the function~$f$ itself: our proof establishes computability via a concrete non-trivial procedure which constructs the function and is guaranteed to terminate, but does not itself provide any closed-form upper bound on its growth.

\subparagraph*{Applications.}
\cref{thm:treedepth_meta_theorem_simplified} resolves the parameterized complexity w.r.t.\ treedepth for a number of problems for which it was previously unknown.
Concretely, we use our metatheorem to provide $f(\td(G)) \cdot |V(G)|$-time algorithms for computing the stack number, queue number, track number, twin-width, and boxicity of a graph $G$.
These applications (which include all the examples listed below Fact~\ref{fact:nonuniform_snp}) are provided in Section~\ref{sec:applications}.

\subparagraph*{Barriers for Generalizing the Metatheorem.}
We exclude in \Cref{sec:extensions} three natural generalizations of our result unless $\P = \NP$.
First, we use Fagin's theorem~\cite{fagin1974generalized} to show
that \SNP\ formulas cannot be replaced with \ESO\ formulas (where $\psi$ can quantify over first-order variables, see, e.g.,~\cite{LibkinElementsOfFiniteModelTheory}) while maintaining tractability w.r.t.\ treedepth.
Second, we observe that there are \SNP{}-problems which are \NP-hard on graphs of pathwidth 2, meaning that it is impossible to generalize \Cref{thm:treedepth_meta_theorem_simplified} to pathwidth.
This shows a striking contrast to Courcelle's theorem~\cite{Courcelle90}, and the known results for \upsf{ASNP} and \upsf{GMSNP} described below.
Third, we rule out a generalization to the complexity class $\textsf{MaxSNP}_0$~\cite{CCbook,PapadimitriouY91MaxSNP}, which is the natural optimization variant of \SNP{}.

\subparagraph*{Related Work.}

Feder and Vardi~\cite{FederV98} established close connections between constraint
satisfaction problems (CSPs) and monotone monadic \SNP{} without inequality (\upsf{MMSNP}).
They also showed that
every problem in \NP{} is polynomial-time \emph{equivalent} to a problem
in \SNP{}.\footnote{These reductions do not necessarily preserve bounded
treedepth, and therefore this does not imply that all problems in \NP{} are tractable on graphs of bounded treedepth.}
Hence, assuming $\P \neq \NP$, we obtain from Ladner's theorem~\cite{Ladner75} that, unlike for finite domain CSPs with a fixed constraint language and \upsf{MMSNP}~\cite{Bulatov17CSPDich,Zhuk20CSPDich,Kun13}, there is no complexity dichotomy for \SNP{}.
Madelaine~\cite{Madelaine09} studies the connections between \upsf{MMSNP} and CSPs on classes of bounded treedepth.
Other relevant fragments of \SNP{} include guarded monotone \SNP{}
(\upsf{GMSNP})~\cite{BienvenuCLW14} and amalgamation \SNP{}
(\upsf{ASNP})~\cite{BodirskyKS20}.
The latter work showed that for each fixed \upsf{ASNP} or \upsf{GMSNP} sentence and treewidth
bound there is a polynomial-time model checking algorithm (see also the introduction of Bodirsky and Guzm{\'{a}}n-Pro~\cite{BodirskyG26} for further background on \upsf{GMSNP}).
Since \SNP{} sentences define hereditary properties, \SNP{} satisfies
a $0$-$1$ law: each sentence is either true on all finite
structures or false on almost all finite structures.
Kolaitis and Vardi~\cite{KolaitisV87} showed that it is moreover
decidable, given an \SNP{} sentence, which of these two cases holds.
Their proof uses Ramsey arguments on atomic types, which are core building blocks that our proof relies on as well.
We provide a detailed discussion of the similarities and differences of their work compared to ours in   \Cref{paragraph:kolatis_vardi_comp}.

Classical metatheorems provide \FPT{} algorithms for monadic second-order logic
\MSOtwo{} on graphs of bounded treewidth~\cite{Courcelle90,BodlaenderALinearTime},
\MSOone{} on graphs of bounded clique-width~\cite{CourcelleMR00,OumS06}, and first-order logic (\FO{}) on effectively nowhere dense graph classes~\cite{GroheKS17}.
The bounded-treewidth results also cover \MSO{} on relational structures
and output-sensitive query evaluation~\cite{FlumFG02QueryEvaluationTreeDecompositions}.
Gajarsk\'y and Hlin\v{e}n\'y~\cite{gajarskyKernelizingMSOProperties2015}
obtained algorithms for \MSO{} model checking on graphs of bounded treedepth
by kernelizing labeled trees of bounded height.
More recently, Bergougnoux, Chekan, and
Stamoulis~\cite{BergougnouxCS26} proposed randomized algorithms with
single-exponential dependence on treedepth and polynomial space
for each fixed formula of a logic extending fully existential
\MSOtwo{}, assuming that an elimination forest is given in the input.

The logical scope of these results differs from ours.
On the one hand, \SNP{} can only express properties that are closed under taking induced substructures, and therefore cannot express non-hereditary properties such as \problemfont{Hamiltonian Cycle}, which are captured by \MSOtwo{}~\cite{Courcelle90}.
On the other hand, \SNP{} allows existential quantification over arbitrary
relations on the domain, whereas \MSOone{} and \MSOtwo{}
quantify over vertex sets and, for \MSOtwo{}, also edge sets.
As mentioned above, some of our applications are already \NP-hard on graphs of bounded pathwidth, so their tractability cannot follow from metatheorems that apply to classes of bounded pathwidth, assuming $\P \neq \NP$.

Recall that our result provides a single uniform algorithm for all of \SNP{}, parameterized
by the formula length and the treedepth of the input structure.
So to some extent, the distinction between uniform and non-uniform tractability is at the heart of our contribution.
Due to the graph minor theorem~\cite{RobertsonS04} coupled with polynomial-time minor testing~\cite{RobertsonS95b}, this question has historically been of special interest for minor-closed graph classes~\cite{FellowsL87Nonconstructive,FellowsL89AnAnalogueoftheMyhill-NerodeTheorem,FellowsL94}.
Courcelle, Downey, and Fellows~\cite{CourcelleDF97} showed that the forbidden
minors cannot be computed from an \MSOtwo{} description of the class alone.
Under additional assumptions, Adler, Grohe, and Kreutzer~\cite{AdlerGK08ComputingExcludedMinors} give an algorithm that computes sets of forbidden minors.
For vertex deletion to a fixed minor-closed class, Sau, Stamoulis, and
Thilikos bound the size of the resulting obstructions~\cite{SauST23} and give
parameterized algorithms~\cite{SauST22}.

\section{Technical Overview}
\label{sec:overview}
In this section, we provide a high-level overview of the main proof ingredients used to establish \Cref{thm:treedepth_meta_theorem_simplified}.
For formal definitions and proofs, we refer to later sections.
Throughout this section, let $(\varphi, \RR)$ be an instance of \snpmodelchecking, i.e., the problem of deciding whether $\RR \models \varphi$ holds, where $\varphi = \exists R_1,\dots,R_r \forall x_1,\dots, \forall x_k \psi$ is an \SNP\ sentence and $\RR = (V, E_1, \ldots, E_b)$ a relational structure.
We call a tuple $(R'_1,\dots,R'_r)$ of relations $R'_i \subseteq V(\RR)^{\arity(R_i)}$, for $i \in [r]$, a \emph{solution} if
\[
(\RR,R'_1,\dots,R'_r) \models \forall x_1 \cdots \forall x_k\,\psi.
\]
For the purpose of this overview, we assume that the domain $V$ comes with (an arbitrarily constructed) fixed total order, allowing us to speak about \emph{smaller} and \emph{smallest} elements.
Before we discuss our main result, let us first describe a pruning-based \emph{non-uniform} algorithm~$\mathbb{U}$ for \snpmodelchecking\ and examine the guarantees it provides.

\subsection{Extending the Non-Uniform Algorithm}
\label{sec:overview-extend-non-uniform}
Fix $\varphi$ and a positive integer $h$.
As hinted in \Cref{sec:introduction}, there is some \emph{finite} obstruction set $\JJ_{h,\phi}$, such that a relational structure $\RR$ with treedepth at most $h$ satisfies $\phi$ if and only if it contains no structure of $\JJ_{h,\phi}$ as an induced substructure.
When analyzing the structure of an elimination forest $T$ of the Gaifman graph of $\RR$, it is not hard to see that we can define some notion of \emph{types} for $V(T)$ and ``prune'' $T$ to only keep a bounded number of siblings of the same type.
Indeed, it is possible to define the types in such a way that, given an integer $p$, we can prune $\RR$ to an induced substructure $\RR^p$ of size bounded by a computable function $g(p,h)$, such that $\RR$ contains an induced substructure of size at most $p$ if and only if $\RR^p$ contains the same induced substructure.
Intuitively, we keep at most $p$ sibling-leaves for each possible interaction between the ancestors of a leaf and the relational structure $\RR$ (whose number can be bounded in $h$ and $\Size{\varphi}$).
After this step, the substructures rooted at the parents of the leaves can only have one out of boundedly-many ``isomorphism and neighborhood types'', and we can iteratively apply the above-sketched idea on our way to the root of $T$.

Thus, if an oracle provided us with the size $M$ of the largest obstruction in $\JJ_{h,\phi}$, we could set $p = M$ and have $\RR^p \models \phi$ if and only if $\RR \models \phi$.\footnote{Similar pruning arguments appear in, e.g., algorithms for
	one-queue layouts~\cite{BhoreGMN22} and for \MSO\ model checking
	on labeled trees of bounded height~\cite{gajarskyKernelizingMSOProperties2015}.
	Madelaine~\cite{Madelaine09} uses a similar idea to reduce colored structures of bounded treedepth to bounded-size induced substructures while preserving homomorphic equivalence.}
Since the size of $\RR^p$ is bounded by $g(p,h)$, we can decide whether $\RR^p \models \phi$ by brute force.
This immediately gives rise to our pruning-based non-uniform algorithm $\mathbb{U}$: Prune $\RR$ to $\RR^p$ and check whether $\RR^p \models \varphi$.
This approach has one inherent catch: up front, it is not clear how we can compute the correct value of $p$ as it depends on the unknown set of forbidden substructures $\mathcal{J}_{h,\phi}$.

A natural approach to circumvent this limitation is to try increasing values of $p$.
This leads to a ``semi-\FPT'' algorithm:
If we find a $p$ such that $\RR^p \not \models \phi$, we face a negative instance and can terminate.
In fact, since $\JJ_{h,\varphi}$ is finite, this algorithm will eventually terminate for every negative instance $(\varphi, \RR)$ and the maximum iteration count depends only on $h$ and $\varphi$.
The crux lies with positive instances:
since the correct upper bound for $p$ is unknown, we cannot phrase a termination criterion for our algorithm.
Hence, the issue is one-sided: for negative instances, the obstructions in $\JJ_{h,\phi}$ form a finite set of bounded-sized ``negative certificates'', and we eventually uncover such a certificate when facing a negative instance.
What is missing to turn this into a working \FPT\ algorithm is a set of certificates for positive instances that give us similar guarantees.

The key step for bounded treedepth is to obtain positive certificates that can be recognized on a bounded-size pruned instance and extended to the original input.
We define such certificates and show that if $\RR \models \varphi$ then there exists a large enough $p$ (depending only on $\varphi$ and $\td(\RR)$) such that $\RR^p$ contains one.
The uniform \Cref{alg:main-algorithm} is then a natural extension of the above-sketched ``semi-\FPT'' algorithm.
Starting at a certain threshold value (depending only on~$\varphi$ and $\td(\RR)$), we iteratively examine values for $p$ in increasing order.
For each such $p$, we check whether $\RR^p \not \models \varphi$, in which case we can correctly output that $\RR\not\models\varphi$.
If $\RR^p \models \varphi$, we additionally check whether $\RR^p$ admits one of our positive certificates.
If this is the case, then we can correctly deduce $\RR\models\varphi$.
Otherwise, we continue with $p+1$.

What makes this approach work is that bounded-size certificates now exist for any relational structure $\RR$, regardless of whether we are facing a positive or negative instance.
Besides allowing us to show that $\mathbb{A}$ terminates in time $f(\phi, \td(\RR)) \cdot \Size{\RR}$, positive certificates will also provide us with a ``recipe'' to extend a witness for $\RR^p \models \varphi$ into a witness for $\RR\models\varphi$.
Interestingly, we do not know a closed-form expression for the function $f$: it depends on the size of the largest obstruction in $\JJ_{h,\phi}$, for which we have no explicit bound.
However, we stress that the function $f$ is \emph{computable}.
Next, we describe these positive certificates in greater detail.

\subsection{Certificates for Positive Instances}
\label{sec:overview-certificates}
Let $T$ be a treedepth decomposition of $\RR$ of height $h \leq \td(\RR) + 1$, i.e., an elimination \emph{tree} of the primal graph of $\RR$.
The additional $+1$ stems from requiring $T$ to be a tree, which simplifies the upcoming discussion. Since the vertex set of $T$ equals the set of elements of $\RR$, we use the terms ``element'' and ``vertex'' to refer to the same object.
As indicated in \Cref{sec:overview-extend-non-uniform}, it is natural to define some notion of types when working with $T$.
We build our notion of types on the standard notion
of an \emph{atomic type}  (see, for example,~\cite[p.~226]{LibkinElementsOfFiniteModelTheory}).
The atomic type of a tuple $(v_1,\dots,v_z) \in V^z$, for some $z \in \mathbb{N}$,
captures which atoms on $v_1,\dots,v_z$ are true and which are false, which we use in the following form.

\begin{restatable}[Atomic type]{definition}{defAtomicType}
	\label{def:atomic_type}
	Let $\mathcal{R} = (V, E_1,\dots,E_b)$ be a relational structure and let $z\in \mathbb{N}$.
	An \emph{atomic formula} over the variables $x_1,\dots,x_z$ is a formula of the form
	$E_i(x_{n_1},\dots,x_{n_a})$ with $i \in [b]$, $a = \arity(E_i)$ and $(n_1,\dots,n_a) \in [z]^a$,
	or of the form $x_n = x_{n'}$ with $(n,n') \in [z]^2$.
	The \emph{atomic type} of a tuple $(v_1,\dots,v_z)\in V^z$ in $\mathcal{R}$ is
	\[
		\atp_{\mathcal{R}}(v_1,\dots,v_z) \coloneq
		\bigl\{\, \alpha(x_1,\dots,x_z) \text{ atomic} \;\bigm|\; \mathcal{R} \models \alpha(v_1,\dots,v_z) \,\bigr\}.
	\]
\end{restatable}

In addition to atomic types, we need some notion of type that is defined w.r.t.\ $T$, which we call the \emph{vertex $p$-type} for an integer $p \geq 1$ that will play the role of the pruning threshold from \Cref{sec:overview-extend-non-uniform}.
For simplicity, let us assume for a moment that $\RR$ encodes a graph $G$.
Then, the vertex $p$-type of a leaf $v$ of $T$ stores which of its up to $h - 1$ ancestors are neighbors.
For a non-leaf $v$ of $T$, we additionally store how many children $v$ has of each vertex $p$-type up to the number $p$ to ensure that the overall number of types remains bounded in $p$ and $h$.
This is a common technique when tackling graph problems parameterized by treedepth;
see, e.g.,~\cite{gajarskyKernelizingMSOProperties2015,BhoreGMN22}.
For general relational structures, adjacency is replaced by the atomic type of the root-to-$v$ path in $T$.
Formally, the vertex $p$-type is defined as follows.

\begin{restatable}[Vertex $p$-type]{definition}{defVertexType}
	\label{def:vertex-type}
	Let $p\in \mathbb{N}$ and $\mathcal{R}$ be a relational structure with treedepth decomposition $T$.
	For any $v\in V(\mathcal{R})$, its \emph{vertex $p$-type} in $(\RR, T)$ is the tuple $\vtxptype(v) = (\delta_v,\lambda_v)$ for an ancestor label $\lambda_v$ and a subtree label $\delta_v$, defined as follows.

	Let $d$ be the depth of $v$ in $T$ and $P = w_1,w_2,\dots,w_{d}$ be the path from root $w_1$ of $T$ to vertex $v = w_d$.
	The ancestor label $\lambda_v = \atp_\RR(w_1,w_2,\dots,w_{d})$ is the atomic type of the vertices on $P$.

	The subtree label is defined recursively.
	All leaves in $T$ have the empty subtree label $\varepsilon$.
	For a non-leaf vertex $v \in V(T)$, let $\mathcal{C}_v$ be the set of children of $v$ in $T$, and $\mathcal{T}$ be the set $\{\vtxptype(w) \mid w \in \mathcal{C}_v\}$.
	The subtree label is then a function $\delta_v: \mathcal{T} \to [0,p-1] \cup \{\star\}$, defined by
	\begin{align*}
		\delta_v(t) = \begin{cases}
			\Size{\{w \in \mathcal{C}_v \mid \vtxptype(w) = t\}} & \text{if $\Size{\{w \in \mathcal{C}_v \mid \vtxptype(w) = t\}} \leq p-1$,} \\
			\star                                         & \text{otherwise.}
		\end{cases}
	\end{align*}
\end{restatable}

Hence, $\lambda_v$ stores information about the ancestors of $v$ in $T$, while $\delta_v$ tells us how many of each kind of child $v$ has.
Observe that $\delta_v(t) = \star$ indicates that $v$ has at least $p$ children of vertex $p$-type~$t$, meaning that the amount of stored information is bounded by a function of the treedepth and $p$.
In the following, we say that a vertex $p$-type $t$ is \emph{$p$-abundant in the type}
$t'$ if $\delta_v(t) = \star$ for some vertex $v$ with vertex $p$-type $t'$.
This allows us to define a crucial notion: the \emph{$s$-multiplicity structure} $\RR^s$ of $\RR$, $T$, and $p$, for an integer $s \geq p$.
{Intuitively, $\RR^s$ has precisely $s$ copies of the $p$-abundant children of each vertex $v$; non-abundant types remain unchanged.}
Thus, the value of $p$ determines which vertex $p$-types of children of a vertex are abundant, whereas $s$ specifies their multiplicity in $\RR^s$---put simply, we count up to $p-1$ and then represent ``$p$ or more'' via $s$.

We can construct $\RR^s$ in a two-step procedure:
First, we keep the smallest
$s$ siblings of each $p$-abundant type and delete the remaining ones (if there are any) together with their entire subtrees.
By doing so, we also have to restrict all input relations to the retained elements.
Second, we process the pruned tree bottom-up.
If we encounter a vertex with fewer than $s$ children of a $p$-abundant type, we construct copies to bring their number up to $s$.

For $s = p$, the above construction is a formalization of the pruning operation described in \Cref{sec:overview-extend-non-uniform}.
In particular, we can observe that $\RR^p$ is an induced substructure of $\RR$ and thus $\RR^p \not\models\phi$ implies $\RR \not\models\phi$.
Note that for $s > p$, $\RR^s$ may contain elements not present in $\RR$.
While allowing for $\RR^s\not \subseteq \RR$
sounds counterintuitive at first, this is useful in \Cref{sec:overview-obtain-pure-solution} when we show that positive certificates exist for a suitable choice of $p$.

The $p$-multiplicity structure $\RR^p$ is an important stepping stone towards our concept of positive certificates.
Essentially, positive certificates are well-structured witnesses of $\RR^p \models \varphi$ that can be extended to a witness of $\RR \models \varphi$.
Observe that an \SNP\ formula $\phi = \exists R_1,\dots,R_r \,\forall x_1,\dots,x_k \colon \psi$ can express forbidden substructures on at most $k$ elements.
To detect these forbidden substructures, we cannot consider individual elements but have to consider tuples of size at most $k$.
This makes an extension of vertex $p$-types to tuples natural.

Defining the correct notion of a ``tuple type'' is non-trivial.
On the one hand, if the notion is too coarse, there is no guarantee that we can obtain a solution for $\RR$ from a positive certificate for $\RR^p$.
On the other hand, if the notion is too fine-grained, the number of types cannot be bounded, and we cannot hope that $\RR^p$ admits a positive certificate.
Intuitively, we want a notion of ``tuple type'' such that the following implication holds:
\begin{align*}
	\text{tuples } (v_1,\dots,v_k) \text{ and } &(w_1,\dots,w_k) \text{ have the same ``tuple type''}\\
	&\Longrightarrow \tag{IMPL}\label{eq:equivalence-overview}\\
	v_1,\dots,v_k \text{ behave identically to } & w_1,\dots,w_k \text{ in a positive certificate.}
\end{align*}

One could now be tempted to define the ``tuple type'' of $(v_1,\dots,v_k)$ by taking the $k$-tuple of the vertex $p$-types of $v_1,\dots,v_k$.
However, \textbf{this initial idea completely breaks down even for the simplest structures}, as the following example shows.
Let $\RR$ encode a star on $n$ leaves.
An \SNP\ sentence $\varphi$ can express the existence of a total order $\prec$ on the elements $V(\RR)$ of $\RR$.
Any total order suffices, so clearly $\RR \models \varphi$.
Although any two $k$-tuples of leaves would have the same ``naive'' tuple type, they cannot behave identically in any solution. To make this concrete, consider a total order $\prec$ on $V(\RR)$, fix $k=2$ and some leaves $v_1,v_2$, $w_1,w_2$ with $v_1 \prec v_2$ and $w_2\prec w_1$. Then, $(v_1,v_2)$ and $(w_1,w_2)$ would have the same type, yet $(v_1,v_2)$ behaves differently than $(w_1,w_2)$ in the considered solution.
In a nutshell, \SNP\ is powerful enough to distinguish between elements that seem indistinguishable in the input relational structure.
This requires us to define some
notion of ``tuple types'' which can, on the one hand, distinguish between siblings in a meaningful manner but, on the other hand, still ensure that the number of types is bounded.
Towards finding the correct notion of ``tuple types'', we define \emph{$p$-routes}, which enrich the concept of a vertex $p$-type with a ``map'' of how to reach it from the root of $T$.

\begin{restatable}[$p$-route and $p$-route of a vertex]{definition}{defRoutes}
	\label{def:p_route}
	Consider a relational structure $\mathcal{R}$ with treedepth decomposition $T$, and an integer $p \geq 1$.
	A \emph{$p$-route} is a non-empty finite sequence $\varrho = (e_1,\dots,e_d)$ of \emph{entries} $e_\ell = (t_\ell,i_\ell)$, each consisting of a vertex $p$-type $t_\ell$ and an index $i_\ell \in \mathbb{N}$; we call $d$ the \emph{length} of $\varrho$ also denoted as $\Size{\varrho}$.
	The entry $e_\ell$ is \emph{abundant} if $\ell \geq 2$ and $t_\ell$ is $p$-abundant in $t_{\ell-1}$.

	We write $\rt(v)$ for the $p$-route of a vertex $v$ of $T$, which is defined recursively as follows.
	For the root $u$ of $T$ we set $\rt(u) = ((\vtxptype(u), 1))$.

	Let $v$ be a non-root vertex in $T$ and $t = \vtxptype(v)$.
	Let $w$ be the parent of $v$ and let $i\in \mathbb{N}$ be such that $v$ is the $i$-th smallest\footnote{{With respect to the aforementioned total order on $V$}.} element
	among all children of $w$ of type $t$.
	Then $\rt(v) = \rt(w) \diamond (t,i)$.
\end{restatable}

Any given element of $\RR$ is uniquely identified by its $p$-route in $T$.
To see this, start from the root of $T$, and then descend the tree according to the route.
Each entry of the route tells us, for each vertex we encounter, which of its children we need to follow; \Cref{fig:example-route} visualizes such a descent.
\begin{figure}
	\centering
	\includegraphics{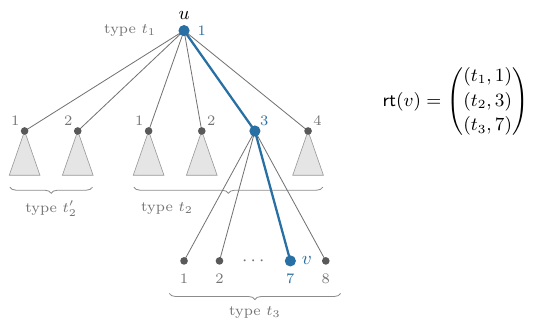}
	\caption{A $p$-route $\rt(v)$ (right) identifying a vertex $v$ in a treedepth decomposition $T$ (left).}
	\label{fig:example-route}
\end{figure}
However, on its own it is not a useful definition of a ``tuple type'':
The indices can be arbitrarily large, giving us no hope that we can bound the number of $p$-routes.

The key idea now is that $p$-routes can be normalized in a way that achieves the desired properties.
Recalling the total order example from the previous page, we already know that the vertex $p$-types do not suffice to capture behavior w.r.t.\ a hypothetical witness. However, we will later show that our witnesses can be assumed to be, in a carefully defined sense, uniform towards the considered ``arbitrary'' total order on $V$. So while $(a,b)$ can  behave differently than $(c,d)$ w.r.t.\ a solution if these are the respectively $10$th, $100$th, $30$th and $25$th elements of the same vertex $p$-type, we will make sure that it is only their relative order that matters in the end.
The idea is that we will not store the counting information in the $p$-route ``maps'' exactly, but in a way which is guaranteed to preserve the relative ordering among a set $S$ of such maps.
This is captured by the notion of a \emph{$p$-route pattern}.

\begin{restatable}[Normalization and $p$-route pattern]{definition}{defPatterns}
	\label{def:p_route_pattern}
	Consider a relational structure $\mathcal{R}$ with treedepth decomposition $T$ and an integer $p \geq 1$.
	Let $S$ be a finite set of $p$-routes. Define
	\begin{align*}
		I(S) = \bigl\{i \mid \text{ there is an abundant entry $(t,i)$ in at least one route of $S$} \,\bigr\},
	\end{align*}
	and let $\pi_S \colon I(S) \to [\Size{I(S)}]$ map each index to its rank in $I(S)$, that is, $\pi_S(i) = \Size{\{i' \in I(S) \mid i' \leq i\}}$.
	We now replace every entry $e_\ell$ in every $p$-route in $S$ by
	\begin{align*}
		\hat e_\ell = \begin{cases}
			(t_\ell, \pi_S(i_\ell)) & \text{if $e_\ell = (t_\ell, i_\ell)$ is abundant,} \\
			e_\ell                               & \text{otherwise,}
		\end{cases}
	\end{align*}
	and call $\nu_S(\varrho) = (\hat e_1,\dots,\hat e_d)$ the \emph{normalization of $\varrho$ with respect to $S$}.
	We write $\nu_S(S') = \{\nu_S(\varrho) \mid \varrho \in S'\}$ for $S' \subseteq S$.

	Let $(v_1,\dots,v_z)$ be a tuple of elements of $\RR$, (i.e., vertices of $T$) and let $S = \{\rt(v_1),\dots,\rt(v_z)\}$ be the set of their $p$-routes.
	The \emph{$p$-route pattern} of $(v_1,\dots,v_z)$ is the tuple $\bigl(\nu_S(\rt(v_1)),\dots,\nu
	_S(\rt(v_z))\bigr)$.
\end{restatable}

Note that the normalization is applied to all abundant entries of a tuple of $p$-routes together, regardless of where in the route they appear.
Let us give an example to illustrate the definition.

\begin{example*}
	\definecolor{abundantblue}{HTML}{286FA3}
	\newcommand{\highlight}[1]{\mathbf{\textcolor{abundantblue}{\ensuremath{#1}}}}

	Let $p = 4$, let $t_1$ be the vertex $p$-type of the root of $T$, and suppose the $p$-routes of $v_1,v_2,v_3,v_4$ are, respectively, displayed as column vectors below:
	\[
	\begin{pmatrix} (t_1,1) \\ \highlight{(t_2,3)} \\ \highlight{(t_3,7)} \\ \highlight{(t_4,6)} \end{pmatrix},\qquad
	\begin{pmatrix} (t_1,1) \\ \highlight{(t_2,3)} \\ \highlight{(t_3,4)} \end{pmatrix},\qquad
	\begin{pmatrix} (t_1,1) \\ \highlight{(t_2,7)} \\ \highlight{(t_3,3)} \end{pmatrix},\qquad
	\begin{pmatrix} (t_1,1) \\ (t_2',2) \\ (t_3,3) \end{pmatrix}.
	\]

	Suppose further that $t_2$ is $p$-abundant in $t_1$, that $t_3$ is $p$-abundant in $t_2$, and that $t_4$ is $p$-abundant in $t_3$, while $t_2'$ is not $p$-abundant in $t_1$ and $t_3$ is not $p$-abundant in $t_2'$.
	The abundant entries are then all entries of the first three routes except their first ones (highlighted in $\textbf{\textcolor{abundantblue}{blue}})$, and the indices they contribute are $I(S) = \{3,4,6,7\}$, so $\pi_S$ maps $3,4,6,7$ to $1,2,3,4$, respectively.
	Hence the $p$-route pattern of $(v_1,v_2,v_3,v_4)$ is
	\[ \left(
	\begin{pmatrix} (t_1,1) \\ \highlight{(t_2,1)} \\ \highlight{(t_3,4)} \\ \highlight{(t_4,3)} \end{pmatrix},\qquad
	\begin{pmatrix} (t_1,1) \\ \highlight{(t_2,1)} \\ \highlight{(t_3,2)} \end{pmatrix},\qquad
	\begin{pmatrix} (t_1,1) \\ \highlight{(t_2,4)} \\ \highlight{(t_3,1)} \end{pmatrix},\qquad
	\begin{pmatrix} (t_1,1) \\ (t_2',2) \\ (t_3,3) \end{pmatrix} \right).
	\]
	Note that $\pi_S$ is applied to an index regardless of where that index occurs.
	The index $3$ will be replaced by $1$ both in the entries $(t_2,3)$ of the first two routes and in the entry $(t_3,3)$ of the third.
	Likewise, the index $7$ is replaced with $4$ both in $(t_3,7)$ and in $(t_2,7)$.
	The $p$-route of $v_4$ is left untouched, since it consists of entries of non-$p$-abundant vertices, whose indices are at most $p-1$ to begin with.
	In particular, one and the same entry may be renumbered in one route and kept in another: the entry $(t_3,3)$ becomes $(t_3,1)$ in the third route, where it is abundant, and stays $(t_3,3)$ in the fourth, where it is not.
	\noindent\hfill{\small\textcolor{black}{\textbf{End of example.}}}\par
\end{example*}

The $p$-route pattern of an element tuple of $\RR$, or equivalently a vertex tuple in $T$, now forms the sought-after notion of a ``tuple type''.
Note that, for a fixed tuple length $z$, the number of $p$-route patterns for tuples of length at most $z$ can be bounded by a function of $p$, $z$, $h$, and $\phi$.

Now that we have established our notion of ``tuple type'', it is time to re-visit the intended~Implication~\eqref{eq:equivalence-overview} that we aim to achieve.
There, we wanted to ensure that two tuples with the same $p$-route pattern behave identically in a positive certificate.
We identify the notion of \emph{pureness} as the answer to this question.

Consider a solution $\mathcal{S}$ and a tuple $(v_1,\dots,v_z)$ over $V(\RR)$.
The \emph{solution type} of $(v_1,\dots,v_z)$ is the atomic type of $(v_1,\dots,v_z)$ in $\mathcal{S}$, i.e., it describes the behavior of the tuple in $\mathcal{S}$.
We are now interested in solutions that behave homogeneously on all tuples of the same $p$-route pattern.\footnote{
	Obtaining homogeneous sub-solutions forms the basis of many existing pruning-based algorithms, e.g.,~\cite{DepianFGS25,KolaitisV87}.}
\begin{restatable}[$p$-pure solution]{definition}{defPPureSolution}
	\label{def:p_pure_solution}
	Let $p\in \mathbb{N}$ and $(\RR, \varphi)$ be an \snpmodelchecking{} instance with treedepth decomposition $T$ and a solution $\mathcal{S} = R_1',\dots,R_r'$.
	A $p$-route pattern $\mu$ is \emph{$p$-pure} in $\mathcal{S}$ if any two tuples with $p$-route pattern $\mu$ have the same solution type in $\mathcal{S}$.
	We call $\mathcal{S}$ \emph{$p$-pure} if all $p$-route patterns of tuples of length at most $\Size{\varphi}$ are $p$-pure.
\end{restatable}
A consequence of this definition is that $p$-pure solutions admit a bounded-size description by assigning a solution type to each $p$-route pattern of length at most $\Size{\varphi}$.

\subsection{Extending Pure Solutions}
For pure solutions to serve as positive certificates, they must satisfy two crucial properties: (i) a $p$-pure solution of $\RR^p$ must extend to a solution of $\RR$, and (ii) if $\RR \models \varphi$, then $\RR^p$ admits a $p$-pure solution for sufficiently large $p$.
This section describes the proof technique for Property~(i) that is captured by the following theorem.

\begin{restatable}[Extension theorem]{theorem}{thmExtension}
	\label{thm:extension}
	Let $(\mathcal{R},\varphi = \exists R_1,\dots,R_r \,\forall x_1,\dots,x_k \colon \psi)$ be an instance of \snpmodelchecking{} with treedepth decomposition $T$ of height $h$ and $p \geq \Size{\varphi} \cdot h$.
	If $\mathcal{R}^p$ has a $p$-pure solution $\mathcal{S}^p$, then $\mathcal{R}$ has a solution $R_1', \ldots, R_r'$.

	Moreover, for a computable function $f$, after spending $f(p) \cdot \Size{\mathcal{R}}$ time on preprocessing, for every $i\in [r]$ and $\bar u \in V(\mathcal{R})^{a}$ we can decide whether $R_i'(\bar u)$ holds in $\mathcal{O}(ah \log (ah))$ time, where $a = \arity(R_i)$.
\end{restatable}
\begin{proof}[Proof sketch]
	We focus on the existence of $R_1', \ldots, R_r'$ and defer its construction to the full proof in \Cref{sec:metatheorem}.

	Let $R_1^p,\dots,R_r^p$ be a $p$-pure solution of $\RR^p$.
	The overall idea is to use tuples of $\RR^p$ as blueprints for tuples of $\RR$.
	More formally, to build a solution $R_1',\dots,R_r'$ for $\RR$ consider, for every $i \in [r]$,  every tuple $\overline{v} = (v_1,\dots,v_a)$ of elements of \RR, where $a$ is the arity of $R_i$.
	Let $\mu$ be the $p$-route pattern of~$\overline{v}$.
	Then, using the fact that $p \geq \Size{\varphi} \cdot h \geq a \cdot h$, it can be shown that there are elements $z_1,\dots,z_a$ in $\RR^p$ such that $\overline{z} = (z_1,\dots,z_a)$ has $p$-route pattern $\mu$.
	Since $R_1^p,\dots,R_r^p$ is $p$-pure, the concrete choice of $\overline{z}$ matters only up to the requirement on having the $p$-route pattern $\mu$ (which forces all matching tuples to have the same behavior in the solution).
	We let $R_i'(\overline{v})$ hold if and only if $R_i^p(\overline{z})$ holds.

	The difficult part is proving that $R_1',\dots,R_r'$ indeed form a solution for $\RR$.
	Assume that this is not the case.
	Then, there exist elements $v_1,\dots,v_k$ in $\RR$ such that $\psi(R_1',\dots,R_r',v_1,\dots,v_k)$ is false on $\RR$, that is, together with the chosen relations we find a forbidden substructure.
	We now argue that there are elements $w_1,\dots,w_k$ in $\RR^p$ that also form this forbidden substructure together with $R_1^p,\dots,R_r^p$.
	Let $\mu$ be the $p$-route pattern of $(v_1,\dots,v_k)$.
	We can find a tuple $(w_1, \ldots, w_k)$ in $\RR^p$ with the same $p$-route pattern $\mu$.
	Using the definition of $p$-route patterns, we can show that these tuples have the same atomic type, and, therefore, behave the same with respect to input relations and equality.

	The remaining part of the proof is showing that both tuples also behave identically w.r.t.\ the \emph{solution} relations.
	To this end, consider some relation $R_i$ of the solution of arity $a \leq \Size{\varphi}$ and a tuple $(v_{n_1},\dots,v_{n_a})$ of elements of $\RR$, where $n_1,\dots,n_a \in [k]$.
	Let $\mu'$ be the $p$-route pattern of $(v_{n_1},\dots,v_{n_a})$ in $\RR$.
	Recall that $R_i'(v_{n_1},\dots,v_{n_a})$ holds if and only if $R_i^p(z_1,\dots,z_a)$ does, where $(z_1,\dots,z_a)$ has the $p$-route pattern $\mu'$ in $\RR^p$.
	Crucially, the definition of $p$-route patterns ensures that also $(w_{n_1},\dots,w_{n_a})$ has the $p$-route pattern $\mu'$.
	Note that there is no guarantee that $w_{n_1},\dots,w_{n_a}$ forms the tuple $(z_1,\dots,z_a)$ that we used to define $R_i'$.
	Precisely for this reason we required the solution to be $p$-pure.
	Purity gives us the guarantee that $R_i^p(\overline{z})$ holds if and only if $R_i^p(\overline{w'})$ does.
	Hence, $R_i^p(\overline{w'})$ holds if and only if $R_i'(\overline{v'})$ does.
	Put differently, this shows that $\overline{v'}$ and $\overline{w'}$ also behave identically with respect to the relations $R_i'$, respectively $R_i^p$.
	Overall, with this proof strategy we can show that the two tuples $(v_1,\dots,v_k)$ and $(w_1,\dots,w_k)$ behave identically with respect to any atom of $\psi$.
	Therefore, $(w_1,\dots,w_k)$ would already form the forbidden substructure with respect to $R_1^p,\dots,R_r^p$ and $\RR^p$, a contradiction.
\end{proof}

\subsection{Obtaining Pure Solutions}
\label{sec:overview-obtain-pure-solution}
\Cref{thm:extension} can only shine if we can guarantee that $\RR^p$ admits a $p$-pure solution for sufficiently large~$p$ (provided $\RR \models \varphi$).
For the following, let us assume that $\RR$ has a solution.
Thus, also $\RR^p$ has a solution $\mathcal{S}$ as \SNP-satisfaction is closed under taking induced substructures (and $\RR^p$ is an induced substructure of $\RR$).
However, $\mathcal{S}$ might not be $p$-pure.
That is, there may be two tuples $\bar v$ and $\bar w$ with the same $p$-route pattern $\mu$, but different solution types.

A naive attempt to repair $\mathcal{S}$ is to delete vertices that make $\mathcal{S}$ impure, e.g., the vertices from $\bar w$.
However, this alters the structure $\RR^p$ and, consequently, can affect the $p$-types of the remaining vertices.
If we start from a larger structure $\RR'$ with $\RR' \models \varphi$ from which we delete vertices in such a way that we end up with a structure isomorphic to $\RR^p$, this problem would not occur.
In particular, we could guarantee that the $p$-route patterns of $\RR$ and the resulting structures would remain identical.

The natural candidate for $\RR'$ is the $s$-multiplicity structure $\RR^{s}$ for some $s \geq p$ (recall \Cref{sec:overview-certificates}), where $p$ remains fixed and defines the vertex $p$-types.
However, we need to show that $\RR^{s} \models \varphi$.
To this end, let us recall that $\RR \models \phi$ implies that $\RR$ does not contain an obstruction from $\JJ_{h,\phi}$.
Pick for $p$ a value that is at least $M = \max\bigl(\{1\}\cup\{\Size{V(H)}:H\in\JJ_{h,\varphi}\}\bigr)$, i.e., at least the number of vertices of every obstruction in $\JJ_{h,\varphi}$.
Then, $\RR$ and $\RR^{s}$ contain exactly the same induced substructures of size at most $p$.
Moreover, $\RR^{s}$ still has a treedepth decomposition of height $h$, so the same obstruction set applies.
Hence, $\RR^{s}$ does not contain an obstruction from $\JJ_{h,\phi}$, and we have $\RR^{s} \models \phi$:

\begin{restatable}{lemma}{thmBlowupFeasibility}
	\label{thm:blowup_preserves_feasibility}
	Let $(\mathcal{R},\varphi)$ be an instance of \snpmodelchecking{} and $T$ a treedepth decomposition of $\mathcal{R}$ of height $h$.
	Set $M = \max\bigl(\{1\}\cup\{\Size{V(H)}:H\in\JJ_{h,\varphi}\}\bigr)$.
	Then, for all $s \geq p \geq M$ we have that
	$\mathcal{R} \models \varphi$ if and only if $\mathcal{R}^{s} \models \varphi$, where $\mathcal{R}^{s}$ is the $s$-multiplicity structure for $(\mathcal{R}, T, p)$.
\end{restatable}

\Cref{thm:blowup_preserves_feasibility} shows that we can start from a solution of $\RR^{s}$.
Moreover, we can afford to delete many vertices and still retain an induced copy of $\RR^p$.
It remains to choose our deletion strategy such that we eventually arrive at a $p$-pure solution.
One way of ensuring this is finding a (sub)structure in which the solution type
of a tuple depends only on its $p$-route pattern.
Then, the solution will be $p$-pure ``by design''.
We now apply Ramsey-based arguments to find such a structure.

\begin{restatable}[Existence of a $p$-pure solution]{theorem}{thmExistenceOfPureSolution}
	\label{thm:existenceOfPureSolution}
	Let $(\mathcal{R},\varphi)$ be an instance of \snpmodelchecking{} and $T$ a treedepth decomposition of $\mathcal{R}$ of height $h$.
	Set $M = \max\bigl(\{1\}\cup\{\Size{V(H)}:H\in\JJ_{h,\varphi}\}\bigr)$.
	Then, for all $p \geq \max\{M,\Size{\varphi}h\}$ we have that
	if $\mathcal{R} \models \varphi$, then $\mathcal{R}^{p}$ has a $p$-pure solution, where $\mathcal{R}^{p}$ is the $p$-multiplicity structure for $(\mathcal{R}, T, p)$.
\end{restatable}
\begin{proof}[Proof sketch]
	On a high level, we want to find some integer $s \geq p$, such that the following holds.
	Given a solution $\mathcal{S}^{s}$ for the $s$-multiplicity structure $\RR^{s}$, we can prune $\RR^{s}$ down to a structure isomorphic to $\RR^p$ such that $\mathcal{S}^{s}$ restricted and transferred to $\RR^p$ is $p$-pure.
	We start our proof idea by illustrating how we achieve this for a single $p$-route pattern $\mu$.

	Fix some $p$-route pattern $\mu$ of length at most $\Size{\varphi}$.
	Recall that each entry of such a pattern is a pair $(t,i)$ where $t$ is a vertex $p$-type, and $i$ is an integer.
	The entry is abundant if $t$ is $p$-abundant in the vertex $p$-type in the preceding entry (if such an entry exists).
	Let $\beta$ be the number of distinct indices of abundant entries of $\mu$.
	If $\beta=0$, only one tuple has the $p$-route pattern $\mu$ and is therefore pure in every solution.
	So we assume that $\beta>0$.

	By normalization, the distinct indices in abundant entries are precisely $1,\dots,\beta$.
	We now establish the following crucial property: there exists a bijection $\pi$ from strictly increasing integer tuples in $[s]^\beta$ to vertex tuples of $\RR^{s}$ with $p$-route pattern $\mu$.
	Given a vertex tuple $\bar v$ of $\RR^{s}$ with $p$-route pattern $\mu$, let $\gamma$ be the tuple of the $p$-routes of $\bar v$ (and \emph{not} their $p$-route pattern).
	Now, let $\bar n$ be the tuple of length $\beta$ obtained by collecting the distinct indices of abundant entries of~$\gamma$, in increasing order.
	This integer tuple together with $\mu$ determines the original $p$-routes, which uniquely identify $\bar v$.
	Conversely, for any strictly increasing integer tuple $\bar n  = (n_1,\dots,n_\beta) \in [s]^\beta$, replace every abundant entry $(t,i)$ of $\mu$ by $(t,n_i)$, leaving all other entries unchanged.
	The resulting routes identify exactly one vertex tuple $\pi(\bar n)$, whose $p$-route pattern is $\mu$ (see \Cref{lem:pattern_realization}).

	Recall that we want to choose some value $s$ such that we can prune $\RR^{s}$ to a copy of $\RR^p$ and ensure that $\mu$ is $p$-pure in the restricted solution.
	We now use the above-established bijection $\pi$ in our Ramsey arguments as follows.
	Let $D$ be the number of possible solution types for tuples of length $|\mu|$.
	The assumption $p \geq \Size{\varphi}h$ ensures that $\beta \leq p$.
	To achieve our goal, we can set $s = \ramsey_\beta(p,D)$, where $\ramsey_\beta(p,D)$ is the finite hypergraph Ramsey number~\cite{ramseyProblemFormalLogic1930}: it is the smallest integer $\delta$ such that any edge coloring of a complete $\beta$-uniform hypergraph on $\delta$ vertices using $D$ colors contains a monochromatic hyperclique on $p$ vertices.

	To see that this is indeed the correct value for $s$, consider a solution $\mathcal{S}$ for the structure $\RR^{s}$.
	Note that $\mathcal{S}$ exists by \cref{thm:blowup_preserves_feasibility}.
	We create a $\beta$-uniform complete hypergraph $\mathcal{H}$ whose vertex set is the integer set $[s]$.
	For each subset of $[s]$ of size $\beta$, we color the corresponding hyperedge in $\mathcal{H}$ by using the above-sketched bijection $\pi$.
	Concretely, for such a subset $S'$ create the $\beta$-tuple $\bar n$ by ordering the vertices of $S'$ increasingly, and use this tuple to obtain the corresponding vertex tuple~$\pi(\bar n)$ of $\RR^{s}$ with $p$-route pattern $\mu$.
	Then, $\pi(\bar n)$ has some solution type in $\mathcal{S}$, and we color the hyperedge $S'$ by this solution type.

	By Ramsey's theorem, there is some subset $N$ of $[s]$ of size exactly $p$, such that $\mathcal{H}$ induced by~$N$ is monochromatic.
	By the construction of $\mathcal{H}$, the corresponding vertex tuples, or equivalently element tuples, all have the same solution type.
	Now, prune $\RR^{s}$ as follows.
	Starting from the root, let $v$ be a non-pruned vertex of type $t$, and let $t'$ be a type that is abundant in $t$.
	Then, $v$ has exactly $s$ children of type $t'$ in $\RR^{s}$, label them $1$ to $s$.
	Keep exactly those children of type $t'$ whose label belongs to $N$, and delete the others together with their subtrees.
	We use the same set $N$ at every retained vertex and for every abundant type, and leave all children of non-abundant types untouched.

	Since $|N| = p$, the resulting structure $\mathcal{U}$ is isomorphic to $\RR^p$.
	In the end, each tuple of $\RR^p$ with $p$-route pattern $\mu$ corresponds to a tuple of $\mathcal{U}$ whose distinct abundant indices belong to $N$ (and thus have $p$-route pattern $\mu$).
	Therefore, all such tuples have the same solution type in $\mathcal{S}$ restricted to $\mathcal{U}$.
	Transferring this restricted solution along the isomorphism gives a solution for $\RR^p$ that is pure w.r.t.\ $\mu$.

	By nesting the above arguments, we can achieve the result for all relevant $p$-route patterns of~$\RR$.
	More concretely, $s$ is obtained by nesting the hypergraph Ramsey numbers at most $K$ times, where $K$ is the number of $p$-route patterns of tuples of length at most $\Size{\varphi}$ occurring in $\RR$.
	Let $\mu_1, \ldots, \mu_K$ be the $p$-route patterns in $\RR$.
	We fix one solution for $\RR^{s}$ and iteratively choose nested sets $N_i$ of indices.
	More concretely, for $i = 1, \ldots, K$, we apply Ramsey's theorem to the hypergraph restricted to $N_{i - 1}$.
	This way, we can obtain a monochromatic set $N_{i + 1} \subseteq N_i$ w.r.t.\ $\mu_i$.
	Passing from $\mu_{i - 1}$ to $\mu_{i}$ preserves the purity already obtained for the patterns $\mu_j$, $j \leq i - 1$.
	Hence, the index set $N_k$ that we eventually obtain is pure w.r.t.\ all $p$-route patterns.
	We select the Ramsey bounds such that we can ensure that $N_K$ has $p$ indices.
	If we use the above-described pruning procedure for $\RR^{s}$ w.r.t.\ $N_K$, we obtain a $p$-pure solution.

	As to be expected when applying (iterative) Ramsey arguments, the obtained value $s$ is astronomical.
	However, note that the above arguments are purely existential.
	Our algorithm never constructs $\RR^{s}$.
\end{proof}

\subsection{The Algorithm \texorpdfstring{$\boldsymbol{\mathbb{A}}$}{A}}
We are now ready to combine \cref{thm:extension,thm:existenceOfPureSolution} to obtain our \Cref{alg:main-algorithm} for \snpmodelchecking{}.
Given an instance $(\RR, \varphi)$ of \snpmodelchecking, we first compute the primal graph $G$ of $\RR$.
Then, we compute its treedepth decomposition $T$ of height $h \leq \td(\RR) + 1$ using a known algorithm~\cite{ReidlRVS14}.
Afterwards, we iteratively consider, for larger and larger integers $p \geq \Size{\varphi} \cdot h$, the $p$-multiplicity structure $\RR^p$ for $(\RR, T, p)$ until the following holds:
Either the set of all solutions $\mathcal{L}^p$ of $\RR^p$ is empty, in which case we face a no-instance since \SNP-satisfaction is closed under induced substructures, or $\mathcal{L}^p$ contains a $p$-pure solution $\mathcal{S}^p$.
Applying \Cref{thm:extension} to $\mathcal{S}^p$ then allows us to conclude that $\RR\models\varphi$.

While $\mathbb{A}$ is easy to state and its correctness follows almost immediately from previous arguments, its termination is not immediate.
To this end, we show that $\mathbb{A}$ terminates in an iteration for $p \leq \max\{M, \Size{\varphi}(\td(\RR)+1)\}$, where $M = \max\bigl(\{1\}\cup\{\Size{V(H)}:H\in\JJ_{\td(\RR)+1,\varphi}\}\bigr)$.
Consider the iteration for $p = \max\{M, \Size{\varphi}(\td(\RR)+1)\}$, if it is reached.
If $\RR \models \varphi$, then by \cref{thm:existenceOfPureSolution}, we have that $\RR^p$ has a $p$-pure solution, which is found by the algorithm and allows us to correctly detect the yes-instance.
Otherwise, if $\RR \not \models \varphi$, then we know that already $\RR^p$ must contain a member of $\JJ_{\td(\RR)+1,\varphi}$, and the algorithm correctly detects the no-instance.
See \Cref{lem:algorithm-termination-correctness} for the formal proof.

However, although \Cref{lem:algorithm-termination-correctness} yields an upper bound on the number of iterations of $\mathbb{A}$, on its own it is not strong enough to bound the running time of $\mathbb{A}$ by a computable function $f$ in $\Size{\varphi}$ and $\td(\RR)$---after all, \Cref{lem:algorithm-termination-correctness} gives no bound on the size of $M$ in terms of $\Size{\varphi}$ and $\td(\RR)$.
As our last missing ingredient, we show that there exists an algorithm to compute the obstruction set $\JJ_{\td(\RR) + 1,\varphi}$ (up to isomorphism) in time $g(\Size{\varphi}, \td(\RR))$ for some computable function $g$.

The idea to show that this function is computable is now based on the following enumeration argument.
For the enumeration argument, fix a formula $\phi$ and a treedepth bound $d \geq 1$, so that the elimination tree of a structure $\RR$ with $\td(\RR) \leq d$ has height at most $d+1$.
Set $M_d = \max\bigl(\{1\}\cup\{\Size{V(H)}:H\in\JJ_{d+1,\phi}\}\bigr)$ and try increasing values of $p \geq \Size{\varphi} \cdot (d+1)$.
For such a fixed value, we enumerate all structures of matching treedepth and signature such that $\mathcal{R} = \mathcal{R}^p$.
We check if all enumerated structures either do not fulfill $\phi$, or have a $p$-pure solution.
If this is the case, we stop.
The size of such a structure $\RR$ is bounded by $p,d,\Size{\varphi}$, and hence this enumeration is possible for each fixed $p$.
The arguments above guarantee that, once we have reached $p = \max\{M_d,\Size{\varphi}(d+1)\}$, all structures we enumerate either do not fulfill the formula, or have a $p$-pure solution.
Hence, this enumeration algorithm can compute an upper bound on the number of iterations algorithm~$\mathbb{A}$ requires.
In fact, we can even use this idea to compute the obstruction set $\JJ_{d,\phi}$.
Formally, we obtain the following.

\begin{restatable}{theorem}{theoremObstructionSetComputable}
	\label{thm:obstruction-max-size-computable}
	Let $h \geq 1$ be an integer and $\varphi$ be an \SNP formula.
	There is an algorithm that computes the obstruction set $\JJ_{h,\varphi}$ (up to isomorphism) in time $f(h, \Size{\varphi})$ for some computable function $f$.
	In particular, we have $\max\bigl(\{1\}\cup\{\Size{V(H)}:H\in\JJ_{h,\varphi}\}\bigr) \leq g(h, \Size{\varphi})$ for some computable function $g$.
\end{restatable}

\Cref{thm:obstruction-max-size-computable} bounds the iteration count of $\mathbb{A}$ and thus also the maximum value of $p$ by some computable function $g(\Size{\varphi}, \td(\RR))$.
Combining \Cref{thm:extension,thm:existenceOfPureSolution,thm:obstruction-max-size-computable} with the observation that every iteration of $\mathbb{A}$ and the preprocessing steps run in time $g'(\Size{\varphi}, \td(\RR))\cdot \Size{\RR}$ for some computable function $g'$, allows us to obtain the decision variant of \Cref{thm:treedepth_meta_theorem_simplified}.
For the constructive part, we can retain a $p$-pure solution $\mathcal{S}^p$ in $\mathbb{A}$.
We can use \Cref{thm:extension} to extend it to a solution $\mathcal{S}$ of $\RR$ by iteratively considering all $\arity(\varphi)$-tuples $\overline{u}$ over $V(\RR)$.
Since the preprocessing step from \Cref{thm:extension} can be done once at the beginning and, as discussed above, $p$ can be bounded by the computable function $g$, the second part of \Cref{thm:treedepth_meta_theorem_simplified} follows.

Finally, let us remark that we could, in principle, turn the non-uniform algorithm $\mathbb{U}$ sketched in \Cref{sec:overview-extend-non-uniform} into a uniform algorithm by first computing the obstruction set with \Cref{thm:obstruction-max-size-computable} and then checking whether $\RR$ contains an obstruction as an induced substructure.
Our \Cref{alg:main-algorithm} has the additional advantage that it also constructs a witness for positive instances.

\subsection{Applications}
Although we see the development of \Cref{thm:treedepth_meta_theorem_simplified} as our main result, we also want to highlight some of its applications.
In particular, with our metatheorem we can solve several open problems.

Let us sketch this for the problem of computing the stack number of a graph.
Given a graph~$G$, the stack number of $G$, $\sn(G)$, is the smallest integer $k$ such that there is a total order $<$ on the vertices of $G$, and an edge mapping, called page-assignment, $\pi: E(G) \to [k]$.
The mapping and ordering must be so that there are no four vertices $v_1,v_2,w_1,w_2$ with $v_1 < w_1 < v_2 < w_2$ with $v_1v_2 \in E(G)$, $w_1w_2 \in E(G)$, and $\pi(v_1v_2) = \pi(w_1w_2)$.
{Already deciding whether $\sn(G) \leq 2$ holds is \NP-complete~\cite{Wid.CHC.1982,BK.btg.1979} and the problem is known to be \FPT\ with respect to the vertex integrity of $G$~\cite{DepianFGS25} and the feedback edge number of $G$~\cite{GMO+.TST.2024}.
For the special case of $\sn(G) \leq 2$, the problem is known to be \FPT\ in the treewidth of $G$~\cite{bannisterCrossingMinimization1page2018,GMO+.TST.2024}.
Whether the general problem is \FPT\ with respect to less restrictive parameters, in particular treedepth or treewidth, has repeatedly been posed as an open question~\cite{GMNZ.PCG.2021,BSK.Bop.2025,GMO+.TST.2024,DepianFGS25}.
We answer this question in the affirmative with respect to treedepth by capturing the corresponding decision problem using \SNP{} sentences.
}

For a fixed integer $k$, there is a simple formula $\phi_k$ such that $G \models \phi_k$ if and only if the stack number of $G$ is at most $k$.
Indeed, the following formula suffices:
\begin{align*}
    \varphi_k \coloneq{}& \exists\Lo\,\exists\Pg_1\cdots\exists\Pg_k\;\forall u,v,x,y\;\bigl(
    \mathrm{TotalOrder}(\Lo,u,v,x) \\
    &\qquad\land \mathrm{Assigned}(\Pg_1,\dots,\Pg_k,u,v) \\
    &\qquad\land \bigwedge_{i\in[k]}\mathrm{NoCrossing}(\Lo,\Pg_i,u,v,x,y)\bigr),
\end{align*}
where $\Lo$ and $\Pg_1,\dots,\Pg_k$ are binary relations, with $\Lo(u,v)$ representing $u< v$ and $\Pg_i(u,v)$ indicating that the edge $uv$ belongs to page $i$.
The formula $\mathrm{TotalOrder}(L,u,v,x)$ ensures that $L$ is a strict total order, $\mathrm{Assigned}(P_1,\dots,P_k,u,v)$ ensures that each edge is assigned to some page, and $\mathrm{NoCrossing}(L,P_i,u,v,x,y)$ checks that the obtained assignment does not have any crossings on page $P_i$.
All of these properties are easy to guarantee with forbidden substructures on at most four elements.

Note that the length of the formula above depends on $k$.
However, it is known that a graph with treedepth at most $k$ has stack number at most $k$~\cite{DW.GTG.2007}, and hence we can use the sketched formula to compute the stack number of a graph in time $f(\td(G)) \cdot |V(G)|$ using our algorithm.

In \Cref{sec:applications}, we demonstrate that our metatheorem can also be used to obtain \FPT\ algorithms parameterized by treedepth for further problems where this remained open until now: deciding if the queue number of $G$ is at most $k$, deciding if the track number of $G$ is at most $k$, deciding if the twin-width of a graph $G$ is at most $k$, and deciding if the boxicity of a graph $G$ is at most $k$.
In particular the formula for twin-width is significantly more involved than the one above.

\section{Preliminaries}
\label{sec:prelims}

For an integer $p \geq 1$ we also denote the set $\{1,2,\dots,p\}$ as $[p]$, and we let $[0] = \emptyset$.
Similarly, for integers $i,j$ with $i \leq j$ we set $[i,j] = \{i,i+1,\dots,j\}$.
We use the convention $\max \emptyset = 0$, so that maxima over empty index sets are well defined.
We allow appending another object to a tuple: for a (possibly empty) tuple $(x_1, \dots, x_z)$ and an object $y$ we define $(x_1, \dots, x_z) \diamond y = (x_1, \ldots, x_z, y)$.

\paragraph{Graphs and Hypergraphs.}
We use standard graph and hypergraph terminology, but remind the reader of some notions here.
For some integer $d \geq 1$, a \emph{$d$-uniform hyperclique} is a hypergraph on a vertex set $V$ with edge set $\binom{V}{d}$.
That is, it contains all possible hyperedges of size exactly $d$, and no further edges.
Given a hypergraph $H = (V,E)$, a \emph{$d$-uniform hyperclique in $H$} is a hypergraph $(V',E')$ with $V' \subseteq V$ and $E' \subseteq E$, such that $E' = \binom{V'}{d}$.

\paragraph{Relational Structures.}

\begin{definition}[Relational structures and primal graph]
	\label{def:relational_structures_primal_graph}
	A \emph{relational structure} $\RR = (V, E_1,\dots,E_b)$ consists of a non-empty \emph{domain} $V$, which we also denote $V(\RR)$, followed by \emph{relations} $E_1,\dots,E_b$ over $V$.
	The \emph{arity of relation} $E_i$ is denoted as $\arity(E_i)$ and the arity of a relational structure is $\arity(\RR) \coloneqq \max_{i\in [b]} \arity(E_i)$. Throughout the paper we assume that all relations have positive arity.
	The \emph{size} of $\RR$ is $\Size{\RR} \coloneqq \Size{V} + \sum_{i = 1}^{b} \sum_{\overline{u} \in E_i} \Size{\overline{u}}$.

	The \emph{primal graph} \emph{of a relational structure} is the graph with vertex set $V$ that contains each edge $\{v,w\}\in \binom{V}{2}$ if and only if $v$ and $w$ appear together in at least one tuple of a relation in $\RR$, that is, if there is $i\in [b]$ and $\bar u \in E_i$ such that $v \in \bar u$ and $w \in \bar u$.
\end{definition}

The primal graph of a relational structure is also often called the Gaifman graph of the structure. Since its vertex set is exactly the set $V(\RR)$ of elements, we use the terms ``element'' and ``vertex'' to some degree interchangeably.
For convenience and without loss of generality, we assume that there is a fixed total order on $V(\RR)$, which lets us speak of a smaller or larger element, and of the smallest or largest element of a set.
Concretely, we let it be the order in which the elements first appear in the input encoding of $\RR$, so that this order is available after a single scan of the input.

Moreover, note that the assumption that the arities of all relations are positive is without loss of generality, in the sense that relations of arity $0$ can be encoded by unary relations, and that removing relations of arity $0$ and introducing unary relations does not affect the primal graph.

Two notions on relational structures that we use throughout are the following.

\begin{definition}[Signature and induced substructure]
	\label{def:signature_substructure}
	The \emph{signature} of a relational structure $\RR = (V, E_1,\dots,E_b)$ is the sequence $\bigl(\arity(E_1),\dots,\arity(E_b)\bigr)$ of the arities of its relations.
	For a subset $U \subseteq V$ of the domain, the \emph{substructure of $\RR$ induced by $U$} is the relational structure
	\begin{align*}
		\RR[U] = \bigl(U,\ E_1 \cap U^{\arity(E_1)},\ \dots,\ E_b \cap U^{\arity(E_b)}\bigr),
	\end{align*}
	that is, we keep exactly those tuples all of whose entries lie in $U$.
	A relational structure $\RR'$ is an \emph{induced substructure} of $\RR$ if $\RR' = \RR[U]$ for some $U \subseteq V$; note that $\RR'$ then has the same signature as $\RR$.
\end{definition}

With some abuse of notation, we use the same notation for relation symbols and their interpretations, and write $E(\bar u)$ interchangeably with $\bar u \in E$ when the intended interpretation is clear.
We also write $v \in \bar u$ when $v$ occurs as an entry of $\bar u$; tuples are ordered and may contain repeated entries.

To categorize tuples of elements that are indistinguishable by the relations of a structure, we use the standard model-theoretic notion of an atomic type (see, e.g.,~\cite[p.~226]{LibkinElementsOfFiniteModelTheory}), which we use in the following form.

\defAtomicType*{}

We now provide a basic bound on the number of atomic types.

\begin{lemma}[Number of atomic types]
	\label{lem:atomic_type_bound}
	Let $\RR = (V,E_1,\dots,E_b)$ be a relational structure and let $z \in \mathbb{N}$.
	The number of possible atomic types of tuples of length $z$ is bounded by $2^{z^2 + \sum_{i \in [b]} z^{\arity(E_i)}}$.
\end{lemma}
\begin{proof}
	An atomic type of a tuple of length $z$ is a set of atomic formulas over $x_1,\dots,x_z$.
	For each relation $E_i$ there are $z^{\arity(E_i)}$ such formulas $E_i(x_{n_1},\dots,x_{n_a})$, one per tuple of indices, and there are $z^2$ equality formulas $x_n = x_{n'}$.
	Hence there are $z^2 + \sum_{i \in [b]} z^{\arity(E_i)}$ atomic formulas in total, and an atomic type is determined by the subset of them that holds.
\end{proof}

\paragraph{\SNP.}
An \SNP{} sentence $\varphi$ on a relational structure $\RR = (V, E_1,\dots,E_b)$ is a second-order logic formula of the form $\exists R_1,\dots,R_r \,\forall x_1,\dots,x_k \colon \psi(R_1,\dots,R_r,x_1,\dots,x_k)$,
where $\psi$ is a quantifier-free formula over relations $R_1,\dots,R_r,E_1,\dots,E_b$ and elements $x_1,\dots,x_k$, built from the atoms $E_i(\bar u)$ for $i \in [b]$, $R_j(\bar u)$ for $j \in [r]$, and $x_n = x_{n'}$ for $n,n' \in [k]$.
In particular,
the domain of each variable~$R_i$ is the set of relations over $V$,
and the domain of each variable $x_i$ is $V$.

For $\bar v = (v_1,\dots,v_k)$, we write $\RR \models \psi(R'_1,\dots,R'_r,\bar v)$ to mean that $\psi$ holds in the expansion $(\RR,R'_1,\dots,R'_r)$, where each $R_i$ is interpreted as $R'_i$ and $x_j$ is assigned $v_j$.
The input relations are left implicit in the arguments of $\psi$, and we omit $\RR \models$ when the underlying structure is clear.

The quantifier order in $\varphi$ that is prescribed by \SNP\ allows for the following observation:
\begin{observation}
	\label{obs:snp-satisfaction-substructure}
	Let $\varphi$ be an \SNP\ sentence over a relational structure $\RR$ such that $\RR\models \varphi$.
	Moreover, let $\RR' \subseteq \RR$ be an induced substructure of $\RR$.
	Then, $\RR' \models \varphi$.
\end{observation}

We call $\arity(\varphi) = \max_{i\in [r]} \arity(R_i)$ the arity of $\varphi$, and assume without loss of generality that each $\arity(R_i)$ is positive.
The \emph{length} $\Size{\varphi}$ of a formula $\varphi$ is the sum, over all \emph{occurrences} of an atom in the inner formula $\psi$, of the number of (not necessarily distinct) element variables of that atom; that is, an occurrence of an atom $x_n = x_{n'}$ contributes length $2$, and occurrences of atoms $E_i(\bar u)$ and $R_j(\bar u)$ contribute length $\arity(E_i)$ and $\arity(R_j)$, respectively.
Without loss of generality, we assume that each variable $R_i$ and $x_i$ is mentioned in the inner formula $\psi$, and that negations in $\psi$ occur only directly in front of atoms.
The latter two conventions ensure that $\varphi$ has an encoding of size $\BigO{\Size{\varphi}\log\Size{\varphi}}$ and that, up to logical equivalence, there are only finitely many \SNP\ formulas of a given length.
Therefore,
\[
	\Size{\varphi} \geq \sum_{i \in [r]} \arity(R_i) \geq \max\{r,\arity(\varphi)\}
	\qquad\text{and}\qquad
	\Size{\varphi} \geq k .
\]

We now define the problem we study.

\defparproblem{\snpmodelchecking}{Relational structure $\mathcal{R}$, \SNP{} sentence $\varphi$}{$\td(\mathcal{R})+\Size{\varphi}$}{Does $\mathcal{R} \models \varphi$?}

Since the problem is trivial otherwise, we always assume that the signatures of $\mathcal{R}$ and $\varphi$ match.
We also assume that $\mathcal{R}$ does not use any relations that do not appear in the formula (clearly, one could otherwise remove them from $\mathcal{R}$), and thus $\Size{\varphi} \geq \sum_{i \in [b]} \arity(E_i) \geq \max\{b,\arity(\RR)\}$.
We also define the natural notion of a solution, and a related solution notion that only concerns some subsets of the guessed relations.

\begin{definition}[Solution and $Q$-solution]
	\label{def:sol_q_sol}
	Let $\varphi$ be an \SNP{} sentence with the relation variables $R_1,\dots,R_r$, and let $\mathcal{R}$ be a relational structure.
	A \emph{solution} is an instantiation of $R_1, \ldots, R_r$ under which $\mathcal{R} \models \varphi$.
	For $Q\subseteq [r]$, an instantiation of $(R_i)_{i\in Q}$ is a \emph{$Q$-solution} if there is a solution to the instance that agrees with $(R_i)_{i\in Q}$. Equivalently, a $Q$-solution is the projection to $(R_i)_{i\in Q}$ of a solution.
\end{definition}

\paragraph{Treedepth.}
A \emph{rooted forest} $\mathcal F$ is a disjoint union of rooted trees.
For a vertex~$x$ in a tree~$T$ of $\mathcal F$, the \emph{depth} of~$x$ in $\mathcal F$ is the number of vertices on the path from the root of~$T$ to~$x$.
The \emph{height of a rooted forest} is the maximum depth of a vertex of the forest.
Let $V(T)$ be the vertex set of any tree $T \in \mathcal{F}$.
Let the \emph{closure} of a rooted forest~$\mathcal{F}$ be the graph $\clos({\mathcal{F}})$ with the vertex set $V(\clos({\mathcal{F}}))=\bigcup_{T \in \mathcal{F}} V(T)$ and the edge set $E(\clos({\mathcal{F}}))=\{xy \mid \text{$x \neq y$ and $x$ is an ancestor of $y$ in some $T\in\mathcal{F}$}\}$.

The \emph{treedepth} of a graph~$G$ ($\td(G)$) is the minimum height of any rooted forest $\mathcal{F}$ such that $G \subseteq \clos(\mathcal{F})$.

A \emph{treedepth decomposition} of a graph $G$ is a rooted tree $T$ such that $G \subseteq \clos(T)$.\footnote{This definition is slightly non-standard and increases the depth by up to $1$ compared to elimination forests, but it is more convenient later on.}
We can use the algorithm of Reidl, Rossmanith, S{\'{a}}nchez Villaamil, and Sikdar~\cite{ReidlRVS14} to compute a treedepth decomposition of $G$ of height at most $\td(G) + 1$ in time $2^{\BigO{\td(G)^2}}\cdot \Size{V(G)}$.
We define the treedepth $\td(\RR)$ of a relational structure $\RR$ as the treedepth of its primal graph, and similarly define a treedepth decomposition of $\RR$ to be a treedepth decomposition of the primal graph of~$\RR$.
Without loss of generality we assume that a treedepth decomposition of a relational structure contains no vertices beyond the elements of that structure.

\section{Formally Proving the Metatheorem}
\label{sec:metatheorem}

We now proceed to the proof of our main theorem, that is, the full version of \cref{thm:treedepth_meta_theorem_simplified}.
Before we state it, let us quickly explain the purpose of $Q$-solutions, which give the more fine-grained guarantees mentioned in \cref{remark:fine_grained_guarantees}.
If $\RR \models \exists R_1,\dots,R_r \forall x_1,\dots,x_k \psi(R_1,\dots,R_r,x_1,\dots,x_k)$ our algorithm can output a solution.
Naturally, the size of such a solution depends on the maximum arity of the relations $R_1,\dots,R_r$.
However, for a concrete application, it is conceivable that some of the relations $R_1,\dots,R_r$ are not of interest, in so far as they could be auxiliary relations needed to encode the problem as an \SNP{} sentence, but do not correspond to the actual witnesses one is aiming to obtain.
Hence, it is preferable to have an algorithm that can output only those relations which are of actual concern, and this is exactly what we do with the $Q$-solutions.

\begin{restatable}{mtheorem}{mainTheorem}
	\label{main_thm:treedepth_meta_theorem}
	There is an \Cref{alg:main-algorithm} that takes an \SNP{} sentence $\varphi$ with relation variables $R_1,\dots,R_r$ and a relational structure~$\mathcal{R}$ as input, runs in time $f(\Size{\varphi}, \td(\mathcal{R})) \cdot {\Size{\mathcal{R}}}$ for some computable function $f$, and decides whether $\mathcal{R} \models \varphi$.
	Moreover, if $\mathcal{R} \models \varphi$, then for any given non-empty $Q\subseteq [r]$ the algorithm can output a $Q$-solution of $\mathcal{R} \models \varphi$ in time $f(\Size{\varphi}, \td(\mathcal{R})) \cdot {\Size{\mathcal{R}}}^{\max_{i\in Q} \arity(R_i)}$.
\end{restatable}

\subsection{Types, Patterns, and Basic Properties}

We begin by establishing the basic definitions and properties that will later allow us to prove the theorem. The key concept of our algorithm is that two tuples of elements which are equivalent with respect to the input relational structure are also equivalent with respect to a solution. To properly define this, we first establish equivalence classes of individual elements, named \emph{vertex types}, and then use them to define equivalence classes of tuples of elements, which will be known as \emph{route patterns}.

\paragraph*{Vertex Types.}
We characterize the type of an individual element by two components: (i) the way it and its ancestors in the treedepth decomposition $T$ interact with the input relations and (ii) the counts (up to some value $p$) of how many children of each type it has in $T$.

\defVertexType*{}

The value $\star$ of a subtree label marks exactly those child types that appear at least $p$ times.
Intuitively, while a solution may treat vertices of a type differently when there are only a few of them, we will find that it does not need to distinguish between them once there are too many (say, at least $p$).
These vertices thus play an important role in defining equivalence classes, and accordingly we give them a name of their own.

\begin{definition}[$p$-abundant type and $p$-abundant vertex]
	\label{def:p_abundant}
	Let $\RR$ be a relational structure with treedepth decomposition $T$, and let $p \in \mathbb{N}$.
	A vertex $p$-type $t$ is \emph{$p$-abundant in} a vertex $p$-type $t' = (\delta,\lambda)$ if $t$ lies in the domain of $\delta$ and $\delta(t) = \star$.
	A vertex $v \in V(T)$ with parent $w$ is \emph{$p$-abundant} if $\vtxptype(v)$ is $p$-abundant in $\vtxptype(w)$.
\end{definition}

Unfolding the definition of $\delta_w$, a vertex $v$ with parent $w$ is $p$-abundant precisely if $w$ has at least~$p$ children whose vertex $p$-type is $\vtxptype(v)$, that is, if $v$ has at least $p-1$ siblings of its own vertex $p$-type.

One reason why the vertex types are so useful is that the ancestor label they contain already pins down, for any tuple of elements and given their pairwise ancestor-descendant relations, in which relations of the input structure $\RR$ it occurs.
That is, it determines the atomic type of that tuple, even across distinct input structures.

\begin{lemma}[Ancestor labels determine the relations]
	\label{lem:labels_determine_atp}
	Let $z \in \mathbb{N}$, let $\RR$ and $\RR'$ be relational structures over the same signature with treedepth decompositions $T$ and $T'$, and let $\bar v = (v_1,\dots,v_z) \in V(\RR)^z$ and $\bar w = (w_1,\dots,w_z) \in V(\RR')^z$ be such that for all $m,m' \in [z]$
	\begin{bracketenumerate}
		\item\label{item:labels_depth} $v_m$ has the same depth in $T$ as $w_m$ has in $T'$,
		\item\label{item:labels_label} $v_m$ and $w_m$ have the same ancestor label, and
		\item\label{item:labels_ancestor} $v_m$ is an ancestor of $v_{m'}$ in $T$ if and only if $w_m$ is an ancestor of $w_{m'}$ in $T'$; here and below we call a vertex an ancestor of itself.
	\end{bracketenumerate}
	Then $\atp_{\RR}(\bar v) = \atp_{\RR'}(\bar w)$.
\end{lemma}
\begin{proof}
	Since two vertices are equal exactly if each is an ancestor of the other, property~(\ref{item:labels_ancestor}) gives $v_m = v_{m'}$ if and only if $w_m = w_{m'}$, so $\bar v$ and $\bar w$ agree on the equality part of the atomic type of \Cref{def:atomic_type}.

	Now, fix some relation $E_i$ of arity $a$ and indices $n_1,\dots,n_a \in [z]$ such that $E_i(v_{n_1},\dots,v_{n_a})$ holds, and let $d_j$ be the depth of $v_{n_j}$.
	Then the distinct vertices among $v_{n_1},\dots,v_{n_a}$ form a clique in the primal graph of $\RR$, so, as $T$ is a treedepth decomposition, they lie on some path $(u_1,\dots,u_d)$ in $T$, where $u_1$ is the root and $u_d$ is the deepest among them.
	By the definition of the ancestor label $\lambda_{u_d} = \atp_{\RR}(u_1,\dots,u_d)$ in \Cref{def:vertex-type} we thus obtain $E_i(x_{d_1},\dots,x_{d_a}) \in \lambda_{u_d}$.

	By properties~(\ref{item:labels_depth}) and~(\ref{item:labels_ancestor}), the tuples $(v_{n_1},\dots,v_{n_a})$ and $(w_{n_1},\dots,w_{n_a})$ agree on the depths $d_1,\dots,d_a$ and on ancestor-descendant relationships, which shows that also $w_{n_1},\dots,w_{n_a}$ lie on a path from the root in $T'$, and that $w_{n_j}$ is the vertex of depth $d_j$ on it.
	Their deepest entries sit at the same position, say $\ell$; note that $v_{n_\ell} = u_d$, and that the ancestor labels of $v_{n_\ell}$ in $(\RR,T)$ and of~$w_{n_\ell}$ in $(\RR',T')$ are atomic types over the same signature, which coincide by property~(\ref{item:labels_label}).
	Hence $E_i(x_{d_1},\dots,x_{d_a})$ also lies in the ancestor label of $w_{n_\ell}$, that is, $\RR'$ satisfies $E_i$ on the vertices of depths $d_1,\dots,d_a$ of the root-to-$w_{n_\ell}$ path, so
	\[E_i(v_{n_1},\dots,v_{n_a}) \Longrightarrow E_i(w_{n_1},\dots,w_{n_a}).\]
	As the hypotheses are symmetric in $(\RR,T,\bar v)$ and $(\RR',T',\bar w)$, the tuples agree on all relations, concluding the proof.
\end{proof}

Coming up with this equivalence notion of vertex types is only useful if they reduce the number of distinct objects we need to speak about. This is captured by the following lemma.

\begin{lemma}[Enumerating vertex $p$-types]
	\label{lem:number_vertex_types}
	There is a computable function $g$ such that for any instance $(\RR,\varphi)$ of \snpmodelchecking{},  any treedepth decomposition $T$ of $\RR$ of height $h$, and any integer $p \geq 1$
	the number $\Size{\{\vtxptype(v)\mid v\in V(\RR)\}}$ of distinct vertex $p$-types is bounded by $g(h,p,\Size{\varphi})$.
\end{lemma}
\begin{proof}
	Consider some relational structure $\RR = (V,E_1,\dots,E_b)$, formula $\varphi$ of length at most $L$, and treedepth decomposition of $\RR$ of height $h$.
	Since we assume that each relation $E_i$ appears in $\varphi$, we have $b \leq L$ and $\max_{i \in [b]} \arity(E_i) \leq L$.

	By \Cref{lem:atomic_type_bound} the number of possible atomic types of a tuple of length $z$ is bounded by $2^{z^2 + \sum_{i \in [b]} z^{\arity(E_i)}}$.
	The vertex $p$-type only records the atomic types of vertices lying on a path in~$T$ starting at the root, and hence the length of a considered tuple is bounded by $h$.
	For a vertex $v$ of $T$, there are thus at most $2^{h^2 + L h^{L}}$ possibilities for the ancestor label of $v$.

	For each $i \in [h]$ we define a number $N_i$ that bounds the number of different possible vertex $p$-types of vertices at depth $i$ in $T$.

	Since a leaf of $T$ records no information about its children (since it has none), we can set $N_h = 2^{h^2 + L h^{L}}$.
	For $i \in [1,h-1]$ we set $N_{i} = N_h \cdot (p+1)^{N_{i+1}}$.
	This bound is also correct, since a non-leaf vertex again has at most $N_h$ possible ancestor labels, and it additionally stores the number of children of each type up to $p$.
	The final bound on the number of possible types is then given by $g(h,p,L) = \sum_{i \in [h]} N_i$, which is computable.
\end{proof}

\paragraph*{Routes and Route Patterns.}

Having established an equivalence notion on the vertex level, we proceed to define a more complex equivalence notion concerning tuples of vertices.
The motivation for this is evident, when considering \SNP{} sentences: such a sentence asks if there is some solution such that no forbidden pattern of size at most $k$ exists, where the inner formula describes this forbidden pattern.
It is therefore natural to consider equivalence classes of tuples of length at most~$k$. These will prove useful because we construct a solution in which a whole class behaves alike: if one tuple of a class does not witness the forbidden pattern, then no tuple of that class does.

Before formally describing this equivalence class, we define an alternative way to uniquely identify a vertex $v$ of $T$, the \emph{$p$-route} of $v$.
The idea is that we can find a vertex $v$ by describing which path from the root of $T$ one must follow to reach it, where the intermediate vertices of the path are identified by their vertex $p$-type, and their relative position among siblings with the same vertex $p$-type.
For the latter recall that we assume there is a fixed total order on $V(\RR)$.

\defRoutes*{}

Observe that a $p$-route identifies at most one vertex of $T$: two vertices with the same $p$-route have, by induction on the depth, the same parent, and then the same index among the children of that parent of their common vertex $p$-type, so they are equal.
The next lemma shows that we can efficiently determine the vertex $p$-type and $p$-route of all vertices.

\begin{lemma}[Computing vertex $p$-types and $p$-routes]
	\label{lem:computing_types_and_routes}
	There is a computable function $f$ such that for any instance $(\RR,\varphi)$ of \snpmodelchecking{}, any treedepth decomposition $T$ of $\RR$ of height $h$, and any integer $p \geq 1$, the vertex $p$-types and the $p$-routes of all elements of $\RR$ can be computed in time $f(h,p,\Size{\varphi}) \cdot \Size{\RR}$.
\end{lemma}
\begin{proof}
	Traversing $T$ once, we compute the parent and the depth of every element; this takes time $\BigO{\Size{\RR}}$.
	Recall that a vertex $p$-type is a pair $(\delta_v,\lambda_v)$ of a subtree label and an ancestor label, and that $\lambda_v$ depends only on the ancestors of $v$ while $\delta_v$ depends only on its descendants.
	We may therefore compute all ancestor labels in a top-down pass and then all subtree labels in a bottom-up pass as follows.

	\emph{Ancestor labels.}
	We store an atomic type as the set of the atomic formulas it contains, representing $E_i(x_{n_1},\dots,x_{n_a})$ by the pair $\bigl(i,(n_1,\dots,n_a)\bigr)$ and $x_n = x_{n'}$ by the pair $(n,n')$.
	Every tuple of every relation of $\RR$ consists of elements that are pairwise in ancestor-descendant relation in $T$, so it has a unique deepest entry, which we find in time linear in its arity.
	For each relation $E_i$ and each tuple $\bar u \in E_i$ we record the atomic formula $E_i(x_{d_1},\dots,x_{d_{\arity(E_i)}})$ at that deepest entry, where $d_j$ is the depth of the $j$-th entry of $\bar u$.
	Over all tuples this takes time $\BigO{\Size{\RR}}$.
	Now a tuple of elements of the root-to-$v$ path either consists of elements of the root-to-$w$ path, where $w$ is the parent of $v$, or contains $v$ and hence has $v$ as its deepest entry.
	The vertices $w_1,\dots,w_d$ of the root-to-$v$ path are pairwise distinct and $w_n$ has depth $n$, so by \Cref{def:atomic_type} the equality part of $\lambda_v = \atp_\RR(w_1,\dots,w_d)$ is exactly $\{x_n = x_n \mid n \in [d]\}$.
	By the definition of the ancestor label in \Cref{def:vertex-type} we thus obtain~$\lambda_v$ from $\lambda_w$ by adding the atomic formulas recorded at $v$ together with the equality atom $x_d = x_d$, where $d$ is the depth of $v$; if $v$ is the root, $\lambda_v$ consists of the atomic formulas recorded at $v$ together with the equality atom $x_1 = x_1$.
	Each ancestor label is a set of at most $\sum_{i} h^{\arity(E_i)} + h \leq \Size{\varphi} \cdot h^{\Size{\varphi}} + h$ atomic formulas, counting also the at most $h$ formulas of the equality part of the atomic type, so each of these steps takes time $f_1(h,\Size{\varphi})$ for a computable function $f_1$, and the pass takes time $f_1(h,\Size{\varphi}) \cdot \Size{\RR}$.

	\emph{Subtree labels.}
	The subtree label of a leaf is $\varepsilon$.
	For a non-leaf $v$ whose children all have their vertex $p$-types computed, we group the children by their vertex $p$-type and let $\delta_v$ map each occurring vertex $p$-type to the number of children carrying it, or to $\star$ if that number exceeds $p-1$.
	A vertex $p$-type has an encoding of size at most $f_2(h,p,\Size{\varphi})$ for a computable function $f_2$, so the grouping costs $\BigO{f_2(h,p,\Size{\varphi})}$ per child, and summing over all vertices gives $f_2(h,p,\Size{\varphi}) \cdot \BigO{\Size{\RR}}$ for the whole pass.

	\emph{$p$-routes.}
	Recall that the $p$-route of a non-root element $v$ with parent $w$ is $\rt(w) \diamond (\vtxptype(v),i)$, such that $v$ is the $i$-th smallest element among the children of $w$ of vertex $p$-type $\vtxptype(v)$ in the total order on $V(\RR)$.
	For the indices, we scan the elements in increasing order and maintain a counter for every pair consisting of an element and a vertex $p$-type; upon reaching $v$ we increment the counter of $(w,\vtxptype(v))$ and take its new value as the index $i$ of $v$.
	As we scan in increasing order, we meet the children of $w$ of any fixed vertex $p$-type in increasing order, so this is the correct index.
	For the routes, we traverse $T$ from the root downwards, starting with $\rt(u) = ((\vtxptype(u),1))$ at the root $u$ and setting $\rt(v) = \rt(w) \diamond (\vtxptype(v),i)$ at every other element $v$; as a parent is processed before its children, $\rt(w)$ is available when $v$ is reached.
	Each route consists of at most $h$ entries of size at most $f_2(h,p,\Size{\varphi})$.
	Both passes therefore take time $f_3(h,p,\Size{\varphi}) \cdot \Size{\RR}$ for a computable function $f_3$, which concludes the proof.
\end{proof}

For the equivalence classes,
the idea now is to normalize these $p$-routes so that they still record the vertex $p$-types, the depths and the relative order of the entries, while the value of $i$ in every entry $(t,i)$ becomes bounded.
With the normalized $p$-route (which exists when $p$ is chosen sufficiently large), one can no longer uniquely identify vertices of $T$; however, when following the normalized $p$-route of a vertex $v$, one reaches some natural representative $v'$ that behaves similarly to $v$.
This notion is extended in the natural way to tuples of vertices.

In a nutshell, for a vertex tuple $v_1,\dots,v_z$, we first consider the $p$-routes of all its entries.
Then, instead of storing the exact indices in the $p$-route entries, we only store their relative order among all entries of the routes considered.

\defPatterns*{}

Observe that this type of normalization is a purely syntactic operation on the $p$-routes themselves: besides the $p$-routes, the only information it uses is which of their entries are abundant, which is also (implicitly) recorded in them.
All the combinatorial properties of $p$-route patterns that we use are thus properties of the normalization operator, which we collect here for later reference.
The first two speak about properties preserved by the normalization, the third and fourth about consecutive application of normalization, and the last about the bounds on indices after normalization.

\begin{lemma}[Properties of normalization]
	\label{lem:normalization_properties}
	Let $S$ be a finite set of $p$-routes and let $S' \subseteq S$.
	The following properties all hold:
	\begin{bracketenumerate}
		\item\label{item:norm_structure}
		      $\nu_S$ preserves the length of every $p$-route, the vertex $p$-type in every entry, and hence the abundance of every entry.
		\item\label{item:norm_prefix}
		      For all $\varrho,\varrho' \in S$ and every $\ell \in [0,\min\{|\varrho|,|\varrho'|\}]$, the prefixes of length $\ell$ of $\varrho$ and $\varrho'$ coincide if and only if those of $\nu_S(\varrho)$ and $\nu_S(\varrho')$ do.
		\item\label{item:norm_composition}
		      $\nu_{S'} = \nu_{\nu_S(S')} \circ \nu_S$ on $S'$, that is,
		      normalizing $\varrho\in S'$ for $S'$ yields the same as first normalizing $\varrho$ for $S$ and then normalizing for the set of normalized routes $\nu_S(S')$ of $S'$.
		\item\label{item:norm_idempotent}
		      $\nu_{\nu_S(S)}$ is the identity on $\nu_S(S)$, that is, normalizing a normalized set of routes gives the same set.
		\item\label{item:norm_indices}
		      For every $\varrho \in S$, each abundant entry of $\nu_S(\varrho)$ has index at most $\Size{I(S)}$, which is at most $\Size{S}\,(h-1)$ if every route in $S$ has length at most $h$; each non-abundant entry of $\nu_S(\varrho)$ is equal to the corresponding entry of $\varrho$.
	\end{bracketenumerate}
\end{lemma}
\begin{proof}
	\textbf{(\ref{item:norm_structure}).}
	The operator $\nu_S$ replaces index components only, so lengths and the vertex $p$-types in all entries are preserved, and the abundance of an entry depends on nothing else.

	\medskip\noindent\textbf{(\ref{item:norm_prefix}).}
	We proceed by induction on $\ell$, the case $\ell = 0$ being trivial.
	Let $\varrho = (e_1,\dots,e_d)$ and $\varrho' = (e_1',\dots,e_{d'}')$ be routes in $S$ of length at least $\ell \ge 1$, and write $e_\ell = (t_\ell,i_\ell)$ and $e_\ell' = (t_\ell',i_\ell')$.

	Assume that $\varrho$ and $\varrho'$ share their prefix of length $\ell$.
	Then they also share their prefix of length $\ell-1$, so $\nu_S(\varrho)$ and $\nu_S(\varrho')$ do as well by the induction hypothesis.
	The $\ell$-th entry in $\nu_S(\varrho)$ and $\nu_S(\varrho')$ depends only on $e_\ell$ and $e_\ell'$, respectively, and whether they are abundant.
	Since $e_\ell = e_\ell'$ by assumption and abundance depends only on $t_\ell = t_\ell'$ and the vertex $p$-types in their (identical) preceding entry, this gives $\hat e_\ell = \hat e_\ell'$; for $\ell = 1$ there is no preceding entry, but then neither entry is abundant by \Cref{def:p_route} and both are left unchanged.

	Conversely, assume that $\nu_S(\varrho)$ and $\nu_S(\varrho')$ share their prefix of length $\ell$.
	By the induction hypothesis, $\varrho$ and $\varrho'$ share their prefix of length $\ell-1$, and by property~(\ref{item:norm_structure}) their $\ell$-th entries carry the same vertex $p$-type and are either both abundant or neither of them is.
	If they are not abundant, they are left unchanged, so $e_\ell = \hat e_\ell = \hat e_\ell' = e_\ell'$.
	If they are abundant, then by assumption $\pi_S(i_\ell) = \pi_S(i_\ell')$, and $\pi_S$ is injective, so $i_\ell = i_\ell'$.

	\medskip\noindent\textbf{(\ref{item:norm_composition}).}
	By property~(\ref{item:norm_structure}), normalization preserves abundance, so $I(\nu_S(S')) = \pi_S(I(S'))$.
	Since $I(S') \subseteq I(S)$ and $\pi_S$ is strictly increasing, each $i \in I(S')$ has the same rank in $I(S')$ as $\pi_S(i)$ has in $\pi_S(I(S'))$.
	Thus,
	\begin{align*}
		\pi_{\nu_S(S')}\bigl(\pi_S(i)\bigr) = \pi_{S'}(i) \qquad \text{for every } i \in I(S').
	\end{align*}
	Both normalization procedures therefore agree on abundant entries and leave all other entries unchanged, proving the identity on $S'$.

	\medskip\noindent\textbf{(\ref{item:norm_idempotent}).}
	Property~(\ref{item:norm_composition}) with $S' = S$ gives $\nu_S = \nu_{\nu_S(S)} \circ \nu_S$ on $S$. Thus, $\nu_{\nu_S(S)}$ is the identity on $\nu_S(S)$.
	Intuitively, the first normalization assigns each index in an abundant entry its rank within the set~$I(S)$. Ranking an already ranked set just returns the same ranks.

	\medskip\noindent\textbf{(\ref{item:norm_indices}).}
	An abundant entry receives an index in $[\Size{I(S)}]$ by \Cref{def:p_route_pattern}.
	Each element of $I(S)$ is the index of an abundant entry of some route of $S$, and a route of length at most $h$ has at most $h-1$ abundant entries, since its first entry is never abundant; hence $\Size{I(S)} \leq \Size{S}\,(h-1)$.
	Non-abundant entries are left unchanged by definition.
\end{proof}

Since $p$-route patterns are obtained via normalization of element tuples, the above lemma helps us in establishing some useful properties of any two element tuples sharing a $p$-route pattern: (1) they occupy corresponding positions in $T$, (2) they interact with the relations in $\RR$ in the same way, and (3) any deletion, copy, or reordering, as long as applied to both original tuples, yields two tuples which again share a $p$-route pattern.

\begin{lemma}[Properties of $p$-route patterns]
	\label{lem:pattern_properties}
	Let $p,z \in \mathbb{N}$, let $\RR$ be a relational structure with treedepth decomposition $T$, and let $\bar v = (v_1,\dots,v_z)$ and $\bar w = (w_1,\dots,w_z)$ in $V(\RR)^z$ have the same $p$-route pattern.
	The following properties all hold:
	\begin{bracketenumerate}
		\item\label{item:pattern_structure}
		      For all $m,m' \in [z]$, the vertices $v_m$ and $w_m$ have the same depth in $T$ and the same vertex $p$-type, and $v_m$ is an ancestor of $v_{m'}$ if and only if $w_m$ is an ancestor of $w_{m'}$ (even when counting a vertex as an ancestor of itself).
		\item\label{item:pattern_atp}
		      The two tuples have the same atomic type, that is, $\atp_{\RR}(\bar v) = \atp_{\RR}(\bar w)$.
		\item\label{item:pattern_subtuple}
		      For every $a \in \mathbb{N}$ and all $n_1,\dots,n_a \in [z]$, which need be neither distinct nor increasing, the tuples $(v_{n_1},\dots,v_{n_a})$ and $(w_{n_1},\dots,w_{n_a})$ have the same $p$-route pattern.
	\end{bracketenumerate}
\end{lemma}
\begin{proof}
	Let $S_v = \{\rt(v_1),\dots,\rt(v_z)\}$ and $S_w = \{\rt(w_1),\dots,\rt(w_z)\}$ be the sets of the $p$-routes of $\bar v$ and of $\bar w$.
	That the two tuples have the same $p$-route pattern means, by \Cref{def:p_route_pattern}, that
	\begin{align*}
		\nu_{S_v}(\rt(v_m)) = \nu_{S_w}(\rt(w_m)) \quad \text{for every } m \in [z].
	\end{align*}

	\textbf{(\ref{item:pattern_structure}).}
	By \Cref{def:p_route} the length of the $p$-route of a vertex is its depth in $T$, the last entry of the $p$-route carries its vertex $p$-type, and a vertex $x$ is an ancestor of a vertex $y$ if and only if $\rt(x)$ is a prefix of $\rt(y)$.
	It remains to observe that $\nu_{S_v}$ and $\nu_{S_w}$ both preserve lengths and the vertex $p$-type in every entry by \normprop{item:norm_structure}, and both preserve the prefix relation by \normprop{item:norm_prefix}; so these three properties may be read off the common normalized routes.

	\medskip\noindent\textbf{(\ref{item:pattern_atp}).}
	The ancestor label is part of the vertex $p$-type by \Cref{def:vertex-type}, so property~(\ref{item:pattern_structure}) provides the three hypotheses of \Cref{lem:labels_determine_atp} for $\RR' = \RR$ and $T' = T$, which gives the claim.

	\medskip\noindent\textbf{(\ref{item:pattern_subtuple}).}
	Let $S_v' = \{\rt(v_{n_1}),\dots,\rt(v_{n_a})\} \subseteq S_v$ and $S_w' = \{\rt(w_{n_1}),\dots,\rt(w_{n_a})\} \subseteq S_w$ be the sets of the $p$-routes of the two subtuples.
	Taking the images of the positions $n_1,\dots,n_a$ on both sides of the equality above yields $S^\circ \coloneqq \nu_{S_v}(S_v') = \nu_{S_w}(S_w')$.
	Hence, for every $m \in [a]$, we have
	\[
	\begin{aligned}
		\nu_{S_v'}(\rt(v_{n_m}))
		&= \nu_{S^\circ}\bigl(\nu_{S_v}(\rt(v_{n_m}))\bigr)
		&& \text{by \normprop{item:norm_composition} for $S_v' \subseteq S_v$,} \\
		&= \nu_{S^\circ}\bigl(\nu_{S_w}(\rt(w_{n_m}))\bigr)
		&& \text{by the common $p$-route pattern of $\bar v$ and $\bar w$,} \\
		&= \nu_{S_w'}(\rt(w_{n_m})),
		&& \text{by \normprop{item:norm_composition} for $S_w' \subseteq S_w$,}
	\end{aligned}
	\]
	which by \Cref{def:p_route_pattern} precisely gives that $(v_{n_1},\dots,v_{n_a})$ and $(w_{n_1},\dots,w_{n_a})$ have the same $p$-route pattern.
	Note that the indices $n_1,\dots,n_a$ need be neither distinct nor increasing, since $\nu_{S_v'}$ and $\nu_{S_w'}$ are applied to each position separately.
\end{proof}

Moreover, we record that the number of possible $p$-route patterns, and thereby the number of distinct equivalence classes to consider, is bounded by a function in our parameters.

\begin{lemma}[Enumerating $p$-route patterns]
	\label{lem:number_route_patterns}
	There is a computable function $f$ such that for any instance $(\RR,\varphi)$ of \snpmodelchecking{}, any treedepth decomposition $T$ of $\RR$ of height $h$, and any integers $p \geq 1$ and $z \in \mathbb{N}$
	there are at most $f(h,p,\Size{\varphi},z)$ distinct $p$-route patterns of $z$-tuples over $V(\RR)$.
	Moreover, given the set of vertex $p$-types occurring in $\RR$, one can compute in time $f(h,p,\Size{\varphi},z)$ a list of at most $f(h,p,\Size{\varphi},z)$ candidates that contains the $p$-route pattern of every $z$-tuple over $V(\RR)$.
\end{lemma}
\begin{proof}
	By definition, the $p$-route of an arbitrary element in $V(\RR)$ contains at most $h$ entries.
	By \normprop{item:norm_structure}, the same holds for the normalized routes in the $p$-route pattern of any $z$-tuple.
	Each entry consists of a vertex $p$-type and an index.
	\Cref{lem:number_vertex_types} bounds the number of possible vertex $p$-types by a function $g$ of $h,p$, and $\Size{\varphi}$.
	Moreover, the index $i$ in any entry $(t,i)$ in the $p$-route pattern of a $z$-tuple lies between $1$ and $M \coloneqq \max(z(h-1),p-1,1)$.
	Indeed, the index $i$ for abundant types is at most $z(h-1)$
	by \normprop{item:norm_indices}.
	Moreover, the index $i$ for non-abundant, non-root entries is at most $p-1$ by the definition of an abundant type.
	The index of the root entry is always $1$.

	There are thus at most $g(h,p,\Size{\varphi}) \cdot M$ possible entries, hence at most $h \cdot \bigl(g(h,p,\Size{\varphi}) \cdot M\bigr)^{h}$ possible normalized routes, as such a route is a sequence of at most $h$ entries.
	A $p$-route pattern of a $z$-tuple is a $z$-tuple of such routes, so there are at most
	\begin{align*}
		f(h,p,\Size{\varphi},z) = \Bigl(h \cdot \bigl(g(h,p,\Size{\varphi}) \cdot M\bigr)^{h}\Bigr)^{z}
	\end{align*}
	of them, and $f$ is computable since $g$ is.

	For the second statement, observe that the argument above does not merely count these patterns but describes them: by \normprop{item:norm_structure} the $p$-route pattern of a $z$-tuple over $V(\RR)$ is a $z$-tuple of sequences of at most $h$ entries $(t,i)$ in which $t$ is a vertex $p$-type occurring in $\RR$ and $i \in [M]$.
	Enumerating all objects of this shape from the given set of vertex $p$-types yields the desired list, and as there are at most $f(h,p,\Size{\varphi},z)$ of them, each with an encoding of size bounded in $h,p,\Size{\varphi}$ and $z$, writing them all down takes time bounded in the same parameters.
\end{proof}

As mentioned before, the $p$-route of a vertex $v$ uniquely identifies $v$.
For sufficiently large $p$, the routes in a $p$-route pattern identify canonical representative vertices.
We call these the realizing vertices of the $p$-route pattern.

\begin{definition}[Realizing vertices]
	\label{def:realizing_vertices}
	Let $\RR$ be a relational structure with treedepth decomposition $T$ of height $h$, and $\mu = (\varrho_1,\ldots,\varrho_z)$ be the $p$-route pattern of $z$ elements of $\RR$ for some $p \geq z \cdot h$.
	The \emph{realizing vertices} of $\mu$ are $(v_1,\ldots,v_z)$, where $v_i$ is the vertex identified by the $p$-route $\varrho_i$, for all $i\in[z]$.
\end{definition}

We now show that these realizing vertices indeed exist, and that their pattern is again $\mu$.

\begin{lemma}[Unique realizing vertices exist]
	\label{lem:realizing_vertices}
	Let $\RR$ be a relational structure with treedepth decomposition $T$ of height $h$, $p$ a positive integer, and $\mu = (\varrho_1,\ldots,\varrho_z)$ be the $p$-route pattern of $z$ elements of $\RR$ with $p \geq z \cdot h$.
	Then the realizing vertices of $\mu$ exist and are unique, and the $p$-route pattern of the realizing vertices of $\mu$ is again $\mu$.
	Moreover, every entry $(t,i)$ in the $p$-route of a realizing vertex satisfies $i \leq p$.
\end{lemma}
\begin{proof}
	Uniqueness is immediate, since distinct vertices have distinct $p$-routes, as explained right after \Cref{def:p_route}.
	Choose a tuple $(v_1,\dots,v_z)$ with $p$-route pattern $\mu = (\varrho_1,\dots,\varrho_z)$.
	Let $S = \{\rt(v_1),\dots,\rt(v_z)\}$, so that $\Size{S} \leq z$ and every route in $S$ has length at most $h$.

	For existence, fix $m \in [z]$, write $\varrho_m = (\hat e_1,\dots,\hat e_d)$ with $\hat e_\ell = (t_\ell,j_\ell)$, and let $u_1,\dots,u_d$ with $v_m = u_d$ be the vertices on the root-to-$v_m$ path in $T$, so that $\vtxptype(u_\ell) = t_\ell$ by \normprop{item:norm_structure}.
	We show by induction on $\ell$ that $(\hat e_1,\dots,\hat e_\ell)$ is the $p$-route of some vertex $y_\ell$ of $T$; note that then $\vtxptype(y_\ell) = t_\ell$, since the last entry of the $p$-route of a vertex carries its vertex $p$-type.
	For $\ell = 1$ we have $\hat e_1 = (\vtxptype(u_1),1)$ by \Cref{def:p_route_pattern}, which is the $p$-route of the root $u_1$ of $T$.
	For the step, let $y_{\ell-1}$ be given; it suffices to show that $y_{\ell-1}$ has at least $j_\ell$ children of vertex $p$-type $t_\ell$, as the $j_\ell$-th smallest of them is then the desired $y_\ell$.

	Suppose first that $\hat e_\ell$ is not abundant, so that $j_\ell$ is the index of $u_\ell$ among the children of $u_{\ell-1}$ of vertex $p$-type $t_\ell$ by \normprop{item:norm_indices}.
	Write $t_{\ell-1} = (\delta,\lambda)$; note that $\vtxptype(y_{\ell-1}) = t_{\ell-1} = \vtxptype(u_{\ell-1})$, so $\delta$ is the subtree label of both $y_{\ell-1}$ and $u_{\ell-1}$.
	As $u_\ell$ is a child of $u_{\ell-1}$ of vertex $p$-type $t_\ell$, the type~$t_\ell$ lies in the domain of $\delta$, and $\delta(t_\ell) \neq \star$ because $\hat e_\ell$ is not abundant.
	Hence both $u_{\ell-1}$ and $y_{\ell-1}$ have exactly $\delta(t_\ell) \leq p-1$ children of vertex $p$-type $t_\ell$, and $j_\ell \leq \delta(t_\ell)$ as $u_\ell$ is one of those children of $u_{\ell-1}$.

	If, on the other hand, $\hat e_\ell$ is abundant, then $t_\ell$ is $p$-abundant in $t_{\ell-1} = \vtxptype(y_{\ell-1})$, so $y_{\ell-1}$ has at least $p$ children of vertex $p$-type $t_\ell$, while $j_\ell\leq\Size{I(S)}\leq z(h-1)\leq p$ by \normprop{item:norm_indices} and the hypothesis.

	This argument also gives the claimed bound on the indices: the root entry has index $j_1 = 1$, a non-abundant entry has index $j_\ell \leq \delta(t_\ell) \leq p-1$, and an abundant entry has index $j_\ell \leq z(h-1) \leq p$.

	It remains to show that the $p$-route pattern of the realizing vertices $w_1,\dots,w_z$ of $\mu$ is again~$\mu$.
	By \Cref{def:p_route_pattern} we have $\varrho_m = \nu_S(\rt(v_m))$ for every $m \in [z]$, and $\varrho_m$ is the $p$-route of $w_m$ by \Cref{def:realizing_vertices}.
	So the set of the $p$-routes of $w_1,\dots,w_z$ is $\{\varrho_1,\dots,\varrho_z\} = \nu_S(S)$, and applying \Cref{def:p_route_pattern} to the tuple $(w_1,\dots,w_z)$ shows that its $p$-route pattern is $\bigl(\nu_{\nu_S(S)}(\varrho_1),\dots,\nu_{\nu_S(S)}(\varrho_z)\bigr)$.
	By \normprop{item:norm_idempotent}, this is $(\varrho_1,\dots,\varrho_z) = \mu$.
\end{proof}

\paragraph*{Multiplicity Structures.}
The purpose of $p$-route patterns is that we want to treat tuples of elements that share a pattern in an interchangeable way.
In this sense, we define a notion of pruning, in which each vertex keeps at most $p$ children of each vertex $p$-type.
We will show that this preserves the existence of all possible $p$-route patterns, since the realizing vertices of each pattern are untouched.

Indeed, the following defines pruning in a slightly more general way: we consider some $s\ge p$ and keep $s$ children of every abundant type, and even add copies of subtrees to make that exactly~$s$.
That is, if there are at least $p$ siblings in $T$ with the same vertex $p$-type, then the resulting structure has exactly $s$ siblings of this vertex $p$-type instead.

\begin{definition}[$s$-multiplicity structure]
	\label{def:s_multiplicity_structure}
	Let $s,p \in \mathbb{N}$ with $s\ge p \geq 1$, and let $\RR$ be a relational structure with treedepth decomposition $T$.

	We define the \emph{$s$-multiplicity structure} $\RR^s$ for $(\RR, T, p)$ together with the accompanying rooted tree $T^s$ on its elements as follows.
	First, let $\RR'$ be the substructure of $\RR$ induced by the set of elements for which every entry $(t,i)$ in their $p$-route satisfies $i \le s$ for every vertex $p$-type $t$.\footnote{Since the $p$-route of a vertex contains the $p$-routes to all its ancestors, this set of elements is closed under taking ancestors; equivalently, whenever an element is not taken, no element of the subtree rooted at it is.}
	Let~$T'$ be the rooted tree obtained from $T$ by removing all vertices which are not elements of $\RR'$; this is again a tree, as the removed vertices are removed together with their subtrees.

	If $s=p$, let $(\RR^p, T^p) = (\RR', T')$.
	Otherwise $s > p$, and we obtain $(\RR^s, T^s)$ from $(\RR', T')$ by processing the vertices of $T'$ bottom-up, each after all of its descendants, maintaining a current structure together with a current tree on its elements.
	When processing $w$, consider every vertex $p$-type $t$ that is $p$-abundant in $\vtxptype(w)$ and of which $w$ has $c_t < s$ children in the current tree; here, and throughout this definition, all vertex $p$-types are the ones in $(\RR, T)$.
	Let $u$ be the smallest child of $w$ of vertex $p$-type $t$, and add another $s-c_t$ disjoint copies of the subtree rooted at $u$ in the current tree.
	Here a copy is given by a bijection $\iota$ from the elements of that subtree onto fresh elements; we let $\iota(x)$ be a child of $\iota(y)$ whenever $x$ is a child of $y$ in that subtree, we make $\iota(u)$ a further child of $w$, and for every tuple of every relation of the current structure we add the tuple with each entry in the subtree replaced by its $\iota$-image.
	In particular, $V(T^s) = V(\RR^s)$.

	Regarding the total order on the elements of $\RR^s$, each fresh element of $\RR^s$ is sorted \emph{after} all elements of $\RR$, for an arbitrary order among the fresh elements.
\end{definition}

Note that since the subtree of $u$ is already processed when $w$ is, its copies need no further processing.
We now define the natural notion of \emph{retained elements}.

\begin{definition}[Retained elements]
	\label{def:retained_elements}
	Let $s,p \in \mathbb{N}$ with $s\ge p \geq 1$, and let $\RR$ be a relational structure with treedepth decomposition $T$.
	An element of $\RR^s$ is \emph{retained} if it is an element of $\RR$.
	Write $\orig(y)$ for the retained element that an element $y$ of $\RR^s$ is a copy of, obtained by following the copying bijections backwards.
	In particular, $\orig(y) = y$ for retained $y$.
\end{definition}

The following lemma collects the properties of the $s$-multiplicity structure that we use.

\begin{lemma}[Properties of the $s$-multiplicity structure]
	\label{lem:multiplicity_properties}
	Let $s,p \in \mathbb{N}$ with $s \geq p \geq 1$, and let $\RR$ be a relational structure with treedepth decomposition $T$ of height $h$.
	\begin{bracketenumerate}
		\item\label{item:mult_decomposition}
		      $T^s$ is a treedepth decomposition of $\RR^s$ of height $h$.
		\item\label{item:mult_relations}
		      A tuple of retained elements lies in a relation of $\RR^s$ if and only if it lies in that relation in $\RR$.
		\item\label{item:mult_types}
		      Every element $y$ of $\RR^s$ has in $(\RR^s,T^s)$ the vertex $p$-type that $\orig(y)$ has in $(\RR,T)$; in particular $y$ is $p$-abundant in $(\RR^s,T^s)$ if and only if $\orig(y)$ is $p$-abundant in $(\RR,T)$.
		\item\label{item:mult_counts}
		      In $T^s$, every vertex $y$ has exactly $s$ children of each vertex $p$-type that is $p$-abundant in $\vtxptype(y)$, and as many children of every other vertex $p$-type as $\orig(y)$ has in $T$.
		\item\label{item:mult_routes}
		      Every element of $\RR^p$ is retained and has the same $p$-route in $T^p$ as in $T$; in particular every tuple of elements of $\RR^p$ has the same $p$-route pattern in $T^p$ as in $T$.
		\item\label{item:mult_realizing}
		      Let $z \in \mathbb{N}$ with $p \geq z \cdot h$.
		      Then for every $z$-tuple of elements of $\RR$, the realizing vertices of its $p$-route pattern in $T$ are retained in $\RR^p$.
		\item\label{item:mult_patterns}
		      Let $z \in \mathbb{N}$ with $p \geq z \cdot h$.
		      Then a $p$-route pattern of a $z$-tuple occurs in $(\RR,T)$ if and only if it occurs in $(\RR^p,T^p)$.
		\item\label{item:mult_size}
		      If $(\RR,\varphi)$ is an \snpmodelchecking{} instance with $p \geq \Size{\varphi}$, then $\Size{\RR^s} \leq f(h,p,s)$ for a computable function $f$.
	\end{bracketenumerate}
\end{lemma}
\begin{proof}
	\textbf{(\ref{item:mult_decomposition}).}
	We show that throughout the construction in \Cref{def:s_multiplicity_structure}, every tuple of every relation of the current structure is a chain in the current tree, that is, its entries are pairwise in ancestor-descendant relation; applied to the end of the construction, this gives the claim.
	This holds initially as $T$ is a treedepth decomposition of $\RR$, and it is maintained by the passage to $(\RR',T')$, which leaves the ancestor-descendant relationship among the retained elements and their membership in relation tuples unchanged.
	In a copy step, the new tuples arise from tuples $\bar x$ of the current structure by replacing every entry in the subtree of $u$ by its $\iota$-image.
	If $\bar x$ meets that subtree then, being a chain, each of its remaining entries lies in the subtree as well or is a proper ancestor of $u$ and hence also an ancestor of $\iota(u)$, which is a child of $w$; the entries in the subtree keep their ancestor relation, as $\iota$ is an isomorphism onto the attached copy.
	So the image is again a chain, and $T^s$ is a treedepth decomposition of $\RR^s$.
	Copies sit at the depth of their originals and the deletion only removes vertices, so $T^s$ has height at most $h$.
	Conversely, let $v$ be a vertex of depth $h$ in $T$ and write $\rt(v) = ((t_1,i_1),\dots,(t_h,i_h))$.
	Let $u_1$ be the root of $T$ and, for $\ell \geq 2$, let $u_\ell$ be the child of $u_{\ell-1}$ of vertex $p$-type $t_\ell$ and index $1$.
	This is well defined since by construction $u_{\ell-1}$ has the same vertex $p$-type $t_{\ell-1}$ as the ancestor of $v$ at depth $\ell-1$ and thus as that ancestor has a child of vertex $p$-type $t_\ell$, so does $u_{\ell-1}$.
	Now every entry of $\rt(u_h) = ((t_1,1),\dots,(t_h,1))$ has index $1 \leq s$, so $u_h$ is retained, and $T^s$ has height exactly $h$.

	\medskip\noindent\textbf{(\ref{item:mult_relations}).}
	Passing to $\RR'$ keeps exactly those tuples all of whose entries are retained.
	A copy step adds the image of a tuple of the current structure under a copying bijection $\iota$; such an image differs from the tuple itself only if the tuple has an entry in the subtree of $u$, in which case the image contains a fresh element and is thus not a tuple of retained elements.
	So neither phase changes which tuples of retained elements lie in a relation.

	\medskip\noindent\textbf{(\ref{item:mult_types}) and (\ref{item:mult_counts}).}
	Write $\vtxptype(y) = (\delta_y,\lambda_y)$ as in \Cref{def:vertex-type}, that is, $\delta_y$ is the subtree label of $y$ and $\lambda_y$ is the ancestor label of $y$; we write $\delta^s_y$ and $\lambda^s_y$ for the respective labels of $y$ in $(\RR^s,T^s)$.
	For a vertex $y$, we use $c(y)$ to denote the copy-depth of $y$: the number of times one must follow the copying bijection backwards until a retained vertex is obtained.

	\emph{Ancestor labels.}
	We show $\lambda^s_y = \lambda_{\orig(y)}$ by induction on $c(y)$.
	If $y$ is retained, its root-to-$y$ path is the same in $T^s$ as in $T$, since the set of elements inducing $\RR'$ is closed under taking ancestors, and by property~(\ref{item:mult_relations}) the vertices on it lie in the same relation tuples in both structures; hence the two paths have the same atomic type.
	Otherwise $y = \iota(y_0)$ for a copy step that attached a copy of the subtree rooted at $u$ as a further child of the parent $w$ of $u$, and we have $c(y_0) < c(y)$ and $\orig(y) = \orig(y_0)$.
	Extending $\iota$ by the identity maps the root-to-$y_0$ path onto the root-to-$y$ path: it maps the part inside the subtree of $u$ isomorphically onto the attached copy, whose root $\iota(u)$ is again a child of $w$, and it fixes the ancestors of $u$.
	By the copying rule this bijection preserves membership in every relation, and later copy steps do not change relation membership among already existing elements, so the two paths have the same atomic type and $\lambda^s_y = \lambda^s_{y_0} = \lambda_{\orig(y)}$.

	\emph{Subtree labels.}
	We show $\delta^s_y = \delta_{\orig(y)}$ together with property~(\ref{item:mult_counts}) for $y$, by induction bottom-up on the depth of $\orig(y)$ in $T$ and, at each depth, on $c(y)$; with the ancestor labels this also gives property~(\ref{item:mult_types}).
	If $\orig(y)$ is a leaf of $T$, then $y$ has no children either and both labels are $\varepsilon$.

	Let $y$ be retained and let $n_t$ be its number of children of vertex $p$-type $t$ in $T$; by induction these children keep their vertex $p$-types, and once $y$ has been processed neither they nor the copies attached to $y$ change again.
	Passing to $\RR'$ keeps exactly the $\min(n_t,s)$ of them of smallest index.
	If $n_t < p$, then $t$ is not $p$-abundant in $\vtxptype(y)$, all $n_t$ children survive and none of them is copied, so $\delta^s_y(t) = \delta_y(t) = n_t$; in particular, a vertex $p$-type absent in $T$ remains absent in $T^s$.
	If $n_t \geq p$, then $t$ is $p$-abundant in $\vtxptype(y)$ and the copy step raises the number of children of vertex $p$-type~$t$ from $\min(n_t,s)$ to exactly $s$, which it already is for $s = p$; as $s \geq p$, this gives $\delta^s_y(t) = \delta_y(t) = \star$.

	If $y$ is not retained, it is the copy $\iota(y_0)$ of some $y_0$ with $c(y_0) < c(y)$, and the copy step attaches an isomorphic copy of an already processed subtree containing $y_0$, which is not modified afterwards.
	Corresponding children of $y$ and $y_0$ have the same $\orig$-images and hence, by the induction on the depth, the same vertex $p$-types, so $y$ and $y_0$ have the same child-type counts, $\delta^s_y = \delta^s_{y_0} = \delta_{\orig(y)}$, and property~(\ref{item:mult_counts}) carries over from $y_0$ to $y$.

	Finally, for a non-root vertex $y$ of $T^s$ with parent $w$, the parent of $\orig(y)$ is $\orig(w)$, so the preservation of vertex $p$-types also yields the abundance part of property~(\ref{item:mult_types}).

	\medskip\noindent\textbf{(\ref{item:mult_routes}).}
	Since $s = p$, we have $(\RR^p, T^p) = (\RR', T')$, so no copies are added and every element of $\RR^p$ is retained, and $T^p$ arises from $T$ by removing exactly those elements that have an entry $(t,i)$ with $i > p$ in their $p$-route, together with their subtrees.
	Hence every element of $\RR^p$ keeps its parent and all of its ancestors, and its vertex $p$-type is the same in $\RR^p$ as in $\RR$ by property~(\ref{item:mult_types}).
	The route of the root is $((t,1))$ for its vertex $p$-type $t$, in both $T$ and $T^p$.
	For a non-root element, the surviving children of its parent $w$ of a vertex $p$-type $t$ are precisely the (at most $p$) smallest elements among all children of $w$ of that vertex $p$-type, so such an element is preceded by the same number of children of $w$ of its own vertex $p$-type in $T$ as in $T^p$, and thus receives the same index.
	By induction on the depth, every entry of the $p$-route of an element of $\RR^p$ is therefore unchanged.

	\medskip\noindent\textbf{(\ref{item:mult_realizing}).}
	Let $\mu$ be the $p$-route pattern in $T$ of a $z$-tuple of elements of $\RR$.
	By \Cref{lem:realizing_vertices}, which applies as $p \geq z \cdot h$, the realizing vertices of $\mu$ exist in $T$, and every entry of their $p$-routes has index at most $p$.
	Since $s = p$, an element is deleted only if its $p$-route contains an entry $(t,i)$ with $i > p$, so every realizing vertex is retained.

	\medskip\noindent\textbf{(\ref{item:mult_patterns}).}
	Every element of $\RR^p$ is retained by property~(\ref{item:mult_routes}), so every $z$-tuple of $\RR^p$ has the same $p$-route pattern in $T^p$ as in $T$ and its pattern occurs in $(\RR,T)$.
	Conversely, let $\mu$ be the $p$-route pattern in~$T$ of a $z$-tuple of elements of $\RR$.
	Its realizing vertices are retained in $\RR^p$ by property~(\ref{item:mult_realizing}), they have $p$-route pattern $\mu$ in $T$ by \Cref{lem:realizing_vertices}, and by property~(\ref{item:mult_routes}) their $p$-route pattern in $T^p$ is again $\mu$.

	\medskip\noindent\textbf{(\ref{item:mult_size}).}
	By property~(\ref{item:mult_types}), every vertex $p$-type occurring in $(\RR^s,T^s)$ also occurs in $(\RR,T)$, so by \Cref{lem:number_vertex_types} their number is bounded by $g'(h,p,\Size{\varphi})$ for a computable function $g'$; as $\Size{\varphi} \leq p$, it is bounded by the computable function $g(h,p) \coloneqq \max_{q \leq p} g'(h,p,q)$.
	By property~(\ref{item:mult_counts}), every vertex~$y$ of $T^s$ has exactly $s$ children of each vertex $p$-type that is $p$-abundant in $\vtxptype(y)$, and as many children of every other vertex $p$-type $t$ as $\orig(y)$ has in $T$, which is $\delta_{\orig(y)}(t) \leq p-1 \leq s$ by \Cref{def:vertex-type}.
	Hence $y$ has at most $B \coloneqq g(h,p)\cdot s$ children in total.
	Since $T^s$ has height $h$ by property~(\ref{item:mult_decomposition}), it has at most $\sum_{d = 0}^{h-1} B^{d} \leq (B+1)^{h}$ vertices.
	Finally, $\RR^s$ has the same signature as~$\RR$, as passing to $\RR'$ takes an induced substructure and the copy steps add no relations, and every relation of $\RR$ appears in $\varphi$ by our standing assumption that $\RR$ uses no relations outside $\varphi$.
	Hence $\RR^s$ has at most $\Size{\varphi} \leq p$ relations, each of arity at most $\Size{\varphi} \leq p$, and therefore at most $p \cdot (B+1)^{hp}$ tuples in total.
	Altogether $\Size{\RR^s}$ is bounded by a computable function of $h,p$ and $s$.
\end{proof}

\paragraph*{Pure Solutions.}
All definitions so far concerned only the input relational structure $\RR$; the formula $\varphi$ entered only through bounds on the number and the arity of the relations considered.
Now, we come to the central concept that makes the main result work, the notion of a \emph{$p$-pure solution}. Intuitively, $p$-purity says that if two element tuples behave the same with respect to the input relational structure (i.e., they share a $p$-route pattern), then they should behave the same with respect to the solution.
As we show later, this property of purity is extremely useful for our purposes; it is the main ingredient in the algorithm we provide.

For this, we first define the \emph{solution type}, which is just the atomic type of an element tuple with respect to a solution.

\begin{definition}[Solution type]
	\label{def:solution_type}
	Let $(\RR, \varphi)$ be an \snpmodelchecking{} instance with a solution $\mathcal{S} = R_1',\dots,R_r'$.
	For any $q \in \mathbb{N}$, the \emph{solution type} of a tuple $\bar v\in V(\RR)^q$ in $\mathcal{S}$ is the atomic type of $\bar v$ in the relational structure of the solution, $\soltype_{\mathcal{S}}(\bar v) = \atp_{(V(\RR), R_1', \ldots, R_r')}(\bar{v})$.
\end{definition}

Now $p$-purity simply says that the same $p$-route pattern implies the same solution type.

\defPPureSolution*

\subsection{Extending Pure Solutions}
The crucial feature of this notion of purity is that, for sufficiently large $p$, a $p$-pure solution for $\RR^p$ yields a solution for $\RR$.
This is captured in \Cref{thm:extension}.

\thmExtension*{}
\begin{proof}
	Let $\mathcal{S}^p = R_1^p,\dots,R_r^p$ be a $p$-pure solution for $\RR^p$.

	We first record that \multprop{item:mult_realizing} applies to $z$-tuples of elements for every $z \leq \Size{\varphi}$, since then $zh \leq \Size{\varphi}\cdot h \leq p$.
	As $k \leq \Size{\varphi}$ and $\arity(R_j) \leq \Size{\varphi}$ for every $j \in [r]$, the realizing vertices of the $p$-route pattern of any tuple of at most $\Size{\varphi}$ elements are retained in $\RR^p$; this covers every $k$-tuple and every $\arity(R_j)$-tuple of elements.

	We define relations $R_1',\dots,R_r'$, which will be a solution for $\RR$, as follows.
	Consider any $i\in [r]$ and $\bar v \in V(\RR)^{\arity(R_i)}$.
	Let $\mu$ be the $p$-route pattern of $\bar v$ in $T$ and let $\bar w$ be the realizing vertices of $\mu$, which exist in $\RR^p$ by the arguments above.
	We let $R_i'(\bar v)$ hold if and only if $R_i^p(\bar w)$ holds.

	For the sake of contradiction, assume that $R_1',\dots,R_r'$ do not form a solution for $\RR$.
	Then, there exists a $k$-tuple $\bar v \in V(\RR)^k$  such that $\RR \not \models \psi(R_1',\dots,R_r',\bar v)$.
	Let $\mu$ be the $p$-route pattern of $\bar v$ in $T$ for $\RR$, and $\bar w$ be the realizing vertices of $\mu$ in $T$.
	We show that $\RR^p \not \models \psi(R_1^p,\dots,R_r^p,\bar w)$, contradicting $\mathcal{S}^p$ being a solution for $\RR^p$.

	To this end, since $\psi$ is a quantifier-free boolean formula, it suffices to show that each atom of~$\psi$ evaluates to the same truth value under the assignment $\bar v$ in $\RR$ with $R_1',\dots,R_r'$ as under the assignment $\bar w$ in $\RR^p$ with $R_1^p,\dots,R_r^p$.
	Any atom in the formula $\psi$ asks, for a tuple of variables over $\{x_1, \ldots, x_k\}$, whether the elements assigned to them lie in some relation of $\RR$ or of the solution, or whether two of them are equal.
	Formally, every atom has one of the following three shapes:
	\begin{enumerate}
		\item $E_i(\bar u)$ for some relation $E_i$ of $\RR$, where $\bar u = (x_{n_1}, \ldots, x_{n_a})$ with $a = \arity(E_i)$ and $n_m \in [k]$ for all $m\in [a]$, or
		\item $R_j(\bar u)$ for some $j \in [r]$, where $\bar u = (x_{n_1}, \ldots, x_{n_a})$ with $a = \arity(R_j)$ and $n_m \in [k]$ for all $m\in [a]$, or
		\item  $x_n = x_{n'}$ with $n,n' \in [k]$.
	\end{enumerate}
	Under the assignment $\bar v$ such an atom evaluates at $(v_{n_1},\dots,v_{n_a})$, and under $\bar w$ at $(w_{n_1},\dots,w_{n_a})$.

	Consider the first case, an atom of the form $E_i(\bar u)$.
	Recall from \Cref{lem:realizing_vertices} that $\mu$ is also the $p$-route pattern of $\bar w$. Using \patprop{item:pattern_atp}, we obtain that the tuples $\bar v$ and $\bar w$ have the same atomic type in $\RR$.
	By the definition of the atomic type, we have that $E_i(v_{n_1},\dots,v_{n_a}) \Leftrightarrow E_i(w_{n_1},\dots,w_{n_a})$, for every relation $E_i$ of $\RR$ and all indices $n_1,\dots,n_a \in [k]$, where $a = \arity(E_i)$.
	By \multprop{item:mult_relations}, and since all elements of $\bar w$ persist in $\RR^p$, $E_i(w_{n_1},\dots,w_{n_a})$ holds in $\RR^p$ if and only if it holds in $\RR$.
	Thus $E_i(\bar u)$ holds under $\bar v$ in $\RR$ if and only if it holds under $\bar w$ in $\RR^p$.

	Now consider the second case, $R_j(\bar u)$, for any $j\in [r]$.
	Let $\mu'$ be the $p$-route pattern of $(v_{n_1},\dots,v_{n_a})$ in $T$ and let
	$(y_1,\dots,y_a)$ be the realizing vertices of $\mu'$ in $T$; note that these need not be the elements $(w_{n_1},\dots,w_{n_a})$, which is precisely why the $p$-purity of $\mathcal{S}^p$ is needed below.
	By the above definition of $R_j'$, we have that $R_j'(v_{n_1},\dots,v_{n_a})$ holds if and only if $R_j^p(y_1,\dots,y_a)$ holds.
	By \patprop{item:pattern_subtuple}, and since $\bar v$ and $\bar w$ have the same $p$-route pattern $\mu$, we have that $\mu'$ is also the $p$-route pattern of $(w_{n_1},\dots,w_{n_a})$.

	By the construction of $R'_j$, this gives that
	\[R_j'(v_{n_1},\dots,v_{n_a}) \Leftrightarrow R'_j(w_{n_1},\dots,w_{n_a}) \Leftrightarrow R^p_j(y_1,\dots,y_a).\]
	However, this on its own does not suffice to show our claim since without further reasoning we cannot assume that $R_j'(w_{n_1},\dots,w_{n_a})$ if and only if $R_j^p(w_{n_1},\dots,w_{n_a})$, for that step we require the $p$-purity of $\mathcal{S}^p$.

	By the arguments above, the elements of $(y_1,\dots,y_a)$ are retained in $\RR^p$, and so are those of $(w_{n_1},\dots,w_{n_a})$, as the latter occur in $\bar w$.
	Moreover, by \multprop{item:mult_routes}, elements that are retained in $\RR^p$ have identical $p$-routes in $T$ and $T^p$, and tuples of such elements have the same $p$-route pattern in $T$ and in $T^p$.
	Thus, $(y_1,\dots,y_a)$ and $(w_{n_1},\dots,w_{n_a})$ share a $p$-route pattern in $T^p$. By the $p$-purity of $\mathcal{S}^p$, these two tuples of size $a = \arity(R_j) \le \Size{\varphi}$ have the same solution type in $\mathcal{S}^p$.
	In particular, we have $ R^p_j(y_1,\dots,y_a) \Leftrightarrow R^p_j(w_{n_1},\dots,w_{n_a})$. In combination with the above equivalence, this establishes
	$R_j'(v_{n_1},\dots,v_{n_a}) \Leftrightarrow R^p_j(w_{n_1},\dots,w_{n_a})$ as desired, concluding the second case.

	Finally, consider the third case, an atom $x_n = x_{n'}$ with $n,n' \in [k]$.
	It holds under the assignment~$\bar v$ if and only if $v_n = v_{n'}$, and under $\bar w$ if and only if $w_n = w_{n'}$.
	As noted in the first case, $\bar v$ and $\bar w$ have the same atomic type in $\RR$ by \patprop{item:pattern_atp}, so its equality part in \Cref{def:atomic_type} gives $v_n = v_{n'}$ if and only if $w_n = w_{n'}$.
	Moreover, the elements of $\bar w$ are retained in $\RR^p$ and hence elements of $\RR$ by \Cref{def:retained_elements}, so equality among them means the same in $\RR$ and in $\RR^p$.
	Thus $x_n = x_{n'}$ holds under $\bar v$ in $\RR$ if and only if it holds under $\bar w$ in $\RR^p$, concluding the third case.

	Hence, every atom of $\psi$ evaluates the same under the assignment~$\bar v$ in $\RR$ as under the assignment~$\bar w$ in $\RR^p$. Thus, if $\RR \not \models \psi(R'_1,\dots,R'_r,\bar v)$ then $\RR^p \not \models \psi(R^p_1,\dots,R^p_r,\bar w)$, a contradiction to $\mathcal{S}^p$ being a solution for $\RR^p$.

	\medskip
	It remains to prove the second claim.
	By \Cref{lem:computing_types_and_routes} we can precompute the vertex $p$-types and the $p$-routes of all elements of $\RR$ in the allowed time; in particular we obtain the set of vertex $p$-types occurring in $\RR$.
	As there are boundedly many of them by \Cref{lem:number_vertex_types}, we may at the same time give each an integer identifier, store all entries of all routes as pairs of such an identifier and an index, and tabulate for which pairs $(t,t')$ of vertex $p$-types $t'$ is $p$-abundant in $t$.
	The size of each entry and its lookup time in the table is then bounded by a computable function in $h, p$, and $\Size{\varphi}$, that is, in a function of $p$ by $p\ge \Size{\varphi} \cdot h$.
	Applying \Cref{lem:number_route_patterns} to it for every $z \leq \Size{\varphi}$, we obtain in time $f(p)$ a list of at most $f(p)$ candidates that contains the $p$-route pattern of every tuple of at most $\Size{\varphi}$ elements of $\RR$.

	Consider one candidate $\mu$ and follow each of its $p$-routes $\varrho = (e_1,\dots,e_d)$ in $T^p$ as follows.
	By \Cref{def:p_route} the first entry $e_1 = (t_1,i_1)$ denotes the root of $T^p$ itself and is never abundant, so we check that $t_1$ is the vertex $p$-type of that root and that $i_1 = 1$; we then descend, for $\ell = 2,\dots,d$ in turn, at the entry $e_\ell = (t_\ell,i_\ell)$ to the $i_\ell$-th child of the current vertex of vertex $p$-type $t_\ell$.
	If some route fails this check or gets stuck, we discard $\mu$; this loses nothing, because no tuple of elements of $\RR$ has $p$-route pattern~$\mu$ in that case.
	Indeed, if some $z$-tuple with $z \leq \Size{\varphi}$ had $p$-route pattern $\mu$, then by \multprop{item:mult_patterns} some tuple of $\RR^p$ has $p$-route pattern $\mu$ in $T^p$, so the routes of $\mu$ identify vertices of $T^p$ and both the check and the descent would have succeeded.
	Moreover, a traversal which passes the check and does not get stuck ends in the realizing vertices of $\mu$, as these are by \Cref{def:realizing_vertices} precisely the elements identified by the $p$-routes of $\mu$, which by \multprop{item:mult_routes} are the same in $T^p$ as in $T$.
	For each candidate $\mu$ that survives, say a pattern of $z$ routes, and each $i\in [r]$ with $\arity(R_i) = z$, we record whether $R_i^p$ holds at its tuple of realizing vertices.
	For $\arity(R_i) \neq z$ no entry is stored, as $R_i^p$ contains no $z$-tuples.
	By \multprop{item:mult_size}, which applies as $p \geq \Size{\varphi}\cdot h \geq \Size{\varphi}$, the size of $\RR^p$ and hence that of $\mathcal{S}^p$ are bounded by a computable function of $h$ and $p$, so this preprocessing step is feasible in the stated time.

	Now consider any $i\in [r]$ and any element tuple $\bar u$ of arity $a = \arity(R_i)$.
	Using the pre-computed $p$-routes of its elements, we can determine the $p$-route pattern of $\bar u$ in $\BigO{ah \log (ah)}$ time, by identifying the at most $ah$ indices to be ranked and sorting them.
	Looking that pattern up in the table computed above, say in a trie over its entries, takes another $\BigO{ah}$ time and tells us whether $R_i^p$ holds at the realizing vertices of the pattern, which by the construction of $R_i'$ is the case if and only if $R_i'(\bar u)$ holds.
\end{proof}

Though we do not use this property in what follows, we remark that the constructed solution is itself $p$-pure, since whether a tuple is in some solution relation depends only on the $p$-route pattern of that tuple.

\subsection{Pure Solutions Exist}
\label{sec:completeness_ramsey_argument}

Now, we proceed to the key property that makes our algorithm terminate.
Concretely, we show that, for large enough $p$ (depending only on the treedepth and the formula), we can guarantee that if $\mathcal{R} \models \varphi$, then $(\mathcal{R}^p,\varphi)$ has a $p$-pure solution.
Therefore, not only do $p$-pure solutions certify that we have a yes-instance by \cref{thm:extension}, but also such a certificate is guaranteed to exist.

First, we establish what we mean by ``sufficiently large $p$''.
The concrete value depends on the size of forbidden induced substructures, and relies on well-quasi-orderings.
For this, we define the appropriate notion of obstruction sets.

\begin{definition}[Obstruction set]
	\label{def:obstruction_set}
	Let $h \geq 0$ be an integer and $\varphi$ an \SNP{} sentence.
	Let $\mathcal{S}_h$ be the set of all relational structures (with a signature that matches $\varphi$) with treedepth at most $h$.

	Set $\JJ_{h,\varphi}^0 = \{\mathcal{R}^* \in \mathcal{S}_h \mid \mathcal{R}^* \not \models \varphi\}$, that is, the set of relational structures that have treedepth at most $h$ and are not models of $\varphi$.
	Create $\JJ_{h,\varphi}$ by taking one representative of each isomorphism class of the members of $\JJ_{h,\varphi}^0$ that are minimal under taking induced substructures.
	Then, $\JJ_{h,\varphi}$ is the \emph{obstruction set} of $h, \varphi$.
\end{definition}

By~\cite[Fact 3.3.9]{lopezFirstOrderPreservation2023}, which records a consequence of the work of Ding~\cite{dingSubgraphsWellquasiordering1992}, we have that any class of finite relational structures over a fixed finite signature that has bounded treedepth is well-quasi-ordered by the induced-substructure relation (see also Lopez~\cite{lopezWhenLocalityMeets2022}).
As a consequence, the set $\JJ_{h,\varphi}$ is finite.
\begin{lemma}
	\label{thm:obstruction_bound_function}
	For any integer $h \geq 0$ and any \SNP{} sentence $\varphi$ the obstruction set $\JJ_{h,\varphi}$ is finite.
\end{lemma}
\begin{proof}
	By their minimality in \Cref{def:obstruction_set}, the members of $\JJ_{h,\varphi}$ are pairwise non-isomorphic and none of them is an induced substructure of another, so they form an antichain under the induced-substructure relation.
	They all have treedepth at most $h$ and a signature matching $\varphi$, and a well-quasi-order admits no infinite antichain, so $\JJ_{h,\varphi}$ is finite.
\end{proof}

Now, we show that for large enough $p$, if $\mathcal{R} \models \varphi$, then for all $s \geq p$ we have that $\mathcal{R}^s \models \varphi$.
This means that we can ``blow up'' the structure $\mathcal{R}$ by adding unboundedly many copies of vertices with an abundant vertex $p$-type, and the resulting structure still satisfies $\varphi$.
We give a formal proof of this fact, but note that this is easy to see intuitively because if $p$ is larger than all obstructions, then adding further copies of such abundant vertices cannot make a new obstruction appear.

Towards proving this, we show that the set of small induced substructures of $\mathcal{R}$ and $\mathcal{R}^s$ is identical.

\begin{lemma}
	\label{thm:blowup_structure_same_substructures}
	Let $\mathcal{R}$ be a relational structure with treedepth decomposition $T$, and $s \geq p \geq 1$ integers.
	Then, for any relational structure $W$ on at most $p$ elements, we have that $\mathcal{R}$ contains an induced substructure isomorphic to $W$ if and only if $\mathcal{R}^s$ contains an induced substructure isomorphic to $W$.
\end{lemma}
\begin{proof}
	We prove both directions of the lemma in two separate claims.
	\begin{claim}
		\label{claim:induced_1}
		Let $W$ be an induced substructure of $\mathcal{R}$ on at most $p$ elements.
		Then there is an induced substructure in $\mathcal{R}^{s}$ that is isomorphic to $W$.
	\end{claim}
	\begin{claimproof}
		For empty $W$ the claim is trivial, so we assume $W$ contains at least one element.
		Let $\hat W$ be the substructure of $\mathcal{R}$ induced by the ancestor-closure of the elements of $W$ in the tree~$T$.
		We show that $\mathcal{R}^{s}$ contains an induced substructure isomorphic to $\hat W$ too, which yields the claim since $W \subseteq \hat W$.

		This is proven by an explicit mapping $\pi$ of elements of $\hat W$ to elements of $\mathcal{R}^{s}$.
		The root of $\mathcal{R}$, which is part of $\hat W$, is mapped to the root of $\mathcal{R}^{s}$ by $\pi$.
		By \multprop{item:mult_types} the vertex $p$-type of the root of $\RR^s$ and $\RR$ is the same.

		The mapping of non-root elements of $\hat W$ is obtained recursively as follows.
		Let $w$ be an element of $\hat W$ with vertex $p$-type $t_w$ already mapped to some element $w'$ of $\RR^s$ with the same vertex $p$-type~$t_{w}$.
		Then, let $\mathcal{T}$ be the set of vertex $p$-types of all children of $w$ in $\hat W$.
		For each $t \in \mathcal{T}$, let $\mathcal{C}_t$ be the set of children of $w$ in $\hat W$ with type $t$, and let $\mathcal{C}_t^s$ be the set of children of $w'$ in $T^s$ with type $t$.
		We map the elements of $\mathcal{C}_t$ to $\mathcal{C}_t^s$ injectively, which concludes the description of the mapping $\pi$.

		It remains to argue that there is an injective mapping from $\mathcal{C}_t$ to $\mathcal{C}_t^s$.
		Because $w$ and $w'$ have vertex $p$-type $t_w$, we know that $w'$ has children of type $t$ too.
		If $t$ is not abundant in $t_w$, then, by \multprop{item:mult_counts}, $w$ and $w'$ have the same number of children of vertex $p$-type $t$, and therefore the injective mapping exists.
		Otherwise, $w'$ has exactly $s \geq p \geq 1$ children of vertex $p$-type $t$.
		Moreover, since $W$ has at most $p$ elements, we have that $\hat W$ contains at most $p$ siblings: each child of $w$ in~$\hat W$ has a descendant in $W$ (potentially itself).
		Since distinct children have distinct descendants, $w$ can have at most $p$ children.
		Hence, the injective mapping exists.

		Now, we show that the mapping is an isomorphism from $\hat W$ to the substructure $\hat W^s$ of $\RR^s$ induced by the image of $\pi$.
		It is clear from the description that $\pi$ is a bijection onto its image, and that it preserves depths and vertex $p$-types.
		Moreover, $\pi$ maps the children of a vertex injectively to children of its image, so vertices in distinct child subtrees have their images in distinct child subtrees; hence $v$ is an ancestor of $v'$ in $T$ if and only if $\pi(v)$ is an ancestor of $\pi(v')$ in $T^s$.
		Since the ancestor label is part of the vertex $p$-type, \Cref{lem:labels_determine_atp} applies to every tuple of elements of $\hat W$ together with its image and shows that the two have the same atomic type.
		As $\hat W$ and $\hat W^s$ are induced substructures of $\mathcal{R}$ and $\mathcal{R}^s$, respectively, this means that $\pi$ is an isomorphism between them, concluding the proof.
	\end{claimproof}

	A similar proof yields the other direction.

	\begin{claim}
		\label{claim:induced_2}
		Let $W$ be an induced substructure of $\mathcal{R}^{s}$ on at most $p$ elements.
		Then there is an induced substructure in $\mathcal{R}$ that is isomorphic to $W$.
	\end{claim}
	\begin{claimproof}
		The proof is largely analogous to the previous case.
		Use the same construction with the roles of $\RR$ and $\RR^s$ flipped.
		Suppose that $w$ of $\RR^s$ is already mapped to $w'$ of $\RR$ with the same vertex $p$-type $t_w$.
		Then, for each vertex $p$-type $t$ that children of $w$ have, note that $w'$ has the same number of children of vertex $p$-type $t$ if $t$ is not abundant in $t_w$.
		Otherwise, from the definition of abundant $p$-types, we obtain that $w'$ has at least $p$ children of vertex $p$-type $t$, which again suffices.
	\end{claimproof}
	Combining both claims completes the proof.
\end{proof}

Now, we can use this lemma to show that $\mathcal{R} \models \phi$ if and only if $\mathcal{R}^s \models \phi$, for appropriate values of $p$ and $s$.

\thmBlowupFeasibility*{}
\begin{proof}
	First, consider that $\mathcal{R} \not \models \varphi$.
	Then, $\mathcal{R}$ is isomorphic to some structure of $\JJ_{h,\varphi}^0$.
	By the definition of $\JJ_{h,\varphi}$,
	there is some structure $\mathcal{R}^\star \in \JJ_{h, \varphi}$ such that $\mathcal{R}^\star \not \models \varphi$, and $\mathcal{R}^\star$ has an induced copy in $\mathcal{R}$.
	Because $\Size{V(\mathcal R^\star)}\le M\le p$, \cref{thm:blowup_structure_same_substructures} gives an induced copy of $\mathcal R^\star$ in $\mathcal R^{s}$.
	Since \SNP{} satisfaction is closed under taking induced substructures (recall \Cref{obs:snp-satisfaction-substructure}), we obtain $\mathcal{R}^{s} \not \models \varphi$.

	Now, assume that $\mathcal{R}^{s} \not \models \varphi$.
	We have $\td(\mathcal{R}^s) \leq h$ by \multprop{item:mult_decomposition}.
	Therefore, there is some~$\mathcal{R}^\star \in \JJ_{h,\varphi}$ such that $\mathcal{R}^\star$ has an induced copy in $\mathcal{R}^{s}$.
	Since $\Size{V(\mathcal R^\star)}\le M\le p$, \cref{thm:blowup_structure_same_substructures} gives an induced copy in $\mathcal R$, proving $\mathcal R\not\models\varphi$.
\end{proof}

We use this lemma in our upcoming proof about the existence of $p$-pure solutions.
Before providing this proof, we need to define the notion of \emph{realizations of $p$-route patterns using integer tuples}, as they play a crucial role in the proof.
Intuitively, such a realization using an integer tuple just takes a $p$-route pattern $\mu$, and replaces the indices of abundant entries with those provided by the integer tuple.
For technical reasons, we define this notion for the $s$-multiplicity structure.

\begin{definition}[Realization of a $p$-route pattern using an integer tuple]
	\label{def:pattern_realization}
	Let $\mathcal{R}$ be a relational structure with treedepth decomposition $T$, and let $s \geq p \geq 1$ be integers.
	Consider the $s$-multiplicity structure $\mathcal R^s$ for $(\mathcal{R}, T, p)$ with its accompanying tree $T^s$.
	Let $\mu=(\varrho_1,\dots,\varrho_z)$ be a $p$-route pattern occurring in $(\mathcal R,T)$ or in $(\mathcal R^s,T^s)$, with $\beta$ distinct indices in its $p$-abundant entries.
	By normalization, these indices are precisely $1,\dots,\beta$.

	For a strictly increasing tuple $\bar n=(n_1,\dots,n_\beta)$ of integers in $[s]$, the \emph{realization of $\mu$ using $\bar n$} is the element tuple in $\mathcal R^s$ whose $p$-routes in $T^s$ are obtained from $\mu$ by replacing every $p$-abundant entry $(t,j)$ by $(t,n_j)$, leaving all other entries unchanged.
	When $\beta=0$, we use the empty integer tuple and make no replacements.
\end{definition}

Now, similarly to \cref{lem:realizing_vertices}, we show that such a realization exists.

\begin{lemma}
	\label{lem:pattern_realization}
	Let $\mathcal{R}$ be a relational structure with treedepth decomposition $T$, and let $s \geq p \geq 1$ be integers.
	Let $\mu=(\varrho_1,\dots,\varrho_z)$ be a $p$-route pattern appearing in $(\mathcal R,T)$ or in $(\mathcal R^s,T^s)$ with $\beta \leq s$ distinct indices in its $p$-abundant entries.
	Moreover, let $\bar n \in [s]^\beta$ be a vector with strictly increasing entries.
	The realization of $\mu$ using $\bar n$ exists and is unique, and its $p$-route pattern is $\mu$.
	Conversely, every element tuple in $\mathcal R^s$ with $p$-route pattern $\mu$ is the realization of $\mu$ using a unique strictly increasing $\beta$-tuple in $[s]$.
\end{lemma}
\begin{proof}
	For the first assertion, replace each abundant index $j$ in $\mu$ by $n_j$ and follow the resulting routes in $T^s$.
	The root entry remains unchanged and identifies the root of $T^s$, whose vertex $p$-type agrees with that of the root of $T$ by \multprop{item:mult_types}.
	At each subsequent step, the current parent vertex has the vertex $p$-type specified by the preceding entry.
	If the current entry has an abundant vertex $p$-type, the parent vertex has exactly $s$ children of the required vertex $p$-type, and the substituted index is at most $s$.
	For a non-abundant vertex $p$-type, the parent vertex $p$-type records the exact number of children of the required type, and the unchanged index remains valid.
	This constructs a unique vertex for each route.
	Therefore, the realization is unique.

	Normalizing the resulting tuple again replaces $n_j$ by $j$ in every abundant entry, since $\bar n$ is strictly increasing, and hence recovers $\mu$.
	Thus the realizing vertex tuple has $p$-route pattern $\mu$.

	Conversely, the increasing list of distinct abundant indices in any element tuple with pattern $\mu$ has length $\beta$ and is its unique substitution tuple.
	Each index of an abundant entry is at most $s$ by the definition of $s$-multiplicity structures.
\end{proof}

Since later arguments build on it, we now define the standard notion of the finite hypergraph Ramsey number.

\begin{definition}[Hypergraph ramsey number]
	\label{def:hypergraph_ramsey}
	For any integers $x,y,c \geq 1$, with $x \geq y$, let $\ramsey_y(x,c)$ be the smallest integer $N$ such that, for any $y$-uniform hyperclique $H$ on at least $N$ vertices, where the edges are colored using at most $c$ colors, there exists a $y$-uniform hyperclique in~$H$ on $x$ vertices whose edges all have the same color.
\end{definition}

The Ramsey theorem for finite hypergraphs guarantees that this number $\ramsey_y(x,c)$ indeed exists~\cite[Theorem B]{ramseyProblemFormalLogic1930}.
Finally, we have all ingredients to prove the existence of a $p$-pure solution.

\thmExistenceOfPureSolution*{}
\begin{proof}
	We begin by sketching the idea of the proof.

	\subparagraph*{Proof idea.}
	By assumption, we have that $\mathcal{R}$ has a solution.
	Then, from \cref{thm:blowup_preserves_feasibility} we obtain that $\RR^p$ has a solution too.
	Our aim is showing that indeed $\RR^p$ has a $p$-pure solution.
	Recall that the main property of such a solution is that, for $p$-route patterns of bounded length, all element tuples with that $p$-route pattern have the same solution type.

	For a single $p$-route pattern, say $\mu$, we seek a set of indices on which all realizations of $\mu$ have the same solution type.
	Keeping only these indices in every abundant child-type class gives a structure $\widetilde{\mathcal R}$.
	Then, we could project our solution down to $\widetilde{\mathcal{R}}$, and would find that $p$-route pattern~$\mu$ is always mapped to the same solution type, as desired.
	Of course, $\widetilde{\mathcal{R}}$ might not contain the same $p$-route patterns as $\mathcal{R}$: taking substructures can change which vertices are $p$-abundant, completely changing the picture.
	However, if every abundant child-type class still has at least $p$ children, the vertex $p$-types remain unchanged.
	Ramsey theory will provide such a common set of indices while retaining sufficiently many children in every abundant class.
	We use nested Ramsey arguments on the structure $\mathcal{R}^{p_\alpha}$ for a large $p_\alpha$ to ensure this works for many patterns.

	\subparagraph*{The blowup structure.}
	First, we determine the value of integer $p_\alpha$.
	Let $\mu_1,\dots,\mu_\alpha$ be the list of all $p$-route patterns of size up to $\Size{\varphi}$ appearing in $\mathcal{R}$, in an arbitrary order.
	Moreover, let $\phi = \exists R_1,\dots,R_r \,\forall x_1,\dots,x_k \colon \psi$.

	A $p$-route pattern of length $z\le\Size{\varphi}$ uses at most $z(h-1)\le p$ distinct indices in abundant entries.
	By definition, every index of an abundant entry in a $p$-route pattern is therefore at most $p$.
	Moreover, by \multprop{item:mult_patterns}, which applies as $p \geq \Size{\varphi} \cdot h \geq zh$ for every $z \leq \Size{\varphi}$, the sets of $p$-route patterns of tuples of length at most $\Size{\varphi}$ appearing in $(\mathcal R,T)$ and in $(\mathcal R^p,T^p)$ are identical.

	For each $\mu_i$, let $\beta_i$ be the number of distinct indices of abundant entries of $\mu_i$ (such that a realization of $\mu_i$ requires an integer tuple of length $\beta_i$).
	Moreover, set
    $
        D=2^{\Size{\varphi}^2+\sum_{j=1}^r \Size{\varphi}^{\arity(R_j)}}.
    $
    By \cref{lem:atomic_type_bound} and \cref{def:solution_type} the number of possible solution types for tuples realizing $\mu_i$ is bounded by $D$ in any possible solution.

	Define $p_0 = p$.
	Inductively define
	\begin{align*}
		p_i = \begin{cases}
			\ramsey_{\beta_i}(p_{i-1},D) & \text{if $\beta_i > 0$},\\
			p_{i-1} & \text{otherwise,}
		\end{cases}
	\end{align*}
	for each $i \in [\alpha]$, where $\ramsey_{\beta_i}(p_{i-1},D)$ denotes the finite hypergraph Ramsey number of \cref{def:hypergraph_ramsey}.
	Since $\beta_i\le p\le p_{i-1}$, the recurrence is well-defined and $p_i\ge p_{i-1}\ge p$.

	\subparagraph*{Obtaining the $\boldsymbol{p}$-pure solution.}
	Since $\beta_i\le p\le p_\alpha$ for every $i\in[\alpha]$, \Cref{lem:pattern_realization} gives that each $p$-route pattern $\mu_i$ also occurs in $(\mathcal R^{p_\alpha},T^{p_\alpha})$.
	Fix a solution $\mathcal S$ for the structure $\mathcal R^{p_\alpha}$, which exists by \cref{thm:blowup_preserves_feasibility}.
	We prove by induction that for each $i \in [0,\alpha]$ there exists a set $N_i\subseteq[p_\alpha]$ of size $p_i$, such that all realizations of each of the $p$-route patterns $\mu_\alpha,\dots,\mu_{i+1}$ using integers in $N_i$ have the same solution type in $\mathcal S$.

	In the base case $i = \alpha$, the conditions are vacuous, and we set $N_\alpha=[p_\alpha]$.

	For the induction step, we assume the statement holds for index $i \geq 1$, and we show it holds for $i-1$ too.
	Let $N_i$ be the set provided by the induction hypothesis.

	If $\beta_i = 0$, then there is only a single element tuple with $p$-route pattern $\mu_i$ in $\mathcal{R}^{p_\alpha}$, and therefore $\mu_i$ is vacuously $p$-pure in $\mathcal S$.
	In this case, we additionally have $p_i = p_{i-1}$, and setting $N_{i-1}=N_i$ suffices.

	Otherwise, we have $\beta_i > 0$.
	We now build an edge-colored $\beta_i$-uniform hyperclique $H_i$ for our application of Ramsey theory.
	The vertex set of the hypergraph is the integer set $N_i$.
	For each tuple $\bar n = n_1,\dots,n_{\beta_i}$ of $\beta_i$ distinct vertices of the graph, ordered increasingly, we color the corresponding hyperedge (that is, the hyperedge whose elements are those of $\bar n$) as follows.
	Let $\bar v$ be the realization of $\mu_i$ using $\bar n$ in $\mathcal{R}^{p_\alpha}$, as given by \Cref{lem:pattern_realization} with $s=p_\alpha$.
	Then, the tuple $\bar v$ has some solution type~$t$ in $\mathcal S$.
	We color the hyperedge on $\bar n$ with color $t$.

	Hypergraph Ramsey yields that there is a $\beta_i$-uniform hyperclique on exactly $p_{i-1}$ vertices in~$H_i$ whose edges all have the same color.
	Let $N_{i-1}$ be the vertex set of this $\beta_i$-uniform hyperclique in $H_i$.
	We claim that all realizations of each of $\mu_\alpha,\dots,\mu_i$ using integers in $N_{i-1}$ have the same solution type.

	That this holds for $\mu_\alpha,\ldots,\mu_{i+1}$ follows from the induction hypothesis, since $N_{i-1}\subseteq N_i$.
	For $\mu_i$, note the following: by \Cref{lem:pattern_realization}, each considered vertex tuple with $p$-route pattern $\mu_i$ is uniquely identified by the increasing list $\bar n$ of its distinct abundant indices.
	This list is a $\beta_i$-tuple of $N_{i-1}$.
	Therefore, from the fact that we have a monochromatic hyperclique on $N_{i-1}$, we obtain that these all have the same solution type.

	At $i=0$, we obtain a set $N_0$ of size $p$; write $N_0=\{n_1<\cdots<n_p\}$.
	Recall that in $\mathcal{R}^{p_\alpha}$, for each vertex $p$-type $t$ that is abundant in the vertex $p$-type of its parent, the parent vertex has exactly $p_\alpha$ children of type $t$.
	In $\mathcal{R}^p$, every such class has exactly $p$ children; non-abundant child-type classes at parents of the same vertex $p$-type have the same size in both structures.

	We now define an induced embedding $\pi$ from $\RR^p$ to $\RR^{p_\alpha}$.
	Concretely, $\pi$ maps the root of $\RR^p$ to the root of $\RR^{p_\alpha}$, and for each vertex $w$ already mapped to $\pi(w)$, we map its children as follows.
	If $t$ is abundant in the vertex $p$-type of $w$, we map the $j$th smallest child of $w$ with vertex $p$-type $t$ to the $n_j$th smallest child of $\pi(w)$ with vertex $p$-type $t$.
	Otherwise, we map the $j$th smallest child of~$w$ with vertex $p$-type $t$ to the $j$th smallest child of $\pi(w)$ with vertex $p$-type $t$.
	The roots have the same vertex $p$-type by \multprop{item:mult_types}, and the child counts above make $\pi$ a well-defined injection preserving depths, vertex $p$-types, and the ancestor relation in both directions.
	Since the ancestor label is part of the vertex $p$-type, \Cref{lem:labels_determine_atp} then gives that every tuple of elements of $\RR^p$ has the same atomic type as its image, that is, $\pi$ is an induced embedding.

	Write $\mathcal S=R_1',\dots,R_r'$.
	Create $\mathcal{S}^p=R_1^p,\dots,R_r^p$ by setting
	\[
		R_j^p(v_1,\dots,v_a) \quad\Longleftrightarrow\quad R_j'(\pi(v_1),\dots,\pi(v_a))
	\]
	for each $j\in[r]$, with $a=\arity(R_j)$, and each tuple $(v_1,\dots,v_a)\in V(\RR^p)^a$.

	Consider any assignment of vertices $v_1,\dots,v_k$ of $\RR^p$ to the variables $x_1,\dots,x_k$, and compare it with the assignment $\pi(v_1),\dots,\pi(v_k)$ in $\RR^{p_\alpha}$.
	For any relation $E_i$ of the input relational structures, we have that $E_i(v_1,\dots,v_z)$ holds in $\mathcal{R}^{p}$ if and only if $E_i(\pi(v_1),\dots,\pi(v_z))$ holds in $\RR^{p_\alpha}$, since $\pi$ is an induced embedding.
	The same holds for the equality relation: $v_1$ and $v_2$ are identical if and only if $\pi(v_1)$ and $\pi(v_2)$ are.
	Finally, for any relation $R_j^p$ of the solution, by construction $R_j^p(v_1,\dots,v_a) \Leftrightarrow R_j'(\pi(v_1),\dots,\pi(v_a))$.
	Therefore, the assignment $v_1,\dots,v_k$ fulfills $\psi$ if and only if the assignment $\pi(v_1),\dots,\pi(v_k)$ does.
	Hence, we have that $\mathcal S^p$ is a solution for $\RR^p$.

	We claim that $\mathcal{S}^p$ is a $p$-pure solution for $\RR^p$.
	Consider any element $v$ of $\RR^p$.
	By construction, the $p$-route of $\pi(v)$ is obtained from that of $v$ by replacing each abundant entry $(t,j)$ by $(t,n_j)$, leaving all other entries unchanged.

	Now, consider any element tuple $v_1,\dots,v_z$ in $\RR^p$ of length at most $\Size{\varphi}$.
	By the argument above, $\pi(v_1),\dots,\pi(v_z)$ have the same atomic type as $v_1,\dots,v_z$.
	Since mapping $j$ to $n_j$ is a strictly increasing mapping, and this mapping is used in every abundant entry, the $p$-route pattern of $v_1,\dots,v_z$ is the same as the $p$-route pattern of $\pi(v_1),\dots,\pi(v_z)$.
	Moreover, all abundant indices of the $p$-routes of the image tuple lie in $N_0$.

	For any $\mu_i$, by the Ramsey argument above, all element tuples of $\mathcal{R}^{p_\alpha}$ with $p$-route pattern $\mu_i$ where all abundant indices lie in $N_0$ have the same solution type in $\mathcal{S}$.
	By our arguments earlier, all element tuples $v_1,\dots,v_z$ with $p$-route pattern $\mu_i$ in $\mathcal{R}^p$ have image tuples $\pi(v_1),\dots,\pi(v_z)$ with the same $p$-route pattern, such that all abundant indices of their $p$-routes lie in $N_0$.
	Hence, all these image tuples have the same solution type in $\mathcal{S}$.
	Now, observe that the solution type of $v_1,\dots,v_z$ in $\mathcal{S}^p$ and of $\pi(v_1),\dots,\pi(v_z)$ in $\mathcal{S}$ is also the same by construction of $\mathcal S^p$ and injectivity of $\pi$.
	Therefore, we have that $\mathcal{S}^p$ is $p$-pure.
\end{proof}

\subsection{Combining Everything: The Algorithm}
\label{sec:algorithm}
With \Cref{thm:extension,thm:existenceOfPureSolution} at hand, we now have two main ingredients for establishing our main theorem.

\mainTheorem*

We proceed by describing (the decision variant of) the algorithm and establishing its properties.
These properties then culminate in the proof of the main theorem.

\begin{algorithm}
	\label{alg:main-algorithm}
	Let $(\mathcal{R},\varphi = \exists R_1,\dots,R_r \,\forall x_1,\dots,x_k \colon \psi)$ be an instance of \snpmodelchecking{}.
	Compute the primal graph of \RR\ and a treedepth decomposition $T$ of it of height $h \leq \td(\RR) + 1$ using the algorithm by Reidl, Rossmanith, S{\'{a}}nchez Villaamil and Sikdar~\cite{ReidlRVS14}.
	Consider iteratively $j= 0,1,2, \ldots$ and let $p = j + \Size{\varphi} \cdot h$.
	In each iteration, compute the $p$-multiplicity structure $\RR^p$ for $(\RR, T, p)$ as well as the set $\mathcal{L}^p$ of all solutions of $\RR^p$.
	Consider the following cases in order:
	\begin{enumerate}
		\item if $\mathcal{L}^p = \emptyset$, we output $\RR \not\models \varphi$, i.e., we face a negative instance,\label{alg:main-algorithm-no-case}
		\item if $\mathcal{L}^p$ contains a $p$-pure solution $\mathcal{S}^p$ for $\RR^p$ and $T^p$, we retain $\mathcal{S}^p$ and output $\RR \models \varphi$, i.e., we face a positive instance,\label{alg:main-algorithm-yes-case}
		\item otherwise, proceed to the next iteration.\label{alg:main-algorithm-otherwise-case}\noindent\hfill{\small\textcolor{black}{\textbf{End of Algorithm.}}}\par
	\end{enumerate}
\end{algorithm}

While \Cref{alg:main-algorithm} is easy to state, it is a priori not clear that it even terminates.
Thus, we first show in \Cref{lem:algorithm-termination-correctness} that it does so and establish its correctness.

\begin{lemma}
	\label{lem:algorithm-termination-correctness}
	Let $(\mathcal{R},\varphi = \exists R_1,\dots,R_r \,\forall x_1,\dots,x_k \colon \psi)$ be an instance of \snpmodelchecking{}.
	Moreover, let $M = \max\bigl(\{1\}\cup\{\Size{V(H)}\colon H\in\JJ_{\td(\RR)+1,\varphi}\}\bigr)$, where $\JJ_{\td(\RR)+1,\varphi}$ is the obstruction set from \Cref{def:obstruction_set}.
	Then, \Cref{alg:main-algorithm} has the following properties:
	\begin{enumerate}
		\item If it enters Case~\ref{alg:main-algorithm-no-case}, then $\RR \not\models \varphi$.
		\item If it enters Case~\ref{alg:main-algorithm-yes-case}, then $\RR \models \varphi$.
		\item In the iteration for $p = \max\{M,\Size{\varphi} (\td(\RR)+1)\}$ (if reached), it does not enter Case~\ref{alg:main-algorithm-otherwise-case}.
	\end{enumerate}
	In particular, \Cref{alg:main-algorithm} terminates in an iteration with $p \leq \max\{M,\Size{\varphi} (\td(\RR)+1)\}$ and correctly determines whether $\mathcal{R} \models \varphi$.
\end{lemma}
\begin{proof}
	Let $T$ be the treedepth decomposition of $\RR$ of height $h \leq \td(\RR) + 1$ computed in \Cref{alg:main-algorithm}.
	We show the three properties separately.

	First, assume that \Cref{alg:main-algorithm} enters Case~\ref{alg:main-algorithm-no-case}.
	This means that for the $p$ that we consider in that iteration, the $p$-multiplicity structure $\RR^p$ for $(\RR, T, p)$ has no solution.
	To see that this implies $\RR \not\models \varphi$, assume otherwise.
	Since $\RR^p$ is an induced substructure of $\RR$ and \SNP-satisfaction is closed under taking induced substructures (recall \Cref{obs:snp-satisfaction-substructure}), we get $\RR^p \models \varphi$.
	However, this contradicts the precondition of Case~\ref{alg:main-algorithm-no-case}, i.e., that the set of solutions $\mathcal{L}^p$ is empty.
	Hence, $\RR \models \varphi$ cannot hold.

	Second, assume that \Cref{alg:main-algorithm} enters Case~\ref{alg:main-algorithm-yes-case}.
	This means that the solution set $\mathcal{L}^p$ contains a $p$-pure solution $\mathcal{S}^p$ for $\RR^p$ and $T^p$.
	\Cref{thm:extension} together with $p \geq \Size{\varphi}\cdot h$ guarantees that we can extend the $p$-pure solution $\mathcal{S}^p$ to a solution for $\RR$, which implies $\RR\models\varphi$.

	Third, consider the iteration for $p = \max\{M,\Size{\varphi}(\td(\RR) + 1)\}$.
	We now argue that \Cref{alg:main-algorithm} either enters Case~\ref{alg:main-algorithm-no-case} or Case~\ref{alg:main-algorithm-yes-case}.
	If $\RR \not\models \varphi$, by \Cref{thm:blowup_preserves_feasibility}, we have that $\RR^p\not\models\varphi$ since $p \geq M$.
	Hence, $\mathcal{L}^p = \emptyset$ and \Cref{alg:main-algorithm} enters Case~\ref{alg:main-algorithm-no-case}.
	If $\RR \models \varphi$, since we have $p = \max\{M,\Size{\varphi}(\td(\RR) + 1)\} \geq \max\{M,\Size{\varphi}h\}$, \Cref{thm:existenceOfPureSolution} ensures that $\RR^p$ has a $p$-pure solution $\mathcal{S}^p$ and \Cref{alg:main-algorithm} enters Case~\ref{alg:main-algorithm-yes-case}.
	In both cases \Cref{alg:main-algorithm} does not enter Case~\ref{alg:main-algorithm-otherwise-case}.

	To combine everything, observe that \Cref{alg:main-algorithm} iteratively considers every $p \geq \Size{\varphi} \cdot (\td(\RR) + 1)$.
	Since $\max\{M,\Size{\varphi}(\td(\RR) + 1)\} \geq \Size{\varphi} \cdot (\td(\RR) + 1)$, the algorithm terminates at the latest for $p = \max\{M,\Size{\varphi}(\td(\RR) + 1)\}$ by entering either Case~\ref{alg:main-algorithm-no-case} or Case~\ref{alg:main-algorithm-yes-case}.
	Moreover, $\RR^p$ has bounded size by \Cref{lem:multiplicity_properties} and thus we can brute-force candidate solutions.
	By the first two properties, we obtain correctness of \Cref{alg:main-algorithm}.
\end{proof}

While \Cref{lem:algorithm-termination-correctness} guarantees correctness \emph{and} termination of \Cref{alg:main-algorithm}, it alone does not yet show that the running time of our algorithm can be bounded by a computable function in the parameters $\td(\RR)$ and $\Size{\varphi}$.
Indeed, the termination property in \Cref{lem:algorithm-termination-correctness} relies on the size $M$ of the largest element in the obstruction set $\JJ_{\td(\RR) + 1,\varphi}$.
We only know that $\JJ_{\td(\RR) + 1,\varphi}$ is finite; recall \Cref{thm:obstruction_bound_function}.
Moreover, the well-quasi-ordering argument only implies that $M$ exists but provides no bound on $M$ in terms of $\td(\RR)$ and $\Size{\varphi}$.
We now show that we can actually bound $M$ by a computable function in the parameters.

\theoremObstructionSetComputable*
\begin{proof}
	We present a procedure $P_{\varphi}$ to compute $\JJ_{h,\varphi}$.
	On a high level, $P_{\varphi}$ emulates the behavior of \Cref{alg:main-algorithm} when given a relational structure $\RR$ with $\td(\RR) \leq h$ and $\varphi$.
	We then aim to find a value $p_{\varphi}$ for $p$ such that we enter Case~\ref{alg:main-algorithm-no-case} or Case~\ref{alg:main-algorithm-yes-case} for \emph{every} structure that can arise in the corresponding iteration of \Cref{alg:main-algorithm}.
	The negative instances witnessed at $p_{\varphi}$ can then be used to construct $\JJ_{h,\varphi}$.

	We now describe the procedure $P_{\varphi}$.
	We start with $p = \Size{\varphi} \cdot (h+1)$ and enumerate every pair $(\mathcal{X}, T_{\mathcal{X}})$ of a relational structure $\mathcal{X}$ with an accompanying treedepth decomposition $T_{\mathcal{X}}$ where $\mathcal{X}$ matches the signature of $\varphi$ such that $\mathcal{X}^p = \mathcal{X}$, $\td(\mathcal{X}) \leq h$, and $T_{\mathcal{X}}$ has height at most $h+1$.
	Recall that the number of distinct vertex $p$-types can be bounded by some computable function in $h$, $p$, and $\Size{\varphi}$ (\Cref{lem:number_vertex_types}).
	Since $\mathcal{X}^p = \mathcal{X}$, every vertex in $T_{\mathcal{X}}$ has at most $p$ children of every possible vertex $p$-type, which allows us to bound the degree of every vertex in $T_{\mathcal{X}}$ by some computable function in these values.
	Since the height of $T_{\mathcal{X}}$ is at most $h+1$, this ensures that $\Size{T_{\mathcal{X}}}$ can be bounded by a computable function in $h$, $p$, and $\Size{\varphi}$.
	Requiring $\mathcal{X}^p = \mathcal{X}$ bounds not only its size by a computable function $g'(h, p)$ (recall \Cref{lem:multiplicity_properties}) but also the number of candidate relational structures $\mathcal{X}$ by a computable function in $h$, $p$, $\Size{\varphi}$.
	Thus, we can enumerate all pairs (where the domain of $\mathcal{X}$ is a subset of $[g'(h, p)]$).
	Consider one of these pairs $(\mathcal{X}, T_{\mathcal{X}})$.
	Since $\Size{\mathcal{X}^p} = \Size{\mathcal{X}} \leq g'(h,p)$, we can brute force all solutions $\mathcal{L}_{\mathcal{X}}$ of $\mathcal{X}$.
	If $\mathcal{L}_{\mathcal{X}} \neq \emptyset$ but contains no $p$-pure solution, we increase $p$ by one and restart the enumeration with the larger $p$.
	Otherwise, if for \emph{every} pair $(\mathcal{X}, T_{\mathcal{X}})$ we have that either $\mathcal{L}_{\mathcal{X}} = \emptyset$ or $\mathcal{L}_{\mathcal{X}}$ contains a $p$-pure solution, we terminate.

	We now argue that $P_{\varphi}$ terminates.
	To this end, observe that the obstruction set $\JJ_{h+1,\varphi}$ only depends on $h$ and $\varphi$ but not on $\mathcal{X}$; recall \Cref{def:obstruction_set}.
	Moreover, $\JJ_{h+1,\varphi}$ is finite by \Cref{thm:obstruction_bound_function}.
	Furthermore, note that $\JJ_{h', \varphi} \subseteq \JJ_{h+1, \varphi}$ holds for all $h' \leq h+1$ (up to isomorphism) by the definition of the obstruction set.
	Hence, there exists some (unknown) value for $p$ such that every relational structure $\mathcal{X} = \mathcal{X}^p$ with $\mathcal{X} \models \varphi$ that we consider in this iteration has a $p$-pure solution (for any $p$ that satisfies the bound from \Cref{thm:existenceOfPureSolution} each considered pair $(\mathcal{X}, T_{\mathcal{X}})$ with $\mathcal{X} \models \varphi$ has a $p$-pure solution).
	All the other relational structures are no-instances, i.e., $\mathcal{L}_{\mathcal{X}} = \emptyset$.
	Consequently, $P_{\varphi}$ terminates eventually.
	We let $p_{\varphi}$ denote the value of $p$ when $P_{\varphi}$ terminates.

	We now use $p_{\varphi}$ to construct $\JJ_{h,\varphi}$.
	To this end, we consider every pair $(\mathcal{X}, T_{\mathcal{X}})$ with $\td(\mathcal{X}) \leq h$, where $T_{\mathcal{X}}$  has height at most $h+1$ such that $\mathcal{X}^{p_{\varphi}} = \mathcal{X}$ and $\mathcal{X} \not\models \varphi$.
	Let $\mathcal{Q}'$ be the set of these structures.
	Consider only the members of $\mathcal{Q}'$ for which every \emph{proper} induced substructure is a model of $\varphi$ (which can be checked since elements of $\mathcal{Q}'$ have bounded size), choose one representative from each such isomorphism class and call the set of these representatives $\mathcal{Q}$.
	We claim that $\mathcal{Q} = \JJ_{h, \varphi}$ (up to isomorphism).

	Note that $\mathcal{Q} \subseteq \JJ_{h, \varphi}$ holds by the definition of $\mathcal{Q}$.
	To see that also $\JJ_{h, \varphi} \subseteq \mathcal{Q}$ holds, take an arbitrary obstruction $\mathcal{Y} \in \JJ_{h, \varphi}$ and fix a corresponding treedepth decomposition $T_{\mathcal{Y}}$ of height at most $h+1$.
	By \Cref{def:obstruction_set} we have $\td(\mathcal{Y}) \leq h$ and $\mathcal{Y} \not\models \varphi$, and $\mathcal{Y}$ is minimal under taking induced substructures; as induced substructures do not increase the treedepth, the latter means that every \emph{proper} induced substructure of $\mathcal{Y}$ is a model of $\varphi$.
	Assume $\mathcal{Y}^{p_{\varphi}} \neq \mathcal{Y}$.
	Thus, $\mathcal{Y}^{p_{\varphi}}$ is a proper induced substructure and by the minimality-definition of $\JJ_{h, \varphi}$, we have $\mathcal{Y}^{p_{\varphi}} \models \varphi$.
	Since $P_{\varphi}$ enumerates every pair $(\mathcal{X}, T_{\mathcal{X}})$ as described above, it also enumerates a pair containing the structure $\mathcal{Y}^{p_{\varphi}}$.
	Indeed, $\mathcal{Y}^{p_{\varphi}}$ matches the signature of $\varphi$ and satisfies $\td(\mathcal{Y}^{p_{\varphi}}) \leq h$, both because it is an induced substructure of $\mathcal{Y}$; its decomposition $T_{\mathcal{Y}}^{p_{\varphi}}$ has height at most $h+1$ by \multprop{item:mult_decomposition}; and $(\mathcal{Y}^{p_{\varphi}})^{p_{\varphi}} = \mathcal{Y}^{p_{\varphi}}$, as for $s = p$ the construction in \Cref{def:s_multiplicity_structure} adds no copies and, by \multprop{item:mult_routes}, its deletion step removes no element of $\mathcal{Y}^{p_{\varphi}}$.
	As $\mathcal{Y}^{p_{\varphi}} \models \varphi$, it must have a $p_{\varphi}$-pure solution (since $P_{\varphi}$ terminated) whose existence is preserved under (order- and decomposition-preserving) isomorphisms.
	By \Cref{thm:extension}, we can use the  $p_{\varphi}$-pure solution of $\mathcal{Y}^{p_{\varphi}}$ together with $p_{\varphi} \geq \Size{\varphi}(h+1)$ to derive $\mathcal{Y} \models \varphi$, a contradiction to $\mathcal{Y} \in \JJ_{h, \varphi}$.
	Thus, $\mathcal{Y}^{p_{\varphi}} = \mathcal{Y}$ and hence $\mathcal{Y} \in \mathcal{Q}$ up to isomorphism, since $\mathcal{Y}$ is a minimal obstruction.
	Both directions allow us to conclude $\mathcal{Q} = \JJ_{h, \varphi}$.
	Since $P_{\varphi}$ terminates, we can bound its number of iterations, and thus also running time, by a computable function $t'_{\varphi}(h)$.
	We can run the above procedure for every possible formula $\varphi'$ of length $\Size{\varphi'} = \Size{\varphi}$.
	Up to logical equivalence and renaming of variables and relation symbols, the number of such formulas can be bounded by a computable function $q(\Size{\varphi})$ and all of them can be enumerated.
	Thus, the function $f(h, \Size{\varphi}) = \max_{\varphi',\ \Size{\varphi'} = \Size{\varphi}} t'_{\varphi'}(h)$ is computable and an upper bound on the running time of $P_{\varphi}$ in terms of $h$ and $\Size{\varphi}$.
	Since $\Size{\mathcal{X}^{p_{\varphi}}} \leq g'(h, p_{\varphi})$ for some computable $g'$ (recall \Cref{lem:multiplicity_properties}), every obstruction has also bounded size.
	The value of $g(h, \Size{\varphi})$ can be obtained by taking the maximum value of $g'(h, p_{\varphi'})$ over all formulas $\varphi'$ of length $\Size{\varphi'} = \Size{\varphi}$, which is computable.
\end{proof}

Observe that in terms of the elements of the obstruction set $\JJ_{h, \varphi}$, \Cref{thm:obstruction-max-size-computable} only shows that their size can be bounded by a computable function in $h$ and $\Size{\varphi}$, but does not give an explicit bound on the value of $g(h, \Size{\varphi})$.
Still, \Cref{thm:obstruction-max-size-computable} is the last missing ingredient to establish our main theorem.

\begin{proof}[Proof of Main Theorem]
	We first focus on the decision variant of the main theorem and then discuss how to make the algorithm constructive and the impact this has on the running time.
	To this end, we use \Cref{alg:main-algorithm} to solve an instance $(\RR, \varphi)$ of \snpmodelchecking.
	Correctness of the algorithm follows from \Cref{lem:algorithm-termination-correctness}, and it remains to argue its running time.

	\subparagraph*{Decision Variant.}
	For bounding the running time of the decision algorithm, we start by noting that \Cref{alg:main-algorithm} has two preprocessing steps.
	First, it has to compute the primal graph $G$ of $\RR$; recall \Cref{def:relational_structures_primal_graph}.
	Let $\mathcal{R} = (V, E_1,\dots,E_b)$.
	The vertex set of $G$ is $V$ and to compute the edge set of $G$, we consider every relation $E_i$, $i \in [b]$, and its tuples $\overline{u} \in E_i$.
	We introduce in $G$ a clique on the (distinct) elements in $\overline{u}$.
	It is well-known that the treedepth of $K_{\eta}$ is $\eta$ and treedepth does not increase when taking subgraphs.
	Thus, $\overline{u}$ consists of at most $\td(\RR)$ distinct elements.
	Hence, we can scan all tuples of $\RR$ once to identify their set of pairwise distinct elements.
	Afterwards, introducing a clique for $\overline{u}$ takes $\BigO{\td(\RR)^2}$ time.
	Overall, this takes $\BigO{\td(\RR)^2 \cdot \Size{\RR}}$ time, which is also an upper bound on the number of edges of the created primal graph $G$.
	In the same time, we can once more iterate over all edges of $G$ and remove any multi-edges to ensure that $G$ is simple.

	Once we have the primal graph $G$, \Cref{alg:main-algorithm} computes a treedepth decomposition $T$ of~$G$ of height $h\leq\td(\RR) + 1$.
	Using the algorithm by Reidl, Rossmanith, S{\'{a}}nchez Villaamil and Sikdar~\cite{ReidlRVS14}, this can be done in $2^{\BigO{\td(\RR)^2}}\cdot \Size{V(G)} \leq 2^{\BigO{\td(\RR)^2}}\cdot \Size{\RR}$ time.
	The latter upper bound is also an upper bound on the preprocessing steps of \Cref{alg:main-algorithm}.
	Next, we analyze the running time of the iteration for one fixed $p$.
	First, we have to compute the $p$-multiplicity structure~$\RR^p$ for $p$, $T$, and $\RR$.
	To this end, we compute the vertex $p$-types and the $p$-routes of all elements of~\RR.
	By \Cref{lem:computing_types_and_routes}, this can be done in time $g_1(h, p, \Size{\varphi}) \cdot \Size{\RR}$ for some computable function $g_1$.
	Once we have the $p$-routes, we can compute the $p$-multiplicity structure $\RR^p$ by iterating once over the elements and once over the relations in $\RR$; recall \Cref{def:s_multiplicity_structure}.
	Thus, we can compute $\RR^p$ in time $g_2(h,p,\Size{\varphi}) \cdot \Size{\RR}$ for some computable function $g_2$.
	The remaining steps in the iteration for $p$ consist of obtaining the solution set $\mathcal{L}^p$ and checking whether it is empty or contains a $p$-pure solution $\mathcal{S}^p$.
	Since \Cref{lem:multiplicity_properties} bounds the size of $\RR^p$ by a computable function in $h$ and $p$, we can upper bound the running time for the entire iteration by some computable function $g_3(\td(\RR) + 1, \Size{\varphi}, p) \cdot \Size{\RR}$.

	By \Cref{lem:algorithm-termination-correctness}, \Cref{alg:main-algorithm} terminates after at most $\max\{M,\Size{\varphi} (\td(\RR) + 1)\}$ iterations, where $M = \max\bigl(\{1\}\cup\{\Size{V(H)}\colon H\in\JJ_{\td(\RR) + 1,\varphi}\}\bigr)$ is the size of the largest element in the obstruction set $\JJ_{\td(\RR) + 1,\varphi}$.
	By \Cref{thm:obstruction-max-size-computable}, we can bound $M$ by a computable function $g_4(\td(\RR) + 1, \Size{\varphi})$.
	Define $I = \max\{g_4(\td(\RR) + 1, \Size{\varphi}),\Size{\varphi} (\td(\RR) + 1)\}$, which is computable.
	Since $g_3(\td(\RR) + 1, \Size{\varphi}, p)$ can be assumed without loss of generality to be monotonically increasing in $p$, all iterations together take at most $I \cdot g_3(\td(\RR) + 1, \Size{\varphi},I) \cdot \Size{\RR}$ time.
	Together with the preprocessing steps and $h \leq \td(\RR) + 1$, we can bound the running time of \Cref{alg:main-algorithm} by $f_1(\Size{\varphi}, \td(\RR)) \cdot \Size{\RR}$,
	where $f_1(\Size{\varphi}, \td(\RR)) = 2^{\BigO{\td(\RR)^2}} + I \cdot g_3(\td(\RR) + 1, \Size{\varphi}, I)$ is a computable function.

	\subparagraph*{Constructing a $\boldsymbol{Q}$-Solution.}
	For the constructive part of the statement assume that $\RR\models\varphi$ (otherwise there is nothing to show).
	\Cref{alg:main-algorithm} retains in this case a $p$-pure solution $\mathcal{S}^p$ for some $p$; recall Case~\ref{alg:main-algorithm-yes-case}.
	Recall that \Cref{thm:extension} not only guarantees the existence of a solution $R_1', \ldots, R_r'$ for $\RR$, but also gives us the means to compute one.
	To this end, for every $i \in Q$, we consider every $a_i$-tuple $\overline{u} \in V(\RR)^{a_i}$, where $a_i = \arity(R_i)$ is the arity of the relation $R_i$.
	By \Cref{thm:extension}, we can decide, after an $f'(p)\cdot\Size{\RR}$ preprocessing step, whether $R_i'(\overline{u})$ holds in $\BigO{a_ih\log(a_ih)}$ time.
	If this is the case, we output $\overline{u}$.
	Since \Cref{thm:extension} guarantees that $R_1', \ldots, R_r'$ is a solution, our output $(R_i')_{i \in Q}$ is a $Q$-solution.
	It remains to argue the running time.
	For a single $i \in Q$, this takes overall $f'(p)\cdot\Size{\RR} + \BigO{{\Size{\RR}^{a_i}}\cdot a_ih\log(a_ih)}$ time.
	Note that the preprocessing step can be done once globally instead of for every $i \in Q$ individually.
	Let $a = \max_{i \in Q} a_i$ be the maximum arity over all relations whose index is in $Q$.
	By the above arguments, we can upper bound the running time for the entire set $Q$, whose size is at most $r \leq \Size{\varphi}$, with $f'(p)\cdot\Size{\RR} + \BigO{\Size{\varphi} \cdot {\Size{\RR}^{a}}\cdot a h\log(a h)}$.
	Recall that $\Size{\varphi} \geq \sum_{i \in [r]} \arity(R_i)$.
	Thus, the above-stated running time can be upper bounded by $f'(p)\cdot\Size{\RR} + \BigO{\Size{\varphi} \cdot {\Size{\RR}^{a}}\cdot \Size{\varphi} h\log(\Size{\varphi} h)}$.
	As shown above, we can bound $p$ by a computable function in $\td(\RR) + 1$ and $\Size{\varphi}$.
	Therefore, we can rewrite the running time bound as $f_2(\Size{\varphi}, \td(\RR)) \cdot \Size{\RR}^{\max_{i \in Q}\arity(R_i)}$, where the function $f_2$ is a computable function that incorporates the computable functions $f'$ and $f_1$ (the running time of the decision variant of \Cref{alg:main-algorithm}) and the hidden constants in the $\mathcal{O}$-notation.

	We can use $f = f_2$ (since $f_2$ incorporates $f_1$) as the computable function from the theorem statement to cover both cases.
\end{proof}

\subsection{Comparison with Kolaitis and Vardi}
\label{paragraph:kolatis_vardi_comp}

Let us now elaborate on the aforementioned connections to Kolaitis and Vardi~\cite[Section~5]{KolaitisV87}.
Call a solution relation canonical if its truth on a tuple depends only on the atomic type of that tuple in the input structure expanded by an auxiliary total order.
A solution is canonical if each of its witness relations is canonical.

Writing $k$ for the number of universally quantified variables in $\phi$, \cite[Theorem~11]{KolaitisV87} states that if the countable random structure $\mathcal{A}$ satisfies $\phi$, then there is a finite ordered structure $\mathcal{B}$ that realizes every ordered input type on at most $k$ elements and admits a canonical solution.
To show this, fix a finite structure $\mathcal{B}$ (also equipped with a total order) such that each ordered atomic type on at most $k$ elements appears in it.
Then, Ramsey theory
is used to show that there exists a large finite structure $\mathcal{C}$, such that if $\mathcal{C} \models \phi$, then $\mathcal{C}$ contains a copy of $\mathcal{B}$ with a canonical solution.
Since $\mathcal{A}$ is the infinite countable random structure, $\mathcal{A}$ contains a copy of $\mathcal{C}$.\footnote{Technically, a copy of $\mathcal{C}$ after forgetting the auxiliary total order.}
From $\mathcal{A} \models \phi$ one can then deduce $\mathcal{C} \models \phi$, and hence obtains that $\mathcal{B}$ has a canonical solution.
Conversely, equip $\mathcal{A}$ with an auxiliary total order.
For each $k$-tuple $\bar v$ of $\mathcal{A}$, choose a $k$-tuple
$\bar w$ of $\mathcal{B}$ with the same ordered atomic input type,
and prescribe the solution type of $\bar v$ to match that of
$\bar w$.
Canonicity ensures that these prescriptions are independent of
the chosen representatives and consistent on shared subtuples,
yielding a solution on $\mathcal{A}$.
Decidability follows from the fact that $\mathcal{B}$ can be constructed effectively, and that it can be checked if $\mathcal{B}$ has a canonical solution.

Recall that in the problem studied by Kolaitis and Vardi, they are only given an \SNP\ sentence $\varphi$ and want to decide whether $\varphi$ holds on every relational structure.
The problem we are dealing with asks whether $\RR\models\varphi$ holds for a \emph{given} relational structure $\RR$.
In particular, we cannot construct an arbitrary structure $\mathcal{B}$ with a canonical solution but have to take $\RR$ into account.

Consequently, a first key difference of their approach compared to ours concerns the information recorded by types.
Ordered atomic input types suffice in the argument of Kolaitis and Vardi, but they are completely insufficient in our setting.
For example, consider $K_2$ on $\{u,v\}$ with $\phi_2$ being the standard \SNP{} encoding of $2$-colorability using two unary relations.
Obviously $K_2 \models \phi_2$.
The singleton tuples $(u)$ and $(v)$ have the same atomic input type even after adding a linear order.
A canonical solution would therefore assign them the same color, so no such solution exists for this encoding.
Hence, even in this simple case, the existence theorem (\cref{thm:existenceOfPureSolution}) cannot hold when using canonical solutions instead of pure solutions.

A second difference concerns the availability of a suitable structure
on which to apply Ramsey arguments.
In their setting, $\mathcal{A}$ contains $\mathcal{C}$ simply because $\mathcal{A}$
contains every finite input structure.
Hence, a solution on $\mathcal{A}$ restricts to one on $\mathcal{C}$.
In our setting, we must establish that $\RR^s\models\phi$ whenever
$\RR\models\phi$.
For $s\geq p\geq M$, this follows from
\Cref{thm:blowup_preserves_feasibility} and relies on the
well-quasi-ordering of structures of bounded treedepth.
Moreover, we implement the Ramsey step on a common integer index set and not on the relational structure $\RR^s$ directly.
There, tuples of vertices of $\RR$ having a fixed $p$-route pattern are encoded by integer subsets of specific sizes.
This encoding allows us to use ordinary hypergraph Ramsey
while retaining the prescribed structure.

Finally, our algorithm must decide $\RR\models\phi$ uniformly over pairs $(\phi,\RR)$, with a computable parameter dependence on the formula length and treedepth.
It searches for a negative certificate or a $p$-pure solution without constructing the large Ramsey structure or knowing $M$ in advance.
\Cref{thm:obstruction-max-size-computable} supplies the computable bound on this search, and \Cref{thm:extension} reconstructs witnesses on the input domain.

\section{Applying the Metatheorem}
\label{sec:applications}

In this section we illustrate the reach of Theorem~\ref{main_thm:treedepth_meta_theorem} by applying it to several graph parameters whose exact computation is \NP-hard already for small fixed values: stack number, queue number, track number, twin-width, and boxicity. For each we show that the property ``the parameter is at most~$k$'' is definable by an \SNP{} sentence whose length is bounded by a function of $k$. Since all five parameters are bounded in terms of treewidth, a well-studied width parameter that is in turn bounded by treedepth~\cite{NesetrilM06}, the meta-theorem yields fixed-parameter algorithms for their exact computation with respect to treedepth; we first record this deduction once, in a generic form, and then instantiate it after each encoding. The tractability with respect to treedepth was, to our knowledge, previously unknown for all five problems.

\begin{lemma}\label{lem:generic}
Let $\pi$ be a graph parameter and let $g_1,g_2$ be computable, non-decreasing functions. Suppose that for every $k\in\N$ there is an \SNP{} sentence $\Phi_k$ of arity at most $a$, computable from~$k$, with $|\Phi_k|\le g_1(k)$ and $G\models\Phi_k\iff\pi(G)\le k$ for every graph $G$, and that $\pi(G)\le g_2(\td(G))$ for every graph $G$. Then, given $G$ and $k$, deciding $\pi(G)\le k$ takes time $f'(\td(G))\cdot\Size{G}$ for a computable function $f'$; $\pi(G)$ can be computed within the same time bound. \end{lemma}
\begin{proof}
We regard $G$ as the relational structure $(V,E)$; its Gaifman graph is $G$, so its treedepth is~$\td(G)$. We compute $t:=\td(G)$ in \FPT\ time~\cite{ReidlRVS14}. If $k\ge g_2(t)$, then $\pi(G)\le g_2(t)\le k$ by assumption, and we answer yes. Otherwise $k<g_2(t)$, so $|\Phi_k|\le g_1(k)\le g_1(g_2(t))$, and running~$\mathbb A$ from Theorem~\ref{main_thm:treedepth_meta_theorem} on $\Phi_k$ and $G$ decides $\pi(G)\le k$ in time $f(g_1(g_2(t)),t)\cdot|G|$. Since $k\mapsto\Phi_k$ is computable, the number of steps needed to compute $\Phi_k$ is a computable function $T(k)$, and $T(k)\le\max_{k'\le g_2(t)}T(k')$ is bounded by a computable function of $t$ because $k<g_2(t)$. To compute~$\pi(G)$, run the above for the smallest value that $\pi(G)$ can possibly take and iterate with increasingly bigger~$k$ until  $g_2(t)$ is reached; return the least $k$ that is accepted.
\end{proof}

The brief complexity-theoretic overviews that complement the following subsections use standard notions of classical and parameterized complexity as well as standard structural graph parameters that play no role in our proofs; we do not define them here and refer the reader to the textbooks~\cite{Die25},~\cite{CyganFKLMPPS15} and~\cite{DowneyF13} instead.
\subsection{Stack and Queue Layouts}
A \emph{linear layout} of a graph $G = (V, E)$ consists of a total order $<$ on $V$ together with a partition of~$E$ into $k$ \emph{pages} $E_1, \ldots, E_k$, for some $k \in \mathbb{N}$, and one typically requires that the edges on each page satisfy a certain condition.
The two classical and by-far most studied linear layouts are stack and queue layouts~\cite{BK.btg.1979,HeathR92,DW.LLG.2004}; see also \Cref{fig:example-stack-queue} for an illustration.
\begin{definition}[Stack and queue layout]\label{def:layout}
Let $G=(V,E)$ be a graph and let $(<,E_1,\dots,E_k)$ be a $k$-page linear layout of $G$ for some $k\in \mathbb{N}$.
   Let $ab, cd \in E$ be two edges with $a < b$ and $c < d$, and assume by symmetry that $a < c$.
   They \emph{cross} if $c < b < d$, and they \emph{nest} if $c < d < b$.
We call $(<,E_1,\dots,E_k)$ a $k$-page \emph{stack layout} if no two edges on the same page cross and a $k$-page \emph{queue layout} if no two edges on the same page nest.
\end{definition}
\begin{figure}
	\centering
	\includegraphics{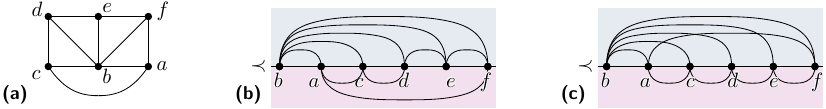}
	\caption{\textbf{\textsf{(a)}} A graph $G$, \textbf{\textsf{(b)}} a 2-page stack layout, and \textbf{\textsf{(c)}} a 2-page queue layout of $G$, both on the same linear order $<$. The two pages are colored blue and purple, respectively.}
	\label{fig:example-stack-queue}
\end{figure}
Stack and queue layouts give rise to the notions of stack and queue number of a graph:

\begin{definition}[Stack and queue number]\label{def:stackqueue-number}
	Let $G = (V, E)$ be a graph.
	The \emph{stack number} $\sn(G)$ of $G$ is the least $k$ such that $G$ admits a $k$-page stack layout, and the \emph{queue number} $\qn(G)$ of $G$ is the least $k$ such that $G$ admits a $k$-page queue layout.
\end{definition}

The relation between structural properties of a graph $G$ and its stack or queue number has been studied extensively in the past, and we refer to the survey of Dujmovi{\'c} and Wood~\cite{DW.LLG.2004} for an overview.
Most notably, planar graphs have stack number at most 4~\cite{Yan.EPG.1989} and queue number at most 49~\cite{DujmovicJMMUW20}, and the upper bound for the stack number is tight~\cite{BKK+.FPA.2020}.
Moreover, $\sn(G)\le\tw(G)+1$~\cite{DW.GTG.2007} and $\qn(G)\le2^{\tw(G)}-1$~\cite{Wie.QNG.2017} hold for every graph $G$, where $\tw(G)$ denotes the treewidth of $G$.

Already deciding whether $\sn(G) \leq 2$ or $\qn(G) \leq 1$ is \NP-complete~\cite{Wid.CHC.1982,BK.btg.1979,HeathR92}, which has motivated studying the complexity of deciding whether $\sn(G) \leq k$ or $\qn(G) \leq k$ with respect to structural parameters of $G$.
When $k$ is part of the input, both problems were shown to be \FPT\ with respect to the vertex cover number of $G$~\cite{BhoreGMN20,BhoreGMN22}.
More recently, this was extended to vertex integrity~\cite{DepianFGS25}.
For stack number, Ganian, M{\"{u}}ller, Ordyniak, Paesani, and~Rychlicki~\cite{GMO+.TST.2024} also obtain an \FPT\ algorithm parameterized by the feedback edge number.
For treedepth and treewidth, earlier fixed-parameter algorithms only handle special cases.
In particular, deciding $\qn(G) \leq 1$ is \FPT\ in $\td(G)$~\cite{BhoreGMN22}.
Bannister and Eppstein~\cite{bannisterCrossingMinimization1page2018} showed that deciding $\sn(G) \leq 2$ is \FPT\ in $\tw(G)$ using Courcelle's theorem, which was later improved to an algorithm with single-exponential dependence on $\tw(G)$~\cite{GMO+.TST.2024}.
Whether the general problems are \FPT\ with respect to less restrictive parameters, in particular treedepth or treewidth, has repeatedly been posed as an open question~\cite{GMNZ.PCG.2021,BSK.Bop.2025,GMO+.TST.2024,DepianFGS25}.
Recently, it was shown that deciding $\qn(G) \leq 1$ is para\NP-hard with respect to the pathwidth of $G$~\cite{bekosRecognitionOuterplanarGraphs2026}, ruling out the existence of an \FPT\ and even \XP\ algorithm with respect to the pathwidth (and thus treewidth) of $G$ (unless $\P = \NP$).
For treedepth, we answer this open question in the affirmative for both stack and queue number.
To this end, we provide \SNP\ sentences for both problems:

\begin{theorem}\label{thm:layouts}
For every $k\ge1$ there are \SNP{} sentences $\Phi^{\mathrm{st}}_k$ and $\Phi^{\mathrm{qu}}_k$ such that for every finite graph $G$,
\[
G\models\Phi^{\mathrm{st}}_k\ \Longleftrightarrow\ \sn(G)\le k
\qquad\text{and}\qquad
G\models\Phi^{\mathrm{qu}}_k\ \Longleftrightarrow\ \qn(G)\le k .
\]
Moreover, $\Phi^{\mathrm{st}}_k$ ($\Phi^{\mathrm{qu}}_k$, respectively) has $k+1$ existentially quantified relation symbols, all of arity~$2$, four first-order variables, and length $\mathcal{O}(k)$; the map $k\mapsto\Phi^{\mathrm{st}}_k$ ($k\mapsto\Phi^{\mathrm{qu}}_k$, respectively) is computable in time polynomial in $k$.
\end{theorem}
\begin{proof}
\noindent\emph{The \SNP sentences.}
Both sentences existentially quantify over a binary relation $\Lo$ and $k$ binary relations $\Pg_1,\dots,\Pg_k$. The intended interpretation is that $\Lo$ is the total order, written $u<v$ for $\Lo(u,v)$, and $\Pg_j(u,v)$ says that the edge $uv$ lies on page $j$.
First, we ensure that $\Lo$ is a total order via
\[
\text{(O1) } \neg \Lo(u,u), \quad \text{(O2) }\Lo(u,v) \wedge \Lo(v,x)\;\to\;\Lo(u,x), \quad\text{ and } \quad \text{(O3) } \Lo(u,v)\;\vee\;\Lo(v,u)\;\vee\;u=v.
\]
The next set of conditions ensures that pages cover exactly the edges of $G$. For all $j\in\{1,\dots,k\}$
\begin{enumerate}[label=(P\arabic*),leftmargin=3.2em]
\item\label{P1} $\Pg_j(u,v)\;\to\;E(u,v)$ \quad and\quad $\Pg_j(u,v)\to\Pg_j(v,u)$,
\item\label{P2} $E(u,v)\;\to\;\Pg_1(u,v)\vee\dots\vee\Pg_k(u,v)$.
\end{enumerate}
Note that the latter does not require the pages to be disjoint: an edge may lie on several pages. Although this deviates from the definition of layouts given above, it is harmless in terms of determining the stack and queue number of a graph.
With our last set of constraints, we enforce a stack or queue layout, respectively by adding the following constraints for each $j\in\{1,\dots,k\}$. Observe that in both cases, the layout condition is expressed via a simple forbidden pattern characterization on four vertices.
\begin{enumerate}[label=(S\arabic*),leftmargin=3.2em]
\item\label{S1} $\neg(\Lo(u,x)\ \wedge\ \Lo(x,v)\ \wedge\ \Lo(v,y)\ \wedge\ \Pg_j(u,v)\ \wedge\ \Pg_j(x,y))$ (no crossing within a page)
\end{enumerate}
\begin{enumerate}[label=(Q\arabic*),leftmargin=3.2em]
\item\label{Q1} $\neg (\Lo(u,x)\ \wedge\ \Lo(x,y)\ \wedge\ \Lo(y,v)\ \wedge\ \Pg_j(u,v)\ \wedge\ \Pg_j(x,y))$ (no nesting within a page)
\end{enumerate}

\medskip
Let $\forall\bar z\,\psi^{\mathrm{st}}_k$ be the conjunction of the universal closures of (O1)--\ref{P2} and \ref{S1}, and $\forall\bar z\,\psi^{\mathrm{qu}}_k$ that of (O1)--\ref{P2} and \ref{Q1}, for all $j$ and with $\bar z = (u,v,x,y)$.
 Then
\[
\Phi^{\mathrm{st}}_k=\exists\Lo\,\exists\Pg_1\cdots\exists\Pg_k\;\forall\bar z\;\psi^{\mathrm{st}}_k
\qquad\text{and}\qquad
\Phi^{\mathrm{qu}}_k=\exists\Lo\,\exists\Pg_1\cdots\exists\Pg_k\;\forall\bar z\;\psi^{\mathrm{qu}}_k .
\]
For each $j$, the conditions (O1)--\ref{P1} and \ref{S1} (respectively \ref{Q1}) have constant length and \ref{P2} consists of $k+1$ atoms; hence $\Phi^{\mathrm{st}}_k$ and $\Phi^{\mathrm{qu}}_k$ have length $\BigO{k}$ and four first-order variables, and each can be written down in time polynomial in $k$.

\medskip
\noindent\emph{Correctness.}
We demonstrate the correctness w.r.t.\ the stack number; the argument for the queue number is identical with ``cross'' replaced by ``nest'' and \ref{S1} by \ref{Q1}.

\smallskip
\noindent($\Leftarrow$) Let $\sn(G)\le k$ and fix a $k$-page stack layout $(<,E_1,\dots,E_k)$ of $G$. Interpret $\Lo$ as $<$ and $\Pg_j$ as $\{(u,v):uv\in E_j\}$. Clauses (O1)--(O3) hold because $<$ is a total order; \ref{P1} because $E_j\subseteq E$ and edges are unordered pairs; \ref{P2} because $E_1\cup\dots\cup E_k=E$. For \ref{S1}, suppose $u<x<v<y$ with $\Pg_j(u,v)$ and $\Pg_j(x,y)$. Then $uv,xy\in E_j$ and they cross in $<$, contradicting the definition of a stack layout. Hence, $G\models\Phi^{\mathrm{st}}_k$.\\

\noindent($\Rightarrow$) Let $(V,E,\Lo,\Pg_1,\dots,\Pg_k)\models\forall\bar z\,\psi^{\mathrm{st}}_k$. By (O1)--(O3), $\Lo$ is a strict linear order on $V$.
By \ref{P1}, each $\Pg_j$ is a symmetric relation contained in $E$, so $F_j:=\{uv:\Pg_j(u,v)\}$ is a set of edges of $G$, and by \ref{P2} every edge of $G$ lies in some $F_j$. Choose for every edge $e$ an index $j(e)$ with $e\in F_{j(e)}$ and put $E_j:=\{e:j(e)=j\}$; then $E_1,\dots,E_k$ is a partition of $E$ with $E_j\subseteq F_j$. It remains to check that no two edges on a common page $E_j$ cross. Suppose there exist $uv,xy\in E_j$ with $u<v$ and $x<y$ that cross in $<$. Then, by Definition~\ref{def:layout},
w.l.o.g.\ $u<x<v<y$. Since $u<x<v<y$ means $\Lo(u,x)$, $\Lo(x,v)$, and $\Lo(v,y)$, and since $E_j\subseteq F_j$ yields $\Pg_j(u,v)$ and $\Pg_j(x,y)$, condition \ref{S1} is falsified by $(u,v,x,y)$---a contradiction. Thus $(<,E_1,\dots,E_k)$ is a $k$-page stack layout and $\sn(G)\le k$.
\end{proof}
\begin{corollary}
Given a graph $G$ and $k\in\N$, deciding $\sn(G)\le k$ ($\qn(G)\le k$, respectively) is fixed-parameter tractable with respect to $\td(G)$. Moreover, $\sn(G)$ and $\qn(G)$ can be computed in time $f(\td(G))\cdot|G|$ for a computable function $f$.
\end{corollary}
\begin{proof}
By Theorem~\ref{thm:layouts}, $\Phi_k$ has length $\mathcal{O}(k)$ and is computable from $k$; moreover, it is known that $\sn(G)\le\tw(G)+1$~\cite{GanleyH01}, hence $\sn(G)\le\td(G)$ and that $\qn(G)\le2^{\tw(G)}-1$~\cite{Wie.QNG.2017}, hence $\qn(G)\le2^{\td(G)-1}-1$. The statement then immediately follows from Lemma~\ref{lem:generic}.
\end{proof}

\subsection{Track Layouts}
Track layouts were introduced by Dujmovi\'c, Morin, and Wood~\cite{DujmovicMW05}, building on the notion of a \emph{layered} drawing: the vertices are arranged on a bounded number of parallel lines, the \emph{tracks}, each carrying a linear order, and one asks how many tracks are needed so that the edges between any two tracks can be drawn without crossings. See also~\cite{DPW04} for a systematic study of track layouts.

\begin{figure}
	\centering
	\includegraphics{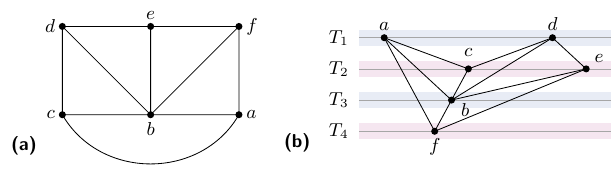}
	\caption{\textbf{\textsf{(a)}} A graph $G$ with track number $4$.
\textbf{\textsf{(b)}} A $4$-track layout of $G$: the vertices are partitioned into
four tracks $T_1,\dots,T_4$, each drawn as a horizontal band and
ordered from left to right. No two edges $uv$ and $xy$ between the same pair of
tracks form an X-crossing, i.e., there are no such edges with $u<x$ on one track and
$y<v$ on the other.}
	\label{fig:example-tracks}
\end{figure}

\begin{definition}[Track layouts, track number]\label{def:track}
Let $t\ge1$. A \emph{$t$-track layout} of $G$ consists of a partition $V=T_1\cup\dots\cup T_t$ of the vertex set into independent sets, the \emph{tracks}, together with a strict linear order $<_i$ on each track $T_i$, such that no two edges form an \emph{X-crossing}, that is, there are no $i\ne j$ and edges $uv,xy$ with $u,x\in T_i$, $v,y\in T_j$, $u<_ix$ and $y<_jv$. The \emph{track number} $\tn(G)$ is the least $t$ such that $G$ has a $t$-track layout.
\end{definition}

Equivalently, draw the tracks as parallel horizontal lines with the vertices placed along them in the given orders; then the layout condition says that the edges between any two tracks, drawn as straight segments, are pairwise non-crossing. Since tracks are independent sets, every edge runs between two different tracks; see Figure~\ref{fig:example-tracks} for a small example.

Deciding whether $\tn(G)\le 2$ is easy, since the graphs of track number at
most~$2$ are exactly the caterpillar forests~\cite{DPW04}.  However, already the next
value is hard: Bannister, Devanny, Dujmovi\'c, Eppstein, and
Wood~\cite{BannisterDDEW19} show that a graph is leveled planar if and only if it is
bipartite and admits a $3$-track layout; since recognizing the former is \NP-complete~\cite{HeathR92}, so is deciding $\tn(G)\le 3$, which settles a question left open by Dujmovi\'c, P\'or, and
Wood~\cite{DPW04}.

On the parameterized side, the same authors show by a Myhill--Nerode-type argument that leveled planarity is not finite-state with respect to treewidth, so that Courcelle's theorem and, more generally, dynamic programming over tree decompositions with a bounded number of states are not applicable~\cite[Sect.~5.1]{BannisterDDEW19}; as finite-state properties are closed under intersection and bipartiteness is finite-state, the same holds for deciding $\tn(G)\le3$. Whether track number is fixed-parameter tractable with respect to treewidth by other means remains open. With respect to treedepth they obtain fixed-parameter tractability only non-uniformly, via the well-quasi-ordering of graphs of bounded treedepth under induced subgraphs~\cite{dingSubgraphsWellquasiordering1992}: for each value of the treedepth an algorithm exists, but it is not known how to compute it. Corollary~\ref{cor:tracks} below replaces this existential statement by a single, explicitly given algorithm that also outputs a track layout.

\begin{theorem}\label{thm:tracks}
For every $t\ge1$ there is an \SNP{} sentence $\Phi^{\mathrm{tr}}_t$ such that for every finite graph $G$,
\[
G\models\Phi^{\mathrm{tr}}_t\quad\Longleftrightarrow\quad\tn(G)\le t .
\]
Moreover, $\Phi^{\mathrm{tr}}_t$ has $t+1$ existentially quantified relation symbols, of arity at most $2$, four first-order variables, and length $\mathcal{O}(t^2)$; the map $t\mapsto\Phi^{\mathrm{tr}}_t$ is computable in time polynomial in $t$.
\end{theorem}

\begin{proof}
We treat a $t$-track layout as a coloring of the vertices with $t$ colors together with a linear order on each color class. As for stack and queue layouts, the layout condition is a forbidden pattern on four vertices.

\medskip
\noindent\emph{The \SNP sentence.}
$\Phi^{\mathrm{tr}}_t$ existentially quantifies over unary relations $\Tr_1,\dots,\Tr_t$ and a binary relation~$\Lo$. The intended interpretation is that $\Tr_i$ is the $i$-th track and that $\Lo(u,v)$ holds if $u$ and $v$ lie on a common track and $u$ precedes $v$ on it. We impose the following conditions, for all $i,j\in\{1,\dots,t\}$ with $i\ne j$.

The first set of conditions ensures that the tracks partition $V$ into independent sets.
\[
\text{(T1) } \Tr_1(u)\vee\dots\vee\Tr_t(u), \quad\quad \text{(T2) } \neg\big(\Tr_i(u)\wedge\Tr_j(u)\big), \quad\quad \text{(T3) } \neg\big(\Tr_i(u)\wedge\Tr_i(v)\wedge E(u,v)\big)
\]
Next, we enforce $\Lo$ to be a strict linear order on each track.
\begin{enumerate}[label=(O\arabic*),leftmargin=3.2em]
\item\label{O1} $\neg\Lo(u,u)$.
\item\label{O2} $\Lo(u,v)\wedge\Lo(v,w)\to\Lo(u,w)$.
\item\label{O3} $\Tr_i(u)\wedge\Tr_i(v)\;\to\;u=v\vee\Lo(u,v)\vee\Lo(v,u)$.
\end{enumerate}
Finally, we forbid X-crossings between two distinct tracks.
\begin{enumerate}[label=(X\arabic*),leftmargin=3.2em]
\item\label{X1} $\neg\big(\Tr_i(u)\wedge\Tr_i(x)\wedge\Tr_j(v)\wedge\Tr_j(y)\wedge E(u,v)\wedge E(x,y)\wedge\Lo(u,x)\wedge\Lo(y,v)\big)$.
\end{enumerate}

\medskip
Let $\forall\bar z\,\psi^{\mathrm{tr}}_t$ be the conjunction of the universal closures of (T1)--\ref{X1} over all admissible $i,j$. We define $\Phi^{\mathrm{tr}}_t:=\exists\Tr_1\cdots\exists\Tr_t\,\exists\Lo\;\forall\bar z\;\psi^{\mathrm{tr}}_t$. Condition (T1)  has $t$ atoms and there are $\mathcal{O}(t^2)$ conditions of type (T2) and \ref{X1}, each of constant length; hence $|\Phi^{\mathrm{tr}}_t|=\mathcal{O}(t^2)$, every condition has at most four variables, and $\Phi^{\mathrm{tr}}_t$ can be written down in time polynomial in $t$.

Note that $\Lo$ is only constrained on pairs of vertices from a common track; its behavior on other pairs is irrelevant, as \ref{X1} only inspects $\Lo$ on such pairs. In \ref{X1}, the atoms $\Tr_i(u),\Tr_i(x),\Tr_j(v),$ and $\Tr_j(y)$ with $i\ne j$, together with (T2) and \ref{O1}, force $u,x$ to be distinct vertices of track $i$ and $v,y$ distinct vertices of track $j$, so the negated conjunction describes exactly an X-crossing of $uv$ and $xy$ in the sense of Definition~\ref{def:track}.

\smallskip
\noindent\emph{Correctness.}
\noindent($\Leftarrow$) Let $\tn(G)\le t$ and fix a $t$-track layout $(T_1,\dots,T_t;<_1,\dots,<_t)$. Interpret $\Tr_i$ as $T_i$ and $\Lo$ as $\bigcup_i{<_i}$, i.e.\ $\Lo(u,v)$ iff $u,v\in T_i$ and $u<_iv$ for some $i$. Then (T1) and (T2) hold because the tracks partition $V$, and (T3) because they are independent sets. \ref{O1} holds because each $<_i$ is irreflexive; \ref{O2} because $\Lo(u,v)\wedge\Lo(v,w)$ places $u,v,w$ on one track $T_i$ (tracks are disjoint), where $<_i$ is transitive; and \ref{O3} because $<_i$ is total on $T_i$. Finally, if the conjunction in \ref{X1} held for some $u,x\in T_i$, $v,y\in T_j$, then $uv,xy\in E$ with $u<_ix$ and $y<_jv$ would be an X-crossing, contradicting the definition of a track layout. Hence $G\models\Phi^{\mathrm{tr}}_t$.

\noindent($\Rightarrow$) Let $(V,E,\Tr_1,\dots,\Tr_t,\Lo)\models\forall\bar z\,\psi^{\mathrm{tr}}_t$. By (T1) and (T2), $T_i:=\{u:\Tr_i(u)\}$ for $i=1,\dots,t$ is a partition of $V$, and by (T3), each $T_i$ is an independent set. Let $<_i$ be the restriction of $\Lo$ to~$T_i$. By \ref{O1} and \ref{O2}, it is irreflexive and transitive, and by \ref{O3} it is total on $T_i$; hence $<_i$ is a strict linear order on $T_i$. It remains to exclude X-crossings. Suppose $i\ne j$ and edges $uv,xy$ with $u,x\in T_i$, $v,y\in T_j$, $u<_ix$ and $y<_jv$ existed. Then $\Tr_i(u),\Tr_i(x),\Tr_j(v),\Tr_j(y)$, $E(u,v)$, $E(x,y)$, $\Lo(u,x)$ and $\Lo(y,v)$ all hold, so the conjunction negated in \ref{X1} is satisfied---a contradiction. Thus $(T_1,\dots,T_t;<_1,\dots,<_t)$ is a $t$-track layout of $G$ and $\tn(G)\le t$.
\end{proof}

\begin{corollary}\label{cor:tracks}
Given a graph $G$ and $t\in\N$, deciding $\tn(G)\le t$ is fixed-parameter tractable with respect to $\td(G)$. Moreover, $\tn(G)$ can be computed in time $f(\td(G))\cdot|G|$ for a computable function~$f$.
\end{corollary}
\begin{proof}
By Theorem~\ref{thm:tracks}, $\Phi^{\mathrm{tr}}_t$ is computable from $t$ and has length $\mathcal{O}(t^2)$. By Dujmovi\'c, Morin and Wood~\cite{DujmovicMW05}, $\tn(G)\le g_0(\tw(G))$ for a computable, nondecreasing function $g_0$, hence $\tn(G)\le g_2(\td(G))$ with $g_2(s):=g_0(s-1)$. The statement then immediately follows from Lemma~\ref{lem:generic}.
\end{proof}

\subsection{Twin-Width}
Twin-width was introduced by Bonnet, Kim, Thomass\'e, and Watrigant~\cite{BonnetKTW22,BKTW.TwI.2020} in 2020. Its definition turns the mechanism behind the algorithm of Guillemot and Marx~\cite{GuillemotM14} for permutation pattern matching into a graph parameter: there, pairs of vertices that are ``almost twins''  are iteratively contracted while errors are charged in a red graph. The parameter proved remarkably general: classes of bounded twin-width include bounded treewidth and clique-width, proper minor-closed classes, unit interval graphs, posets of bounded width and pattern-avoiding permutation graphs~\cite{BonnetKTW22}; they are closed under first-order transductions. Moreover, first-order model checking is fixed-parameter tractable on every class of bounded twin-width, provided a contraction sequence is given~\cite{BonnetKTW22}. Twin-width has since become a central concept in the structural theory of dense graphs~\cite{BGKTW21,BonnetGMSTT24}.

Twin-width is usually defined through trigraphs, in which a contracted pair inherits a red edge whenever the two merged vertices disagree, and red edges persist; see Figure~\ref{fig:example-twinwidth} for a small example.
We use the partition definition by Bonnet, Kim, Thomass\'e, and Watrigant~\cite{BonnetKTW22}.
Throughout, we let $|V|=n$.

Let $\Part$ be a partition of $V$ and $u,v\in V$. We write $[v]_{\Part}$ for the part of the partition that contains~$v$, and $u\sim_{\Part}v$ for $[u]_\Part=[v]_\Part$.
Two distinct parts $A,B\in\Part$ are \emph{homogeneous} if either $A\times B\subseteq E$ or $(A\times B)\cap E=\emptyset$, and \emph{red-adjacent} otherwise, i.e.\ if there are $a,a'\in A$ and $b,b'\in B$ with $E(a,b)$ and $\neg E(a',b')$. The \emph{red degree} $\rdeg(\Part)$ is the maximum, over all $A\in\Part$, of the number of parts red-adjacent to $A$. For a contraction sequence $\Part_n,\dots,\Part_1$ we call the passage from $\Part_{i+1}$ to~$\Part_i$ \emph{step $n-i$}; thus step $j\in\{1,\dots,n-1\}$ produces $\Part_{n-j}$.

\begin{definition}[Twin-width]\label{def:tww}
 A \emph{contraction sequence} of a graph $G$ is a sequence of partitions $\Part_n,\Part_{n-1},\dots,\Part_1$ of $V$ such that $\Part_n$ is the discrete partition $\{\{v\}:v\in V\}$, we have $\Part_1=\{V\}$, and for each $1\le i<n$ the partition $\Part_i$ is obtained from $\Part_{i+1}$ by replacing two distinct parts $A,B\in\Part_{i+1}$ by their union $A\cup B$. Its \emph{width} is $\max_i\rdeg(\Part_i)$, and the \emph{twin-width} $\tww(G)$ is the least width of any contraction sequence of $G$.
\end{definition}

\begin{figure}
	\centering
	\includegraphics{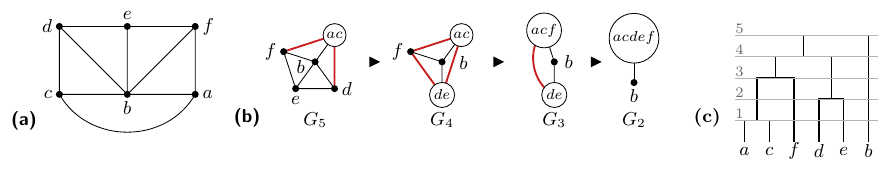}
	\caption{\textbf{\textsf{(a)}} A graph $G$ with $\tww(G)=2$, \textbf{\textsf{(b)}} a contraction sequence of width $2$, shown as the trigraphs $G_5,\dots,G_2$ with red edges drawn in red, and \textbf{\textsf{(c)}} the corresponding dendrogram: two vertices are joined at height $i$ if they first lie in a common part after the $i$-th contraction.}
	\label{fig:example-twinwidth}
\end{figure}

Computing twin-width is hard. Berg\'e, Bonnet, and D\'epr\'es~\cite{BergeBonnetDepres22} showed that deciding $\tww(G)\le k$ is already \NP-complete for $k=4$, ruling out fixed-parameter tractability w.r.t.~$k$.
Graphs of twin-width at most one can be recognized in polynomial time~\cite{BonnetKRTW22}, and the cases $k=2,3$ are open. No fixed-parameter approximation with respect to twin-width itself is known~\cite{GanianR26}. Instead, all known algorithms resort to stronger parameters: An \FPT{} approximation of twin-width can be computed in \FPT{} time with respect to the feedback edge number and the vertex integrity~\cite{BGR24}. Recently, Ganian and Rocton~\cite{GanianR26} obtained such an approximation also with respect to treedepth; this is the first tractability result for a parameter that is not a deletion distance to a simple graph class. The same work also proposes an exact \FPT{} algorithm, parameterized by vertex integrity~\cite{GanianR26}.
Whether twin-width can be computed \emph{exactly} in \FPT{} time with respect to treedepth, or even whether $\tww(G)\le k$ can be decided in \FPT{} time with respect to treedepth for a fixed $k\ge4$, is currently open. Below, we answer this question in the positive.

\begin{theorem}\label{thm:tww}
For every $k\in\N$ there is an \SNP{} sentence $\Phi_k$ such that for every finite graph $G$,
\[
G\models\Phi_k\quad\Longleftrightarrow\quad\tww(G)\le k .
\]
Moreover, $\Phi_k$ has two existentially quantified relation symbols, both of arity $4$, at most $\mathcal{O}(k)$ first-order variables, and length $\mathcal{O}(k^2)$; the map $k\mapsto\Phi_k$ is computable in time polynomial in $k$.
\end{theorem}
\begin{proof}
\noindent\emph{Intuition.}
A natural way to represent a contraction sequence is to index its steps by vertices and to point from each step to the next one; but the sequence must not stop early, so every step would need a successor until everything is merged, and ``every step has a successor'' is inherently a $\forall \exists$ statement that \SNP{} does not allow. We therefore describe the sequence the way a dendrogram in hierarchical clustering does: not by listing the merges in order, but by recording for any two pairs of vertices which of them is joined first. This is a relation on pairs of pairs, and it determines the sequence completely---the partition after a given merge consists of all pairs joined no later than it, and the partition before it of all pairs joined strictly earlier. In this representation the sequence cannot stop early, because a total order on pairs assigns every pair a moment at which it is joined, and totality of an order is expressible with universal quantifiers only.

\smallskip
\emph{The \SNP sentence.} We define $\Phi_k$ as follows: Let $\Phi_k$ existentially quantify over two relations, $\mathrm{C}$ and $\Red$, both of arity $4$. $\mathrm{C}(u,v,x,y)$ is meant to say that whenever $x$ and $y$ lie in a common part of the contraction sequence, so do $u$ and $v$. The desired interpretation of $\Red(x,y,a,b)$ is that the parts containing $a$ and $b$ are red-adjacent in the partition obtained at the step that merges $x$ and~$y$. Below, we force all necessary properties of the intended interpretations.
In doing so, we make use of the following shorthand notation
\begin{align*}
uv \preceq xy&\;\coloneqq\; \mathrm{C}(u,v,x,y),\\
uv \prec xy&\;\coloneqq\; \mathrm{C}(u,v,x,y)\wedge\neg \mathrm{C}(x,y,u,v),\\
uv \equiv xy&\;\coloneqq\; \mathrm{C}(u,v,x,y)\wedge \mathrm{C}(x,y,u,v),
\end{align*}
and let $uv \not\preceq xy$ denote $\neg (uv \preceq xy)$.

First, we enforce the merge order to be a total preorder (i.e., an order that is total, reflexive, and transitive) on pairs of elements. We also ensure that within this order, pairs behave reflexively, symmetrically, and transitively.
\begin{enumerate}[label=(C\arabic*),leftmargin=3.2em]
\item\label{B1} $uv \preceq xy\;\vee\;xy \preceq uv$ (merge order is total and reflexive)
\item\label{B2} $uv \preceq xy
\wedge xy \preceq st\;\to\;uv \preceq st$ (merge order is transitive)
\item\label{B3} $uv \preceq xy\to vu \preceq xy$  and $uv \preceq xy\to uv \preceq yx$ (merge elements are symmetric)
\item\label{B4} $uu \preceq xy$ (merge elements are reflexive) and $x\neq y\to xy \not\preceq uu$ (in the initial partition, no distinct elements are merged)
\item\label{B5} $uv \preceq xy\wedge vw \preceq xy\;\to\;uw \preceq xy$ (merge elements are transitive)
\end{enumerate}
Furthermore, in a contraction sequence, we perform only one merge at a time. This can be expressed as follows.
\begin{enumerate}[label=(C\arabic*),leftmargin=3.2em,resume]
\item\label{B6} $uv\equiv xy\ \wedge\ u'v'\equiv xy\ \wedge\ u\neq v\ \wedge\ u'\neq v'$\\
\hspace*{2em}$\to\;\big(uu'\prec xy\wedge vv'\prec xy\big)\vee\big(uv'\prec xy\wedge vu'\prec xy\big)$.
\end{enumerate}
For $\Red$ we enforce the following: The first statement says that whenever two parts of a partition are not homogeneous, then they must be marked as red-adjacent. The second statement ensures that being red-adjacent is a property of pairs of parts and not of pairs of individual vertices.
\begin{enumerate}[label=(R\arabic*),leftmargin=3.2em]
\item\label{R1} $aa'\preceq xy\ \wedge\ bb'\preceq xy\ \wedge\ ab\npreceq xy\ \wedge\ E(a,b)\wedge\neg E(a',b')\;\to\;\Red(x,y,a,b)$.
\item\label{R2} $\Red(x,y,a,b)\ \wedge\ aa'\preceq xy\ \wedge\ bb'\preceq xy\;\to\;\Red(x,y,a',b')$.
\end{enumerate}
Lastly, the number of red-adjacent parts must be forced to be at most $k$ at all times. We express this by saying that there cannot be $k+1$ such parts at any merging step. Note that this is the only statement that makes use of parameter $k$.
\begin{enumerate}[label=(R\arabic*),leftmargin=3.2em,resume]
\item\label{R3} $\displaystyle\neg\Big(\bigwedge_{i=0}^{k}\Red(x,y,a,b_i)\;\wedge\bigwedge_{0\le i<j\le k}b_ib_j\npreceq xy\Big)$.
\end{enumerate}
Finally, we let $\Phi_k=\exists \mathrm{C}\,\exists\Red\;\forall\bar z\;\psi_k$, where $\forall\bar z\;\psi_k$ is the conjunction of the universal closures of the statements \ref{B1}--\ref{R3}.
Every condition---except \ref{R3}---has constant length, while \ref{R3} consists of~$k+1$ atoms of $\Red$ and $\binom{k+1}{2}$ atoms of $\mathrm{C}$; hence $\Phi_k$ has length $\mathcal{O}(k^2)$ and at most $\mathcal{O}(k)$ first-order variables, and it can be written down in time polynomial in $k$.

\smallskip
\emph{Correctness.}
($\Leftarrow$) Let $\Part_n,\dots,\Part_1$ be a contraction sequence of width at most $k$. Towards defining $\mathrm{C}$ and $\mathrm{Red}$, let $\mu\colon V^2\to\{0,\dots,n-1\}$ with $\mu(u,u)=0$ for all $u \in V$ and
\[
\mu(u,v)\;=\;\min\{\,j\ge 1:\ u\sim_{\Part_{n-j}}v\,\},
\]
for all $u,v \in V$ with $u\ne v$.
 Intuitively, $\mu(u,v)$ denotes the moment at which $u$ and $v$ are first put into a common part; such a moment exists because $\Part_1=\{V\}$.

Since the partitions only coarsen along the sequence, we have
\begin{equation}\label{eq:mu}
u\sim_{\Part_{n-j}}v\quad\Longleftrightarrow\quad\mu(u,v)\le j\qquad(0\le j\le n-1),
\end{equation}
and step $j$ merges two parts $A_j,B_j\in\Part_{n-j+1}$ with $\{(u,v):\mu(u,v)=j\}=(A_j\times B_j)\cup(B_j\times A_j)$. Now, choose $\mathrm{C}$ and $\mathrm{Red}$ as follows.
\begin{align*}
\mathrm{C}&:=\{(u,v,x,y)\in V^4:\ \mu(u,v)\le\mu(x,y)\},\\
\Red&:=\{(x,y,a,b)\in V^4:\ x\ne y,\ \ [a]\ne[b]\ \text{and $[a],[b]$ are red-adjacent in }\Part_{n-\mu(x,y)}\},
\end{align*}
where $[\cdot]$ denotes parts of $\Part_{n-\mu(x,y)}$. It remains to verify the clauses for all assignments of the variables. Note that by \eqref{eq:mu}, $uv\preceq xy$ holds iff $u\sim_{\Part_{n-\mu(x,y)}}v$, and $uv\prec xy$ holds iff $u\sim_{\Part_{n-\mu(x,y)+1}}v$ (for $x\neq y$).

\ref{B1}, \ref{B2}: $\le$ on integers is total, reflexive, and transitive. \ref{B3}: $\mu$ is symmetric. \ref{B4}: $\mu(u,u)=0\le\mu(x,y)$, and $\mu(x,y)\ge1>0$ for $x\ne y$. \ref{B5}: fix $x,y$ and put $j=\mu(x,y)$. By \eqref{eq:mu}, $uv\preceq xy$ holds iff $u$ and $v$ lie in a common part of $\Part_{n-j}$. So if $uv\preceq xy$ and $vw\preceq xy$, then $u,v,w$ all lie in the part of $v$, hence $uw\preceq xy$.

\ref{B6}: If $x=y$, the antecedent is unsatisfiable by \ref{B4}, so let $x \neq y$ and $j=\mu(x,y)\ge1$. Moreover, let $(u,v),(u',v')$ be pairs with $\mu(u,v)=\mu(u',v')=j$ and $u\neq v, u' \neq v'$. Then both lie in $(A_j\times B_j)\cup(B_j\times A_j)$. If $u,u'\in A_j$ then $v,v'\in B_j$, so $u\sim_{\Part_{n-j+1}}u'$ and $v\sim_{\Part_{n-j+1}}v'$, which is the first disjunct; the case $u,u'\in B_j$ is symmetric; if $u\in A_j$ and $u'\in B_j$ then $v\in B_j$, $v'\in A_j$, giving the second disjunct; and symmetrically.

\ref{R1}: if $aa'\preceq xy$, $bb'\preceq xy$, $ab\npreceq xy$, $E(a,b)$ and $\neg E(a',b')$, then in $\Part_{n-\mu(x,y)}$ the parts $[a]\ni a'$ and $[b]\ni b'$ are distinct and contain both an edge and a non-edge between them, hence are red-adjacent, so $\Red(x,y,a,b)$ holds. (If $x=y$, the antecedent is unsatisfiable: $aa'\preceq xx$ forces $a=a'$ and $bb'\preceq xx$ forces $b=b'$, contradicting $E(a,b)\wedge\neg E(a,b)$.) \ref{R2}: by definition, membership in $\Red$ depends only on the parts of $a$ and $b$.

\ref{R3}: suppose $\Red(x,y,a,b_i)$ for $i=0,\dots,k$ and $b_ib_j\npreceq xy$ for $i<j$. Then in $\Part_{n-\mu(x,y)}$ the parts $[b_0],\dots,[b_k]$ are pairwise distinct, each distinct from $[a]$ and red-adjacent to it, so $\rdeg(\Part_{n-\mu(x,y)})\ge k+1$, contradicting the width bound. Hence $G\models\Phi_k$.\\

($\Rightarrow$) Now let $(V,E,\mathrm{C},\Red)\models\forall\bar z\,\psi_k$. We use $\preceq$, $\prec$, $\equiv$ as defined above, also between pairs named by single letters $p\in V^2$.

By \ref{B1}--\ref{B3}, $\preceq$ is a total preorder on $V^2$ that is invariant under swapping the two coordinates of either pair.
Let $L_0<L_1<\dots<L_m$ be its $\equiv$-classes, which we call \emph{levels}, listed in $\preceq$-order. By \ref{B4}, any two diagonal pairs $uu$ and $vv$ are tied, and $uu\prec xy$ whenever $x\ne y$; so $L_0$ consists exactly of the diagonal pairs. For a pair $p$ define two binary relations on $V$ by
\[
u\sim_p v\;:\Longleftrightarrow\; uv\preceq p,\qquad\qquad u\sim_{<p}v\;:\Longleftrightarrow\; uv\prec p .
\]
Both depend only on the level of $p$, and $u\sim_{<p}v$ implies $u\sim_p v$. For diagonal $p$, \ref{B4} gives $u\sim_p v$ iff $u=v$.
Each $\sim_p$ is an equivalence relation on $V$: it is reflexive by \ref{B4}, symmetric by \ref{B3} and transitive by \ref{B5}. So is each $\sim_{<p}$ with $p\in L_i$, $i\ge1$: for $p'\in L_{i-1}$, $uv\prec p$ holds iff $uv$ lies in some $L_{i'}$ with $i'\le i-1$, which is true iff $uv\preceq p'$, so $\sim_{<p}$ coincides with $\sim_{p'}$.  Let $\Part(p)$ denote the partition of $V$ into the classes of $\sim_p$, and write $\Part(L_i)$ for $\Part(p)$ with $p\in L_i$. Thus $\Part(L_0)$ is the discrete partition, and for $p\in L_i$ with $i\ge1$ the classes of $\sim_{<p}$ form $\Part(L_{i-1})$.

Next we show that for $p\in L_i$, $i\ge1$, the partition $\Part(L_i)$ arises from $\Part(L_{i-1})$ by merging exactly two parts. Let $N_p$ be the set of pairs in the level of $p$, i.e.\ of those $uv$ with $uv\preceq p$ but not $uv\prec p$; it contains $p$ and no diagonal pair, and by definition $u\sim_p v$ holds iff $u\sim_{<p}v$ or $uv\in N_p$. Fix $u_0v_0\in N_p$ and let $A$ and $B$ be the parts of $\Part(L_{i-1})$ containing $u_0$ and $v_0$, respectively; then $A\ne B$, because $u_0\not\sim_{<p}v_0$. For any $uv\in N_p$, clause \ref{B6} applied to $uv$ and $u_0v_0$ (with $p=xy$) yields $u\sim_{<p}u_0$ and $v\sim_{<p}v_0$, or $u\sim_{<p}v_0$ and $v\sim_{<p}u_0$; so every pair in $N_p$ joins a vertex of $A$ with a vertex of $B$. Conversely, for $a\in A$ and $b\in B$ we have $a\sim_p u_0$, $u_0\sim_p v_0$ and $v_0\sim_p b$, hence $a\sim_p b$ by transitivity of $\sim_p$. Therefore $\Part(L_i)$ is obtained from $\Part(L_{i-1})$ by replacing $A$ and $B$ with $A\cup B$.

As a consequence, $\Part(L_i)$ has exactly one part fewer than $\Part(L_{i-1})$ for every $1\le i\le m$, hence $n-i$ parts. Every pair lies in some level, so all vertices are $\sim_p$-equivalent for $p\in L_m$, i.e.\ $\Part(L_m)=\{V\}$; hence $m=n-1$. Setting $\Part_{n-i}:=\Part(L_i)$ for $0\le i\le n-1$ we obtain a contraction sequence $\Part_n,\dots,\Part_1$ of $G$.

It remains to bound its width. The discrete partition has red degree $0$. Now, fix $1\le i\le n-1$ and $p=xy\in L_i$, so that $x\ne y$ and $u\sim_p v$ iff $uv\preceq xy$. First, $\Red(x,y,\cdot,\cdot)$ contains every red pair of $\Part(p)$: if $A\ne B$ are parts of $\Part(p)$ that are red-adjacent, choose $a,a'\in A$ and $b,b'\in B$ with $E(a,b)$ and $\neg E(a',b')$; then $aa'\preceq xy$, $bb'\preceq xy$ and $ab\npreceq xy$, so \ref{R1} gives $\Red(x,y,a,b)$, and \ref{R2} propagates this to $\Red(x,y,a'',b'')$ for all $a''\in A$ and $b''\in B$. Now suppose some part $A$ of $\Part(p)$ had $k+1$ distinct red-adjacent parts $B_0,\dots,B_k$. Pick $a\in A$ and $b_j\in B_j$. Then $\Red(x,y,a,b_j)$ for all $j$, and $b_jb_{j'}\npreceq xy$ for $j<j'$ because $B_j\ne B_{j'}$, which would violate \ref{R3}. Hence $\rdeg(\Part(p))\le k$. Consequently, the contraction sequence has width at most $k$, i.e.\ $\tww(G)\le k$.
\end{proof}
\begin{corollary}
Given a graph $G$ and $k\in\N$, deciding $\tww(G)\le k$ is fixed-parameter tractable with respect to $\td(G)$, and $\tww(G)$ can be computed in time $f(\td(G))\cdot|G|$ for a computable function $f$.
\end{corollary}
\begin{proof}
By Theorem~\ref{thm:tww}, $\Phi_k$ has length $\mathcal{O}(k^2)$ and is computable from $k$; moreover, it is known that $\tww(G)\le3\cdot2^{\tw(G)-1}$~\cite{JacobP22}, hence $\tww(G)\le3\cdot2^{\td(G)-2}$. The statement then immediately follows from Lemma~\ref{lem:generic}.
\end{proof}

\subsection{Boxicity}
Boxicity was introduced by Roberts~\cite{Roberts69} in 1969 as a natural higher-dimensional analogue of interval graphs: where an interval graph records overlaps of intervals on a line, a graph of boxicity~$d$ records overlaps of boxes in $\mathbb{R}^d$, and $\bx(G)$ measures how many dimensions are needed to realize~$G$ in this way; see Figure~\ref{fig:example-boxicity} for a small example.

\begin{definition}[Boxicity]
A \emph{$d$-box} is a product $I^1\times\dots\times I^d\subseteq\real^d$ of closed intervals. A \emph{$d$-box representation} of $G$ assigns to every vertex $v$ a $d$-box $B_v$ such that for all distinct $u,v\in V$,
\[
uv\in E\quad\Longleftrightarrow\quad B_u\cap B_v\neq\emptyset .
\]
The \emph{boxicity} $\bx(G)$ \emph{of $G$} is the least $d\ge1$ such that $G$ has a $d$-box representation. In particular, graphs of boxicity $1$ are called \emph{interval graphs}.
\end{definition}
\begin{figure}
	\centering
	\includegraphics{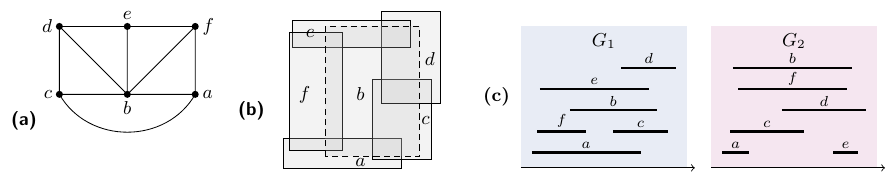}
	\caption{\textbf{\textsf{(a)}} A graph $G$, \textbf{\textsf{(b)}} a $2$-box representation of $G$, and \textbf{\textsf{(c)}} the interval models obtained by projecting the boxes onto the two coordinate axes. The projections define interval graphs $G_1$ (blue) and $G_2$ (purple) on $V(G)$ with $G=G_1\cap G_2$.}
	\label{fig:example-boxicity}
\end{figure}

Computing boxicity is hard: deciding $\bx(G)\le d$ is \NP-complete for every fixed $d\ge2$~\cite{Cozzens81,Yannakakis82,Kratochvil94}, and for every $\varepsilon > 0$ there is no polynomial-time algorithm that approximates the boxicity of an $n$-vertex graph within a factor of $n^{1-\varepsilon}$
unless $\NP=\ZPP$\footnote{i.e., unless every problem in \NP{} can be solved by a zero-error randomized algorithm in expected polynomial time}~\cite{CLN13}.
However, from the perspective of parameterized complexity, deciding $\bx(G)\le d$ is fixed-parameter tractable with respect to the vertex cover number~\cite{AdigaChitnisSaurabh10}, the twin cover number~\cite{Ganian15} and the cluster vertex deletion number~\cite{BCJS16}, and can be approximated within an additive $1$ in \FPT{} time with respect to pathwidth~\cite{BCJS16}. Computing boxicity exactly has been conjectured to remain \NP-complete on graphs of bounded treewidth~\cite{AdigaChitnisSaurabh10}, and even on graphs of bounded bandwidth~\cite{BCJS16}.

 Note that two boxes intersect if and only if they intersect in each of the $d$ coordinates. This yields a fairly straightforward characterization.

\begin{observation}[\cite{Roberts69}]\label{obs:intersection}
For every graph $G$ and $d\ge1$, $\bx(G)\le d$ if and only if $G=(V,\,E(G_1)\cap\dots\cap E(G_d))$ for interval graphs $G_1,\dots,G_d$ on $V(G)$.
\end{observation}

Interval graphs in turn can be characterized by special vertex orderings.

\begin{observation}[\cite{Ramalingam88,Olariu91}]\label{obs:interval}
A graph $H=(V,E)$ is an interval graph if and only if there exists an \emph{interval ordering} of $V$, i.e.,
a linear order $<$ on $V$ with
\[
u<v<w \ \text{ and }\ uw\in E \quad\Longrightarrow\quad uv\in E\qquad\text{for all } u,v,w\in V .
\]
\end{observation}

We now have all the ingredients at hand to express boxicity in \SNP{}.

 \begin{theorem}\label{thm:box}
For every $d\ge1$ there is an \SNP{} sentence $\Phi_d$ such that for every finite graph $G$,
\[
G\models\Phi_d\quad\Longleftrightarrow\quad\bx(G)\le d .
\]
Moreover, $\Phi_d$ has $2d$ existentially quantified relation symbols, all of arity $2$, three first-order variables, and length $\mathcal{O}(d)$; the map $d\mapsto\Phi_d$ is computable in time polynomial in $d$.
\end{theorem}

\begin{proof}
\noindent\emph{The \SNP sentence.}
We let $\Phi_d$ existentially quantify over $2d$ binary relations $\Lo_1,\dots,\Lo_d$ and $\Su_1,\dots,\Su_d$. The intended interpretation is that $\Su_i$ is the edge relation of the $i$-th interval graph~$G_i$, and $\Lo_i$ is a linear order on $V$ that is an interval ordering of $G_i$ (for instance, the order of left endpoints in the interval representation of $G_i$). For each $i\in\{1,\dots,d\}$ we impose the following conditions.

Our first set of conditions ensures that each $\Lo_i$ is a strict linear order:
\[
\mathrm{(O1)}\;\; \neg\Lo_i(u,u)\;\;\quad \mathrm{(O2)}\;\;  \Lo_i(u,v)\wedge \Lo_i(v,w)\;\to\;\Lo_i(u,w)\;\; \quad \mathrm{(O3)}\;\;  \Lo_i(u,v)\;\vee\;\Lo_i(v,u)\;\vee\;u=v.
\]
Next, we define each $\Su_i$ to be a graph that contains $G$, and let all $\Su_i$ intersect in $G$.
\begin{enumerate}[label=(S\arabic*),leftmargin=3.2em]
\item\label{S1_b} $\Su_i(u,v)\to\Su_i(v,u)$  \quad and \quad $\neg\Su_i(u,u)$
\item\label{S2_b} $E(u,v)\;\to\;\Su_i(u,v)$
\item\label{S3_b} $\Su_1(u,v)\wedge\dots\wedge\Su_d(u,v)\;\to\;E(u,v)$
\end{enumerate}
Lastly, we enforce that each $\Lo_i$ is an interval ordering of $\Su_i$.
\begin{enumerate}[label=(I\arabic*),leftmargin=3.2em]
\item\label{I1} $\Lo_i(u,v)\ \wedge\ \Lo_i(v,w)\ \wedge\ \Su_i(u,w)\;\to\;\Su_i(u,v)$.
\end{enumerate}
Finally, $\Phi_d=\exists\Lo_1\cdots\exists\Lo_d\,\exists\Su_1\cdots\exists\Su_d\;\forall\bar z\;\psi_d$, where $\forall\bar z\,\psi_d$ is the conjunction of the universal closures of (O1)--\ref{I1} for all $i$.
For each $i$, the conditions (O1)--\ref{S2_b} and \ref{I1} have constant length and \ref{S3_b} consists of $d+1$ atoms; hence $\Phi_d$ has length $\BigO{d}$ and three first-order variables, and it can be written down in time polynomial in $d$.

\smallskip
\noindent\emph{Correctness.} When given graphs $G_i$, $1 \leq i \leq d$ with a common vertex set $V$, we will write $\bigcap_i G_i$ to denote $(V,\,E(G_1)\cap\dots\cap E(G_d))$.\\
\noindent($\Leftarrow$) Let $\bx(G)\le d$. By Observation~\ref{obs:intersection} there are interval graphs $G_1,\dots,G_d$ on $V$ with $G=\bigcap_iG_i$, and by Observation~\ref{obs:interval} each $G_i$  has an interval ordering $<_i$. Interpret $\Lo_i$ as $<_i$ and $\Su_i$ as $E(G_i)$. Then (O1)--(O3) hold because $<_i$ is a strict linear order; \ref{S1_b} because $G_i$ is a graph; \ref{S2_b} because $E\subseteq E(G_i)$; \ref{S3_b} because $\bigcap_iE(G_i)=E$; and \ref{I1} says that $<_i$ is an interval ordering of $G_i$: whenever $u<_iv<_iw$ and $uw\in E(G_i)$, also $uv\in E(G_i)$. Hence $G\models\Phi_d$.

\smallskip
\noindent($\Rightarrow$) Let $(V,E,\Lo_1,\dots,\Lo_d,\Su_1,\dots,\Su_d)\models\forall\bar z\,\psi_d$. Fix some $i$. By (O1)--(O3), $\Lo_i$ is a strict linear order on $V$, and by \ref{S1_b}, $G_i:=(V,\Su_i)$ is a graph. By \ref{I1}, $\Lo_i$ is an interval ordering of $G_i$, so $G_i$ is an interval graph by Observation~\ref{obs:interval}. By \ref{S2_b}, $E\subseteq\Su_i$ for every $i$, and by \ref{S3_b}, $\bigcap_i\Su_i\subseteq E$ on pairs of distinct vertices; hence, since $E$ is irreflexive, $E=\bigcap_i\Su_i$. Thus $G=G_1\cap\dots\cap G_d$ is an intersection of $d$ interval graphs, and $\bx(G)\le d$ by Observation~\ref{obs:intersection}.
\end{proof}

\begin{corollary}
Given a graph $G$ and $d\in\mathbb{N}$, deciding $\bx(G)\le d$ is fixed-parameter tractable with respect to $\td(G)$. Moreover, $\bx(G)$ can be computed in time $f(\td(G))\cdot|G|$ for a computable function $f$.
\end{corollary}
\begin{proof}
By Theorem~\ref{thm:box}, $\Phi_d$ has length $\mathcal{O}(d)$ and is computable from $d$; moreover, it is known that $\bx(G)\le\mathrm{tw}(G)+2$, hence $\bx(G)\le\td(G)+1$~\cite{ChandranS07}. The statement then immediately follows from Lemma~\ref{lem:generic}.
\end{proof}

\section{Barriers for Generalizing the Metatheorem}
\label{sec:extensions}

In this section, we exclude three natural generalizations of our result under the assumption $\P \neq \NP$.

\subsection{Pathwidth}

Recall that Courcelle's theorem~\cite{Courcelle90} in combination with Bodlaender's theorem~\cite{BodlaenderALinearTime} yields that \MSOtwo{} model checking is possible in time $f(\Size{\varphi},\tw(G)) \cdot |G|$, where $\tw(G)$ is the treewidth of~$G$.\footnote{Interestingly, on graph classes of bounded treedepth, \MSOtwo{} actually has the same expressive power as first-order logic~\cite{elberfeldWhereFirstorderMonadic2016}.}
It is natural to ask if we can solve \snpmodelchecking{} in time $f(\Size{\varphi},\tw(\RR)) \cdot \Size{\RR}^{\BigO{1}}$ too.

As shown in \cref{sec:applications}, it is possible to express the problem of deciding whether a graph has queue number at most one in an \SNP{} sentence of constant length.
However, detecting whether a graph has queue number at most one is \NP-hard, even if the input graphs are restricted to have constant pathwidth~\cite{bekosRecognitionOuterplanarGraphs2026}.
Therefore, it is not possible to obtain a similar theorem when replacing the treedepth of the primal graph with the pathwidth of the primal graph, unless $\P = \NP$.
Another example would be the problem of intervalizing  $4$-colored graphs, which is \NP-hard on graphs with pathwidth at most $2$~\cite{DBLP:journals/dm/AlvarezDS01}.
Although we will not explicitly show this here, it is also easy to express the problem as a fixed \SNP{} sentence on $4$-colored graphs.

\subsection{More Expressive Formulas}

Another reasonable extension, also motivated by Courcelle's theorem, would be to allow formulas of existential second-order logic (\ESO) (see, e.g.~\cite[Definition 7.3]{LibkinElementsOfFiniteModelTheory}), instead of just {\SNP} formulas.
{\ESO} formulas are of the form $\exists R_1,\dots, \exists R_r \psi$, where $\psi$ is a \emph{first-order formula}.
Therefore, the difference between \SNP{} and \ESO{} formulas is that \SNP{} restricts the first-order part of the formula to universal formulas, whereas \ESO{} allows arbitrary first-order formulas.

Again, we show that this extension is unlikely to be possible.
This is a straightforward consequence of Fagin's theorem~\cite{fagin1974generalized}, which essentially states $\ESO = \NP$.
We combine this theorem with an elementary trick to encode a list of integer numbers as a star forest.\footnote{
Related encodings of numerical item sizes by the sizes of stars also appear in the scheduling literature~\cite{AfratiPP88,ZinderBL24,BampisGT96}.}

\begin{theorem}
	\label{thm:eso_not_tractable_treedepth}
	There is a fixed \ESO{} sentence $\phi$ such that deciding whether $G \models \phi$ is \NP-complete, even when the input graph $G$ has treedepth at most $2$.
\end{theorem}
\begin{proof}
	We prove this by a reduction from \problemfont{3-Partition}, which is strongly \NP-hard~\cite{Garey3Partition,GareyJohnsonBook}.
	Let $\mathcal{P}$ be the following graph property: a graph $G$ is in $\mathcal{P}$ if and only if it is a disjoint union of $3n$ stars, for some $n \geq 1$, and the stars can be partitioned into exactly $n$ sets of size $3$ each, such that the number of vertices of the stars contained in each set of the partition is equal.

	It is easy to see that $\mathcal{P}$ is a property closed under isomorphism, and in \NP.
	By Fagin's theorem~\cite{fagin1974generalized} there exists a fixed \ESO{} sentence $\varphi$ such that $G \in \mathcal{P}$ if and only if $G \models \varphi$.

	Let $I_1,\dots,I_{3n} \geq 1$ be the given integers of the input instance of \problemfont{3-Partition}.
	Create a graph~$G$ by, for each $I_i$, adding a star with exactly $I_i$ vertices.
	Then, it is easy to see that the input instance is a yes-instance of \problemfont{3-Partition} if and only if $G \in \mathcal{P}$, and hence if and only if $G \models \varphi$.
	The reduction runs in polynomial-time since we may assume that each integer in the input is given in unary, and since $G$ is a forest of stars it has treedepth at most $2$.

	The \NP-membership of the problem is obvious.
\end{proof}

\subsection{MaxSNP}

Finally, we consider extending the metatheorem to the natural optimization
variant of \SNP{}, introduced by Papadimitriou and Yannakakis~\cite{PapadimitriouY91MaxSNP}.
We use the notation $\textsf{MaxSNP}_0$ from~\cite[Definition~13.5]{CCbook}, where $\textsf{MaxSNP}$ is defined by making $\textsf{MaxSNP}_0$ closed under $L$-reductions.
In \SNP{} we ask for a solution such that each $k$-tuple of elements of $\RR$ satisfies the inner formula $\psi$.
In $\textsf{MaxSNP}_0$, we instead want to maximize the number of $k$-tuples that fulfill the inner formula.
More formally, let $\psi(R_1,\dots,R_r,x_1,\dots,x_k)$ be a quantifier-free formula.
A problem in $\textsf{MaxSNP}_0$ takes a structure $\RR = (V,E_1,\dots,E_b)$ and an integer $p$ as input, and asks if there are relations $R_1,\dots,R_r$ such that there exist at least $p$ $k$-tuples $\bar v$ such that $\psi(R_1,\dots,R_r,\bar v)$ holds.

When using the natural parameter $p$, every fixed problem in this
class is fixed-parameter tractable~\cite{CaiC97MaxSNPTFPT} and admits a polynomial
kernel~\cite{Kratsch12MaxSNPKernel}.
Here, we instead consider treedepth as the parameter.

\begin{theorem}
	\label{thm:max_snp_not_tractable_treedepth}
	There is a fixed quantifier-free formula $\psi(A,x_1,\dots,x_{10})$, where $A$ is a binary relation variable, for which the following decision problem is \NP-complete:
	Given a relational structure $\RR$ and an integer $p$, determine whether there is a relation $A \subseteq V(\RR)^2$ such that
	\[
		\bigl|\{\bar v \in V(\RR)^{10} \mid \RR \models \psi(A, \bar v)\}\bigr| \geq p.
	\]
	This holds even when $\td(\RR) \leq 2$.
\end{theorem}
\begin{proof}
	It is clear that the described problem is in \NP{} by guessing the solution relation $A$ and verifying if the formula holds for each $10$-tuple of the structure.

	For the \NP-hardness, we reduce from \problemfont{Exact Bin Packing}, which is strongly \NP-hard~\cite{DBLP:conf/iwpec/HeegerHS23,Garey3Partition}.
	In this problem, the input consists of $n$ positive integers $I_1,\dots,I_n$, and a positive integer $m \geq 2$, such that $\sum_{i \in [n]} I_i$ is divisible by $m$.
	The question is if it is possible to assign each number to one of $m$ bins, such that the sum of the numbers in each bin is exactly $(\sum_{i \in [n]} I_i)/m$.
	We create a graph $G$ by adding, for each integer $I_i$ of the input, a star on $I_i$ vertices.
	Let $S$ be a unary relation containing all vertices of these stars.
	Add two additional isolated vertices $q_1,q_2$, and two further unary relations $Q_1$, $Q_2$, such that $Q_1 = \{q_1\}$ and $Q_2 = \{q_2\}$.
	Finally, add $m$ further isolated vertices $b_1,\dots,b_m$, and another unary relation $B = \{b_1,\dots,b_m\}$.

	The integers of the input are encoded by stars, and the vertices $b_1,\dots,b_m$ represent the bins these integers can be assigned to.
	The vertices $q_1,q_2$ exist for technical reasons that become clear later.
	Let $\RR$ be the resulting relational structure, and observe that it is a labeled graph with treedepth at most $2$.

	For our formula $\psi$, we will use a binary relation $A$, the intended meaning of $A(x,c)$ is that vertex~$x$ (of a star) is assigned to bin vertex $c$.
	Our first goal will be ensuring that each star is either assigned to no bin, or fully to a bin.
	For this purpose, we first introduce a formula $\mathrm{TypeCheck}(A,x,c)$ which simply verifies that, when $A(x,c)$ holds, then $x$ is a star vertex and $c$ a bin vertex that $x$ is assigned to.
	\begin{align*}
		\mathrm{TypeCheck}(A,x,c) \coloneq (A(x,c) \rightarrow S(x) \land B(c)).
	\end{align*}

	Now, we define the formula $\mathrm{Assigned}(A,x,y,c,d)$ that will ensure a star is either fully, or not at all assigned to a specific bin vertex.
	\begin{align*}
		\mathrm{Assigned}(A,x,y,c,d) \coloneq (A(x,c) \land E(x,y) \rightarrow A(y,c)) \land (A(x,c) \land A(x,d) \rightarrow c = d).
	\end{align*}

	The first conjunct of the formula will later ensure that if a vertex $x$ of a star is assigned to a bin vertex $c$, then all other vertices of the same star are also assigned to $c$.
	Note that this requires some (trivial) transitive argument based on the fact that a star is a connected graph.
	The second conjunct ensures that a vertex cannot be assigned to $c$ and $d$ at the same time when $c \neq d$.

	Now, if each $4$-tuple $(x,y,c,d)$ of elements fulfills $\mathrm{TypeCheck}(A,x,c) \land \mathrm{Assigned}(A,x,y,c,d)$, then each star is either fully assigned to a single bin by $A$, or not assigned to any bin.
	By using a simple trick that makes it catastrophic for the solution size if a $4$-tuple does not fulfill this formula, we will ensure it indeed holds for each $4$-tuple.
	Hence, the remaining task is encoding that the assignments to each bin all sum to the same integer.

	The key is creating a formula which prefers even-spread assignments.
	Concretely, consider
	\begin{align*}
		\mathrm{Even}(A,x,y,c,d) \coloneq (A(x,c) \land A(y,d) \land c \neq d).
	\end{align*}

	Let us give some intuition about the purpose of the formula $\mathrm{Even}(A,x,y,c,d)$.
	Consider some solution $A$, and assume that $\mathrm{TypeCheck}(A,x,c) \land \mathrm{Assigned}(A,x,y,c,d)$ holds for each $4$-tuple of elements $(x,y,c,d)$ under the solution $A$.
	For each bin $b_i$, let $\ell_i$ be the number of vertices assigned to $b_i$, that is, the number of elements $x$ for which $A(x,b_i)$ holds.
	Moreover, let $N = \sum_{i \in [m]} \ell_i$ be the total number of assigned vertices.
	Then, the number of $4$-tuples $(x,y,c,d)$ for which $\mathrm{Even}$ holds is exactly $N^2 - \sum_{i \in [m]} \ell_i^2$.
	This is because $\mathrm{Even}(A,x,y,c,d)$ holds exactly when $x$ is assigned to $c$, $y$ is assigned to $d$, and $c \neq d$.
	For each assigned vertex, there is exactly one bin $c$ it is assigned to.
	Hence, each pair of the $N^2$ pairs of assigned vertices fulfills the first two conditions, and then we only need to remove those pairs again which are assigned to the same bin.
	The number of $4$-tuples which fulfill the formula is maximized when all vertices are assigned and each bin contains exactly the same number of vertices.

	Now, we have all ingredients necessary to describe the formula.
	Formally, the fixed formula $\psi$ reads
	\begin{align*}
		\psi(A,z,x,y,c,d,a_1,\dots,a_5) \coloneq &[Q_1(z) \land \mathrm{TypeCheck}(A,x,c) \land \mathrm{Assigned}(A,x,y,c,d)] \\
		& \lor [Q_2(z) \land \mathrm{Even}(A,x,y,c,d) \land \bigwedge_{i \in [5]} a_i = z].
	\end{align*}
	Set $T = \sum_{i \in [n]} I_i$ and $D = T/m$.
	The integer $p$ is set to $|V(\RR)|^9 + T^2 - m D^2$.
	We will show that $I_1,\dots,I_{n},m$ is a yes-instance of \problemfont{Exact Bin Packing} if and only if there is a relation $A$ such that at least $p$ $10$-tuples $\bar v$ fulfill $\psi(A,\bar v)$.

	\begin{claim}
		If $(I_1,\dots,I_{n},m)$ is a yes-instance of \problemfont{Exact Bin Packing}, then there is a relation $A$ such that at least $p$ $10$-tuples $\bar v$ fulfill $\psi(A,\bar v)$.
	\end{claim}
	\begin{claimproof}
		Let $P_1,\dots,P_m$ be the corresponding partition.
		We add $(v,b_i)$ to $A$ if and only if the integer corresponding to the star containing $v$ is part of $P_i$.

		It is easy to check that this means each $10$-tuple $(z,x,y,c,d,a_1,\dots,a_5)$ with $z = q_1$ fulfills $\psi$.
		Indeed, then $Q_1(z)$ is true, and we have clearly chosen $A$ in such a way that $\mathrm{TypeCheck}(A,x,c)$ and $\mathrm{Assigned}(A,x,y,c,d)$ hold.
		This yields exactly $|V(\RR)|^9$ tuples for which $\psi$ holds.

		Now, we argue that there are $T^2 - mD^2$ further tuples for which $\psi$ holds.
		Concretely, fix some $4$-tuple $(x,y,c,d)$, where $x,y$ are vertices belonging to stars, vertex $x$ is assigned to $c$ and vertex $y$ is assigned to $d$ by relation $A$, and $c \neq d$.
		The number of such tuples is in fact exactly $T^2 - mD^2$ by arguments provided earlier in this proof, and we have that $\mathrm{Even}(A,x,y,c,d)$ holds.
		Then, the tuple $(q_2,x,y,c,d,q_2,q_2,q_2,q_2,q_2)$ fulfills $\psi$, because it fulfills the second disjunct of $\psi$.
	\end{claimproof}

	\begin{claim}
		If there is a relation $A$ such that at least $p$ tuples $\bar v$ fulfill $\psi(A,\bar v)$, then $(I_1,\dots,I_{n},m)$ is a yes-instance of \problemfont{Exact Bin Packing}.
	\end{claim}
	\begin{claimproof}
		Observe that the maximum number of tuples that can fulfill $\psi$ is bounded by $|V(\RR)|^9 + |V(\RR)|^4$, because the first disjunct can only hold if the tuple starts with $q_1$, and the second disjunct can only hold if the tuple starts with $q_2$, and the last five elements of the tuple are also $q_2$.

		We show that for each $4$-tuple $(x,y,c,d) \in V(\RR)^4$, $\mathrm{TypeCheck}(A,x,c)$ and $\mathrm{Assigned}(A,x,y,c,d)$ hold.
		Towards a contradiction, assume that this is not the case.

		If, for a fixed tuple $(x,y,c,d)$ we have that $\mathrm{TypeCheck}(A,x,c)$ or $\mathrm{Assigned}(A,x,y,c,d)$ do not hold, the number of tuples for which $\psi$ holds immediately decreases by $|V(\RR)|^5$, because the first disjunct of $\psi$ leaves $a_1,\dots,a_5$ free.
		So, the total number of tuples for which $\psi$ holds is at most $|V(\RR)|^9 - |V(\RR)|^5 + |V(\RR)|^4$.
		However, $|V(\RR)|^9 - |V(\RR)|^5 + |V(\RR)|^4 < |V(\RR)|^9 < p$, a contradiction.

		Hence, for each $4$-tuple $(x,y,c,d)$ we have that $\mathrm{TypeCheck}(A,x,c)$ and $\mathrm{Assigned}(A,x,y,c,d)$ hold.
		As argued before, this ensures that either all vertices of a star are assigned to the same unique bin vertex, or the vertices of the star are not assigned.

		We argue that then, the relation $A$ gives us a solution to the input instance.
		This is because there must be at least $T^2 - mD^2$ further tuples for which $\psi$ holds, which must then be tuples of the form $(q_2,x,y,c,d,q_2,q_2,q_2,q_2,q_2)$.
		As argued earlier, the number of such tuples which fulfill $\psi$ is exactly $N^2 - \sum_{i \in [m]} \ell_i^2$, where $\ell_i$ is the number of vertices assigned to $b_i$, and $N = \sum_i \ell_i$.

		By Cauchy--Schwarz, we have that $N^2/m \leq \sum_{i \in [m]} \ell_i^2$.
		Therefore,
		\begin{align*}
			N^2 - \sum_{i \in [m]} \ell_i^2 \leq (1- 1/m) N^2 \leq (1- 1/m) T^2 = T^2 - m D^2.
		\end{align*}
		Since $p = |V(\RR)|^9 + T^2 -mD^2$, equality must hold throughout.
		This implies $T = N$, and thus all vertices of stars must be assigned.
		Moreover, we must have $\sum_{i \in [m]} \ell_i^2 = m D^2$.
		From this and the fact that $\sum_{i \in [m]} \ell_i = T = m D$ we obtain
		\begin{align*}
			\sum_{i \in [m]} (\ell_i - D)^2 = \sum_{i \in [m]} \ell_i^2 - 2D \sum_{i \in [m]} \ell_i + m D^2 = mD^2 - 2mD^2 + mD^2 = 0.
		\end{align*}
		This directly yields that $\ell_i = D$ for each $i \in [m]$, showing that exactly $D$ vertices are assigned to each bin.
		Together with the previous observation that for each star either all or no vertices are assigned to a bin, this shows that the input instance of \problemfont{Exact Bin Packing} is a yes-instance.
	\end{claimproof}
	Since \problemfont{Exact Bin Packing} is strongly \NP{}-hard, we can assume the integers of the input are given in unary, and hence the reduction runs in polynomial-time.
	This concludes the proof.
\end{proof}

\section*{AI Methodology}

We used GPT-5.6 Sol to obtain a proof that computing the stack number of a graph is \FPT\ parameterized by treedepth.
Moreover, the same system was used to consider potential generalizations, eventually leading us to look into \snpmodelchecking{}.
The proofs of the main result for \snpmodelchecking\ follow but go beyond the proof for stack number provided by GPT, and these proofs were written by the authors directly.

GPT-5.6 Sol was used to look for further examples of problems expressible by \SNP{}, leading us to include \problemfont{Boxicity} as an additional example application.
GPT-5.6 Sol and GPT-6 Astra were used to look for an example of a problem that can witness the non-extendability of our main result, leading to \cref{thm:eso_not_tractable_treedepth}.
The result in \cref{thm:max_snp_not_tractable_treedepth} was obtained using GPT-6 Astra.

We also used various large language model tools (including GPT-5.6 Sol, GPT-6 Astra, Claude Opus~5, Opus~5.5, and Fable 5.1) to proofread the paper, to help with writing, and to search the literature. They also assisted in the creation of our figures.

The authors checked the correctness of the content, investigated relevant prior work, and take full responsibility for the article.

\subparagraph*{Acknowledgments.}
T.\ Depian, R.\ Ganian, M.\ Gr\"undel, and  S.\ Wietheger acknowledge support from the Austrian Science Fund (FWF Projects 10.55776/Y1329 and 10.55776/COE12) and the Vienna Science and Technology Fund (WWTF Project 10.47379/ICT22029).

\bibliography{references}

\end{document}